\documentclass{statsoc}
\usepackage[top=1in, bottom=1in, left=1in, right=1in]{geometry}
\usepackage{color}
\usepackage{enumerate}
\usepackage{subcaption}

\usepackage{amsmath}
\usepackage{amssymb}
\usepackage[all]{xy}
\usepackage{graphicx}
\usepackage{float}
\usepackage{xr}
\makeatletter
\newcommand*{\addFileDependency}[1]{
  \typeout{(#1)}
  \@addtofilelist{#1}
  \IfFileExists{#1}{}{\typeout{No file #1.}}
}
\makeatother

\usepackage{framed}
\usepackage{fancyhdr}

\usepackage{amsmath}
\usepackage{amsfonts}
\usepackage{algorithm}
\usepackage{algorithmic}
\usepackage{graphicx}
\usepackage{appendix}
\usepackage{times}
\usepackage{graphicx} 

\usepackage{natbib}

\usepackage{algorithm}
\usepackage{algorithmic}
\usepackage{comment}

\usepackage{framed}
\usepackage{fancyhdr}

\def\fatnorm#1{|\kern-.2ex|\kern-.2ex| #1 |\kern-.2ex|\kern-.2ex|}
\newcommand{\twonorm}[1]{\left\lVert#1\right\rVert_2}

\newcommand{\fnorm}[1]{\left\lVert#1\right\rVert_{F}}
\newcommand{\norm}[1]{\left\lVert#1\right\rVert}
\newcommand{\abs}[1]{\left\lvert#1\right\rvert}

\newcommand{\T}{\mathcal{T}}

\newcommand{\cov}{\textsf{Cov}}

\newcommand{\off}{{\rm off}}

\newcommand{\half}{\ensuremath{\frac{1}{2}}}
\newcommand{\inv}[1]{\frac{1}{#1}}
\def\conv{\mathop{\text{\rm conv}\kern.2ex}}

\newcommand{\ip}[1]{\;\langle{\,#1\,}\rangle\;}

\newcommand{\onenorm}[1]{\ensuremath{\left|#1\right|_1}}
\newcommand{\offone}[1]{\ensuremath{\left|#1\right|_{1,\off}}}

\newcommand{\silent}[1]{}
\newcommand{\mvec}[1]{\rm{vec}\left\{\,#1\,\right\}}

\newcommand{\ve}{\varepsilon}

\def\qed{\hskip1pt $\;\;\scriptstyle\Box$}
\def\Ber{\mathop{\text{Bernoulli}\kern.2ex}}
\def\supp{\mathop{\text{supp}\kern.2ex}}
\def\corr{\mathop{\text{corr}\kern.2ex}}
\def\prec{\mathop{\text{precision}\kern.2ex}}
\def\recall{\mathop{\text{recall}\kern.2ex}}
\def\cov{\mathop{\text{Cov}\kern.2ex}}
\def\mnorm{\mathcal{N}_{f,m}\kern.2ex}
\def\var{\mathop{\text{Var}\kern.2ex}}
\def\ess{\mathop{\text{ess}\kern.2ex}}
\def\dom{\mathop{\text{dom}\kern.2ex}}
\def\lin{\mathop{\text{lin}\kern.2ex}}

\newcommand{\func}[1]{\ensuremath{\mathrm{#1}}}
\newcommand{\diag}{\func{diag}}
\newcommand{\offd}{\func{offd}}

\let\hat\widehat
\let\tilde\widetilde

\newcommand{\tr}{{\rm tr}}

\newcommand{\sign}{{\rm sign}}

\def\E{{\mathbb E}}

\def\supp{\mathop{\text{\rm supp}\kern.2ex}}
\def\argmin{\mathop{\text{arg\,min}\kern.2ex}}

\newcommand{\prob}[1]{\ensuremath{\mathbb P}\left(#1\right)}

\newcommand{\beq}{\begin{equation}}
\newcommand{\eeq}{\end{equation}}
\newcommand{\ben}{\begin{eqnarray}}
\newcommand{\een}{\end{eqnarray}}
\newcommand{\bnum}{\begin{enumerate}}
\newcommand{\enum}{\end{enumerate}}
\newcommand{\bit}{\begin{itemize}}
\newcommand{\eit}{\end{itemize}}
\newcommand{\bens}{\begin{eqnarray*}}
\newcommand{\eens}{\end{eqnarray*}}

\newcommand{\Sc}{\ensuremath{S^c}}

\newcommand{\A}{\mathcal{A}}

\newenvironment{proofof}[1]{\hspace*{20pt}{\it Proof}{ of #1}.\hskip10pt}{\qed\vskip5pt}
\newenvironment{proofof2}{\hskip10pt}{\qed\vskip5pt}

\newtheorem{theorem}{Theorem}
\newtheorem{corollary}{Corollary}
\newtheorem{lemma}[theorem]{Lemma}
\newtheorem{proposition}[theorem]{Proposition}
\newtheorem{definition}{Definition}

\graphicspath{FiguresTalk/}
 \graphicspath{{Figures/}}
 
 \title{Tensor Graphical Lasso (TeraLasso)}
\author[Greenewald {\it et al.}]{Kristjan Greenewald}
\address{IBM Research,
Cambridge,
         USA.}
\author{Shuheng Zhou}
\address{University of California,
Riverside,
         USA.}
\author[Greenewald, Zhou, Hero]{Alfred Hero III}
\address{University of Michigan,
Ann Arbor,
         USA.}

\begin{document} 

\begin{abstract}

This paper introduces a multi-way tensor generalization of the Bigraphical Lasso (BiGLasso), which uses a two-way sparse Kronecker-sum multivariate-normal model for the precision matrix to parsimoniously model conditional dependence relationships of matrix-variate data based on the Cartesian product of graphs. We call this generalization the {\bf Te}nsor g{\bf ra}phical Lasso (TeraLasso). We demonstrate using theory and examples that
the TeraLasso model can be accurately and scalably estimated from very limited data samples of high dimensional variables with multiway coordinates
such as space, time and replicates.
Statistical consistency and statistical rates of convergence are established for both the BiGLasso and TeraLasso
estimators of the precision matrix and estimators of its support (non-sparsity) set, respectively. We
propose a scalable composite gradient descent algorithm and analyze the computational convergence rate, showing that the
composite gradient descent algorithm is guaranteed to converge at a
geometric rate to the global minimizer of the TeraLasso objective function. Finally, we illustrate the TeraLasso using both simulation and experimental data from a meteorological dataset, showing that we can
accurately estimate precision matrices and recover meaningful conditional dependency graphs from high dimensional complex datasets.

\end{abstract}

\section{Introduction}
\label{Sec:Not}

%
%
%
%
%
%


The increasing availability of matrix and tensor-valued data with complex dependencies has fed the fields of statistics and machine learning. Examples of tensor-valued data include medical and radar imaging modalities, spatial and meteorological data collected from sensor networks and weather stations over time, and biological, neuroscience and spatial gene expression data aggregated over trials and time points.
Learning useful structures from these large scale, complex and high-dimensional data in the low sample regime is an important task in statistical machine learning, biology and signal processing.

As the precision matrix (inverse covariance matrix) encodes interactions and, for tensor-valued Gaussian distributions, conditional independence relationships
between and among variables, multivariate statistical models, such as the matrix normal model (\cite{dawid1981some}), have been proposed for estimation of these matrices. However, the number of parameters of the precision matrix of a $K$-way data tensor $X \in \mathbb{R}^{d_1 \times \dots \times d_K}$ grows as $\prod_{i=1}^K d_i^2$. Therefore in high dimensions unstructured precision matrix estimation is impractical, requiring very large sample sizes. 
Undirected graphs are often used to describe high dimensional distributions. Under sparsity conditions, the graph can be estimated using $\ell$1-penalization methods, such as the graphical Lasso (GLasso) \citep{FHT07} and multiple (nodewise) regressions \citep{meinshausen2006high}. Under suitable conditions, such approaches yield consistent (and sparse) estimation in terms of graphical structure and fast convergence rates with respect to the operator and Frobenius norm for the covariance matrix and its inverse. However, many of the statistical models that have been considered still tended to be overly simplistic and not fully reflective of reality. For example, in neuroscience one must take into account temporal correlations as well as spatial correlations, which reflect the connectivity formed by the neural pathways. Yet, this line of high dimensional statistical literature
mentioned above has primarily focused on estimating linear or graphical models with i.i.d. samples. In the case of graphical models, the data matrix is usually assumed to have independent rows or columns that follow the same distribution. The independence assumptions substantially simplify mathematical derivations but they tend to be very restrictive. For instance, the cortical circuits can change over time due to activities such as motor learning, attention or visual stimulation. This data typically has a complex structure that is organized by the experiment'€™s design, with one or more experimental factors varying according to a predefined pattern.

On the theoretical and methodological front, recent work demonstrated another regime where further reductions in the sample size are possible under additional structural assumptions on the conditional dependency graphs which arise naturally in the above mentioned contexts when handling data with complex dependencies. For example, the matrix-normal model as studied in \cite{tsiligkaridis2013convergence} and \cite{zhou2014gemini} restricts the topology of the graph to tensor product graphs where the precision matrix corresponds a Kronecker product representation. Moreover, \citep{zhou2014gemini} showed that one can estimate the covariance and inverse covariance matrices well using only one instance from the matrix-variate normal distribution. Along the same lines, the Bigraphical Lasso framework was proposed to parsimoniously model conditional dependence relationships of matrix-variate data based on the Cartesian product of graphs \citep{kalaitzis2013bigraphical} as opposed to the direct product graphs of the matrix-normal models above. These models naturally generalize to multilinear settings with more than two axes of structure as demonstrated in the present work. The present work addresses the problem of sparse modeling of a structured precision matrix for tensor-valued data; more precisely, we aim to estimate the structure and parameters for a class of Gaussian graphical models by restricting the topology to the class of Cartesian product graphs, with precision matrices represented by a Kronecker sum for data with complex dependencies.

Toward these goals, we will introduce the tensor graphical Lasso (TeraLasso) procedure for estimating sparse $K$-way decomposable precision matrices. We will show that our concentration of measure analysis enables a significant reduction in the sample size requirement in order to estimate parameters and the associated conditional dependence graphs along different coordinates such as space, time and experimental conditions. We establish consistency for both the Bigraphical Lasso and Tensor graphical Lasso estimators and obtain optimal rates of convergence in the operator and Frobenius norm for estimating the associated precision matrix, and for structure recovery. Finally, we demonstrate using simulations and real data that the Kronecker sum precision model has excellent potential for improving computational scalability, structural interpretation, and its applications to classification, prediction, and visualization for complex data analysis.

A philosophical motivation of TeraLasso's Kronecker sum (Cartesian graph) model is that it achieves the maximum entropy among all models for which the tensor component projections of the covariance matrix are fixed, see Section \ref{sec:Models}. A compelling justification for the proposed Kronecker sum model for the precision matrix is that similar models have been successfully used in other fields, including regularization of multivariate splines, design of physical networks, and decomposition of solutions of partial differential equations governing many physical processes. Additional discussion of these practical motivations for the model is in Section \ref{sec:motivation} below.

\subsection{The multi-way Kronecker sum precision matrix model}
We follow the notation and terminology of \cite{kolda2009tensor} for modeling tensor-valued data arrays. 
Define the vector of component dimensions $\mathbf{p} = [d_1,\dots,d_K]$ and let $p$ denote the product of these dimensions 
\begin{align*}
p = \prod\nolimits_{k = 1}^K d_k \quad \mathrm{and} \quad m_k = \prod\nolimits_{i \neq k} d_i = \frac{p}{d_k}.
\end{align*}
To simplify the multiway Kronecker notation, we define
\begin{align*}
{I}_{[d_{k:\ell}]} = \underbrace{{I}_{d_{k}} \otimes \dots \otimes {I}_{d_{\ell}}}_{\ell-k+1 \; \mathrm{factors}}
\end{align*}
where $\otimes$ denotes the Kronecker (direct) product and $\ell \geq k$. Using this notation, the $K$-way Kronecker sum of matrix components $\{\Psi_k\}_{k=1}^K$ can be written as
\begin{align}
\label{Eq:model}
\Psi_1 \oplus \dots \oplus \Psi_K = \sum_{k =1}^K I_{[d_{1:k-1}]} \otimes \Psi_k \otimes I_{[d_{k+1:K}]}.
\end{align}
In the special case of $K=2$ this Kronecker sum representation reduces to the more familiar $\Psi_1\oplus \Psi_2=\Psi_1 \otimes I_{d_1}+I_{d_2} \otimes \Psi_2$. 
The vectorization of a $K$-way tensor $X$ is denoted as $\mathrm{vec}(X)$ and is defined as in \cite{kolda2009tensor}. Likewise, we  
define the transpose of a $K$-way tensor $X^T \in \mathbb{R}^{d_K\times \dots \times d_1}$ analogously to the matrix transpose, i.e. $[X^T]_{i_1,\dots,i_K} = X_{i_K,\dots, i_1}$. 

When the precision matrix $\Omega$ has a decomposition of the form (\ref{Eq:model}) the Kronecker sum components $\{\Psi_k\}_{k=1}^K$ are sparse, and the $K$-way data $X$ has a multivariate Gaussian distribution, the sparsity pattern of $\Psi_k$ corresponds to a conditional independence graph across the $k$-th dimension of the data.

Figure \ref{fig:cartAR} illustrates the Kronecker sum model proposed in \eqref{Eq:model} for $K =3$ and $d_k=4$. Specifically, $\Psi_k, k = 1, 2, 3$ are identical $4\times 4$ tridiagonal precision matrices corresponding to a one dimensional autoregressive-1 (AR-1) process.  
In the Figure the precision matrix $\Omega=\Psi_1\oplus\Psi_2\oplus\Psi_3$ is shown on the left and covariance $\Sigma = \Omega^{-1}$ on the right.
\begin{figure}[h]
\centering
\includegraphics[width=5in]{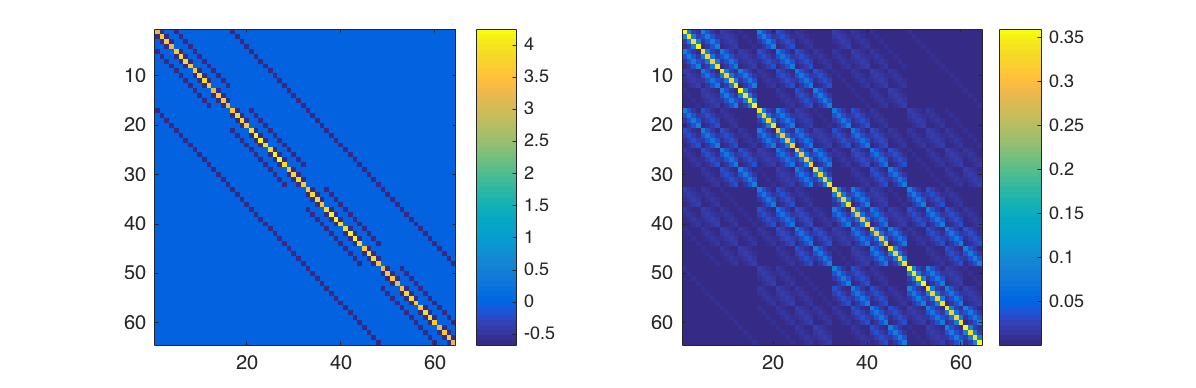}
\caption{Illustration of the Kronecker sum model for a tensor valued AR(1) process. Left: Sparse $4\times 4 \times 4$  precision matrix $\Omega = \Psi_1 \oplus \Psi_2 \oplus \Psi_3$, where $\Psi_k$ are identical tridiagonal  precision matrices corresponding to one dimensional AR(1) models. Right: Covariance matrix $\Sigma = \Omega^{-1}$. Unlike the Kronecker product precision model, the nested block structure in $\Sigma$ is not representable by a product of component factors. 
}
\label{fig:cartAR}
\end{figure}
The entries of each $\Psi_k$ are replicated $m_k=16$ times across $\Omega$ for each $k$. 
This regular structure permits the aggregation of corresponding entries in the sample covariance matrix, resulting in variance reduction in estimating $\Omega$. 
This Kronecker sum gives $\Omega$ a nonseparable and interlocking repeating block structure in the covariance matrix. 



We propose the following sparse Kronecker sum estimator of the precision matrix $\Omega$ in \eqref{Eq:model}, which we call the Tensor Graphical Lasso (TeraLasso). 
The TeraLasso minimizes the negative $\ell_1$-penalized Gaussian log-likelihood function over the domain $\mathcal{K}_{\mathbf{p}}^{\sharp}$ of precision matrices $\Omega$ having Kronecker sum form: 
\begin{align}
\label{eq::lossfunc}
&\hat{\Omega}
=\arg\min_{\Omega \in \mathcal{K}_{\mathbf{p}}^\sharp, \|\Omega\|_2 \leq \kappa}
\left\{-\log | \Omega| + \langle \hat{S}, \Omega\rangle + \sum_{k=1}^K  m_k\sum_{i\neq j}g_{\rho_k}({[{\Psi}_k]_{ij}})
\right\}\\
\label{eq::gram}
&\text{ where }  
\hat{S}  = \inv{n} \sum_{i= 1}^n \mathrm{vec}(X_i^T) \mathrm{vec}(X_i^T)^T,
\end{align}
$g_{\rho}(t)$ is a sparsity-inducing regularization function parameterized by a regularization parameter $\rho$, and 
\ben
\label{eq::kppintro}
\mathcal{K}^\sharp_{\mathbf{p}} &= & 
\{ {A} \succeq 0: \exists \: {B}_k \in \mathbb{R}^{d_k \times d_k} \: \mathrm{s.t.} \: {A} = {B}_1 \oplus \dots \oplus {B}_K   \}
\een
is the set of positive semidefinite matrices that are decomposable into a Kronecker sum of fixed factor dimensions $d_1, \ldots, d_K$.
In this paper we consider $(\mu,\gamma)$-amenable regularizers $g_\rho$ \citep{LOH}. The norm constraint $\|\Omega\|_2 \leq \kappa$ is required for the solution to be well defined when $g_\rho$ is not a convex penalty. These penalties includes nonconvex regularizers such as SCAD and MCP, as well as the traditional $\ell 1$ regularizer $g_\rho(t) = \rho | t|$.

Observe that sparsity in the off diagonal elements of $\Psi_k$ directly creates sparsity in $\Omega$. As in the graphical Lasso, incorporating an $\ell_1$-penalty over entries of $\Omega$ with the tensor-valued Gaussian or matrix-normal (pseudo)-loglikelihood promotes a sparse graphical structure in $\Omega$;
see for example \citep{banerjeeGLASSO,yuanGLASSO,zhou2014gemini,zhou2011high}. In this work, we allow for the more general case of nonconvex regularization functions $g_\rho$ as considered in \cite{LOH}. While sometimes difficult to tune in practice, nonconvex regularization provides strong nonasymptotic guarantees on the elementwise estimation error of $\Omega$, implying strong, single sample support recovery guarantees when the smallest nonzero element of $\Omega$ is bounded from below. 


The contributions of this paper are as follows. The sparse multivariate-normal Bigraphical Lasso (BiGLasso) model is extended to the sparse tensor-variate ($K>2$) TeraLasso model, allowing the modeling of data with arbitrary tensor degree $K$.  
A new subgaussian concentration inequality (Corollary \ref{Cor:Chaos} in the supplement) is presented that gives rates of statistical convergence (Theorems 1-3) of 
the TeraLasso estimator as well as the BiGLasso estimator, when the sample size is low (even equal to one). TeraLasso's generalization of BiGlasso from 2-way to $K$-way decompositions is important as it expands the domain of application, allowing a data scientist to group variables into their natural domains along multiple tensor axes. For example, with a health data set that is collected over space, time, people and replicates, TeraLasso's 3-way tensor decomposition (time$\times$space$\times$people) can account for possible dependency structure between people, while a 2-way BiGLasso or KLasso approach decomposing over (time$\times$space) would unnecessarily enforce an assumption of independence between people. Alternately, BiGLasso or KLasso could group two axes together (e.g. (time$\times$space)$\times$people), however, this would create a large, unstructured factor whose estimation would require many more replicates than the 3-way decomposition that TeraLasso uses to give each axis its own factor. 

A highly scalable, first-order ISTA-based algorithm is proposed to minimize the TeraLasso objective function. We prove (Theorem \ref{thm:conv} in the supplement) that it converges to the global optimum with a geometric convergence rate,
and demonstrate its practical advantages on high dimensional problems. 
As compared to the alternating block coordinate descent algorithm proposed by \cite{kalaitzis2013bigraphical} for the BiGLasso, 
the proposed ISTA algorithm enjoys a per-iteration computational speedup over BiGLasso of order $\Theta(p)$.
Our numerical results show that the BiGLasso algorithm often requires 
many more iterations to converge than our ISTA method.
Numerical comparisons are presented demonstrating that TeraLasso significantly improves performance in small sample regimes. To demonstrate the application of TeraLasso to real world data we use it to estimate the precision matrix of spatio-temporal meteorological data collected by the National Center for Environmental Prediction (NCEP). Our results show that the TeraLasso precision matrix estimator degrades much more slowly than other estimators as one reduces the number of samples available to fit the model. The intuitive graphical structure, the robust eigenstructure and a maximum-entropy interpretation make
the TeraLasso model a compelling choice for modeling tensor data, 
much as the Bigraphical Lasso provides a meaningful alternative to the matrix-normal model.

\subsection{Relevant prior work}

The use of tensor product models for multiway data has a long history. In the statistical context, directly fitting a Kronecker product to multiway data yields a first order approximation corresponding to fitting the mean  \citep{kolda2009tensor} when the fitting criteria is the Frobenius norm of the residuals. Many such methods involve low-rank factor decompositions including: PARAFAC and CANDECOMP as in \cite{harshman1994parafac,faber2003recent}; Tucker decomposition-based methods such as \cite{tucker1966some} and \cite{hoff2016equivariant}; and hybrid methods such as \cite{johndrow2017tensor}. In contrast, second order methods have been used to approximate multiway structure of the covariance \citep{werner2008estimation,pouryazdian2016candecomp}. 
Series decomposition methods have been proposed for fitting the covariance matrix in Frobenius norm using sums of Kronecker products \citep{tsiliArxiv,greenewaldTSP,rudelson2015high,NIPS2017_7165}. 

Kronecker product approximations to the inverse covariance have fitted matrix normal models \citep{allen2010transposable} and sparse matrix normal models \citep{leng2012sparse,zhou2014gemini,tsiligkaridis2013convergence}. In contrast to the Kronecker sum model (\ref{Eq:model}) for the precision matrix $\Omega$, the $K$-way Kronecker product model is $\Omega=\Psi_1 \otimes \ldots \otimes \Psi_K$. The Kronecker product decomposition implies a separable property of the precision matrix across the $K$ data dimensions, which one might expect to become an increasingly restrictive condition as $K$ increases. 
In this paper we show that the proposed Kronecker sum model (\ref{Eq:model}) can be a worthwhile alternative representation. 

A two factor ($K=2$) sparse Kronecker sum model for the precision matrix $\Omega$ was introduced and studied in \cite{kalaitzis2013bigraphical}. The model was fitted to the sample covariance matrix using an iterative procedure called BiGlasso, which required the diagonal entries of $\Omega$ to be known. Conditions guaranteeing convergence were not provided. Here we extend the BiGlasso model to arbitrary $K\geq 2$ and unknown diagonal entries of $\Omega$, provide a faster converging optimization algorithm, and obtain strong convergence guarantees and bounds on the convergence rate for all $K$, including $K=2$. For completeness, we also obtain (Appendix \ref{app:B} of the supplement) bounds on the convergence rate for the known-diagonal setting of \cite{kalaitzis2013bigraphical}.


\begin{figure}[h]
\centering
\includegraphics[width=1in]{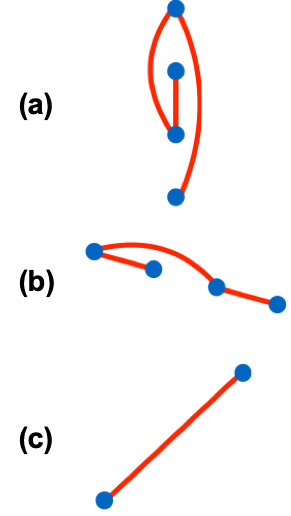}\hspace{6mm}
\includegraphics[width=3.4in]{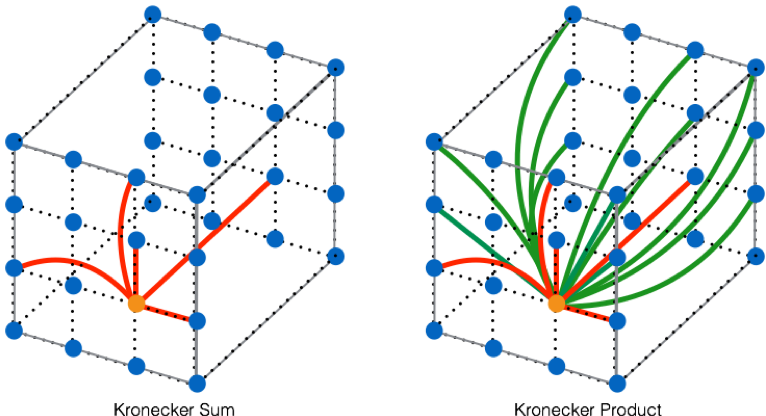}
\caption{Comparison of Kronecker sum (Cartesian product graph) at center and Kronecker product (direct product graph) at right. The products are formed from the component graphs shown in (a)-(b)-(c); the number of factors in the product graphs is $K=3$ and the dimensions are $d_1=d_2=4$, $d_3=2$, leading to product graphs with 32 nodes, arranged in a regular 3 dimensional grid in the figures at bottom. Only the edges emanating from the orange node are indicated (red and green edges). The Kronecker sum model has a total of 64 edges while the Kronecker product model is much less sparse, having a total of 184 edges. }
\label{fig:cube}
\end{figure}

\begin{figure}
\centering
\scriptsize{
\begin{tabular}{|p{.6in}|p{1.8in}|p{1.8in}|}
\hline
& \textbf{Multiway Kronecker Product} & \textbf{Multiway Kronecker Sum}  \\
\hline
\textbf{Covariance Model} &Precision matrix $\Omega$ is separable across $K$ tensor components. & Precision matrix is nonseparable across tensor components, motivated by maximum entropy considerations.\\\hline
\textbf{Graphical Model} & Graph is the direct product of the $K$ graph components. & Graph is the Cartesian product of the $K$ graph components.\\\hline
\textbf{Sparsity} & Number of edges in $\Omega$ grows as the \textbf{product} of the number of edges in each component.  & Number of edges in $\Omega$ grows as the \textbf{sum} of the number of edges in each component. \\\hline
\textbf{Graphical model interpretability} & Edges in sparse factors contribute to large numbers of edges multiplicatively. & Each edge in the sparse factors directly map to edges in the overall precision $\Omega$. Sparsity pattern follows Cartesian Markov-like network.\\\hline
\textbf{Inference} & Nonconvex (multilinear) maximum likelihood estimator, alternative estimators usually favored. & Maximum likelihood estimator is convex.\\\hline
\end{tabular}
}
\captionof{table}{Qualitative differences between multiway Kronecker sum (Teralasso) and multiway Kronecker product (BiGlasso) models for high dimensional precision matrix estimation.}
\label{Fig:SumVsProd}
\end{figure}

The qualitative differences between the Kronecker product and Kronecker sum models for the precision matrix can be better appreciated by considering the product graphs that are induced by them.  
For given sparse Kronecker factors $\Psi_1, \ldots, \Psi_K$, the Kronecker product model corresponds to the direct (tensor) product of the component graphs while the Kronecker sum model corresponds to the Cartesian product\footnote{The Cartesian product of two graphs $G_1 = (V_1,E_1)$ and $G_2 = (V_2,E_2)$ is a graph with vertices being the Cartesian product of $V_1$ and $V_2$, and with edges such that node $(u,u')$ is adjacent to $(v,v')$ if and only if either $u = v$ and $u'$ is adjacent to $v'$ in $G_2$, or $u' = v'$ and $u$ is adjacent to $v$ in $G_1$.}  of these components \citep{hammack2011handbook}.  The direct product graph and Cartesian product graph differ greatly; the former has a number of edges equal to $ \frac{1}{2}\prod_{k=1}^K(2|E_k| + |V_k|) - \prod_{k=1}^K |V_k|$,  while the latter has a number of edges equal to $\sum_{k=1}^K |E_k| \prod_{i\neq k} |V_i|$, where $V_i, E_i$ denote the node and edge sets of the $i$-th component graph\footnote{The notation $|V_i|=d_i$ denotes the row dimension of $\Psi_i$ and $|E_i|$ denotes the number of non-zero upper triangular entries of $\Psi_i$}. To illustrate, if the number of non-zero entries of $\Psi_k$ is $c d_i$ for some $c$, the number of edges induced in the direct product graph by inserting a single new edge into the first component graph is equal to $\frac{1}{2}(2c+1)^{K} (p/d_1) - p$, where we recall that $p=\prod_{k=1}^K d_i$ is the number of covariates (rows of $\Omega$). On the other hand, for the Cartesian product graph it is only $p/d_1$ regardless of $c$. Hence, as $c$ and $K$ increase, using the Kronecker product model a single edge in $\Psi_1$ can create a proliferation of edges while the number of new edges in the Kronecker sum model is fixed, independent of $K$. A concrete example of these differences is illustrated in Figure \ref{fig:cube}. The qualitative differences between the Kronecker product and Kronecker sum models for the precision matrix are summarized in Table \ref{Fig:SumVsProd}.  

\subsection{Rationale for the proposed multiway Kronecker sum model}
\label{sec:motivation}

This paper develops a scalable, fast and accurate estimation procedure, the TeraLasso, for multiway  precision matrices $\Omega$ using higher order Kronecker sum models. To justify the practical utility of the TeraLasso we illustrate it on a spatio-temporal meteorological dataset. We have also applied it to other applications not presented here. While comprehensive validation of the model on a larger corpus of real data is beyond the scope of this paper, there is ample evidence that the model will have many statistical applications. We base this assessment on the wide use of Kronecker sum models, equivalently Cartesian product graph models, in biology, physics, social sciences, and network engineering, among other fields \citep{imrich2008topics,van2000ubiquitous}.  In particular the Kronecker sum arises in solving the celebrated Sylvester equation for a matrix $X$ which, for $K=2$, takes the form $XA+B X=N$. The Sylvester equation can be solved by expressing the equation in vectorized form as $A \oplus B \; {\mathrm{vec}}(X)= {\mathrm{vec}}(N)$ (for arbitrary $K$ this becomes the tensor Sylvester equation $(A_1 \oplus \dots \oplus A_K)\mathrm{vec}(X) = \mathrm{vec}(N)$), but this is often impractical in high dimension. Such equations result from the discretization of separable $K$-dimensional PDEs with tensorized finite elements \citep{grasedyck2004existence,kressner2010krylov,beckermann2013error,shi2013backward,ellner1986new}. As a result Kronecker sums come in many areas of applied math, including, beam propagation physics (\cite{andrianov1997matrix}); control theory 
\citep{luenberger1966observers,chapman2014controllability}; fluid dynamics 
\citep{dorr1970direct}; and spatio-temporal neural processes \citep{schmitt2001numerical}.

Closer to home, the Kronecker sum model arises in multivariate spline data analysis, e.g. as applied to harmonic analysis on graphs (\cite{kotzagiannidis2017splines}). More recently, \cite{fey2018splinecnn} has proposed tensor B-splines defined over a Cartesian product basis for geometric Convolutional Neural Networks (CNN). Kronecker sums have been proposed as precision matrices for weighting the quadratic regularizer in smoothed multivariate spline regression. In particular, \cite{wood2006low} observed that, as compared to the Kronecker product, the Kronecker sum reduces the coupling between the axes when used as a spline smoothing penalty for generalized additive mixed model regression. This observation motivated \cite{wood2006low} and \cite{eilers2003multivariate} to use the inverse of a Kronecker sum matrix as a penalty, or prior, for smoothing $K$-dimensional regressions (see also work by \cite{lee2011p} and \cite{wood2016}). This approach has been applied to spatio-temporal forest health modeling (for which $K=3$) \citep{augustin2009modeling}, brain development modeling \citep{holland2014structural}, and analysis of the impact of climate and weather on spatio-temporal patterns of beetle populations \citep{preisler2012climate}, among other applications. 
In these spline regression problems the Kronecker sum appears as a precision matrix parameterizing a Gaussian prior on the spline coefficient vector $\beta$, where the prior is of the form $p(\beta) \propto \exp(-\beta^T(\lambda_1 S_1 \oplus \dots \oplus \lambda_K S_K)\beta/2)$. Here,  $\lambda_i$ are regularization coefficients and $S_i$ are coordinate-wise smoothing matrices, $i=1, \ldots, K$. 

Instead of using the Kronecker sum to model the {\em a priori} precision matrix of a set of spline parameters, this paper proposes the Kronecker sum as a model for the precision matrix of the multiway data in the likelihood function, where the data matrix $X$ takes the place of the spline coefficient vector $\beta$. The stated advantages of the Kronecker sum model for the spline regression setting \citep{wood2006low} can be expected to carry over to the precision matrix estimation setting of TeraLasso.  In particular, like the spline regression prior, the TeraLasso smooths each axis separately, while summing over the others, thereby reducing coupling between the tensor axes as compared to the Kronecker product. For data that has structure similar to that imposed by \citep{wood2006low} on the spline regression coefficients this should result in a more accurate fit. Indeed, if a population of regression spline problems was available, in principle one could apply the TeraLasso to estimating the best precision matrix of the spline coefficients that would minimize the population-averaged fitting error.

\textbf{Outline}. The remainder of the paper is organized as follows. 
We introduce notation and some preliminary results in Section~\ref{sec:: sec::notation}, and our proposed TeraLasso model in Section \ref{sec:Models}.
High dimensional consistency results are presented in Section \ref{sec:highdim}, first with convex $\ell 1$ regularizers and then with non-convex sparsity regularizers.  
The first order ISTA optimization algorithm is described in Section \ref{sec:alg}, and conditions are specified for which the algorithm converges geometrically to the global optimum. Finally, Sections \ref{sec:exp} and \ref{sec:real} illustrate the proposed TeraLasso estimator on simulated and real data, with Section \ref{sec:conc} concluding the paper.
We place all technical proofs in the supplementary material, along with additional experiments and further exploration of the properties and implications of the Kronecker sum subspace $\mathcal{K}_{\mathbf{p}}$ and the associated identifiable parameterization. 

\section{Notation and Preliminaries}
\label{sec:: sec::notation}
We use upper case letters, e.g. $A$ for matrices and tensors, bold lower case $\mathbf{a}$ for vectors, and denote the $(i,j)$ element of a matrix $A$ as $A_{ij}$ and the $(i_1,i_2,\dots,i_K)$ element of a tensor $A$ as $A_{i_1,i_2, \dots, i_K}$. 
Fibers are the higher-order analogue of matrix rows and columns. A fiber of a tensor
is obtained by fixing every index but one, the
mode-$k$ fiber of tensor $X$ is denoted as the column vector
$X_{i_1,\dots,i_{k-1},:,i_{k+1},\dots, i_K}$.
Following definition by \cite{kolda2009tensor}, tensor unfolding or matricization of $X$ along
the $k$th-mode is denoted as $X_{(k)}$, formed by arranging the mode-$k$ fibers
as columns of the resulting matrix of dimension $d_k \times m_k$. 
The column ordering is not important so long as it is consistent.  

For a vector $y = (y_1, \ldots, y_p)$ in $\mathbb{R}^p$, denote by 
$\twonorm{y} = \sqrt{\sum_{j} y_j^2}$ the Euclidean norm of $y$.
The operator and Frobenius norms of a matrix $A$ are denoted as
$\twonorm{A}$ and $\fnorm{A}$ respectively; the notation $\mathrm{vec}(A)$ denotes the vectorization of the matrix $A$; $\|A\|_\infty$ denotes the matrix infinity norm and $\|A\|_{\max} = \max_{ij} |A_{ij}|$ denotes the max norm.
The determinant is denoted as $\abs{A}$. We use the inner product $\ip{ A, B} = \tr(A^T B)$ throughout.
Define the set of $p\times p$ matrices with Kronecker sum structure of fixed dimensions $d_1, \ldots, d_K$:
\begin{align}
\label{eq:kpp}
\mathcal{K}_{\mathbf{p}}
 &= \{ {A} \in \mathbb{R}^{p\times p}: \exists \: {B}_k \in \mathbb{R}^{d_k \times d_k} \: \mathrm{s.t.} \: {A} = {B}_1 \oplus \dots \oplus {B}_K   \} 
\end{align}
where the set of matrices defined in \eqref{eq::kppintro} is obtained by restricting $\mathcal{K}_{\mathbf{p}}$ to the positive cone, i.e., 
\begin{align*}
\mathcal{K}_{\mathbf{p}}^\sharp &= \{A \succeq 0 | A \in \mathcal{K}_{\mathbf{p}}\}.
\end{align*} 
Note that the set $\mathcal{K}_{\mathbf{p}}$ \eqref{eq:kpp} is linearly spanned by
the $K$ components, since there are no nonlinear interactions between any of the parameters.
Thus $\mathcal{K}_{\mathbf{p}}$ is a linear subspace of $\mathbb{R}^{p \times p}$, and we can define a unique projection operator onto $\mathcal{K}_{\mathbf{p}}$:
\[
\mathrm{Proj}_{\mathcal{K}_{\mathbf{p}}}(A) = \arg \min_{M \in \mathcal{K}_{\mathbf{p}}} \|A - M\|_F^2.
\]
A closed-form expression for $\mathrm{Proj}_{\mathcal{K}_{\mathbf{p}}}(A)$ is given in Section \ref{sec:prjjj} of the supplementary material. Note that the dimensionality of the $\mathcal{K}_{\mathbf{p}}$ subspace is $1 - K + \sum_{k=1}^K d_k^2$, which is often significantly smaller than the ambient dimension $p^2 = \prod_{k=1}^K d_k^2$.   

\noindent\textbf{Parameterization of $\mathcal{K}_{\mathbf{p}}$ by $\Psi_k$.} Note that ${\Omega}= \Psi_1 \oplus \dots \oplus \Psi_K$
 does not uniquely determine $\{{\Psi}_k\}_{k=1}^K$, i.e., without further constraints the Kronecker sum parameterization is not
fully identifiable. 
It is easy to verify, however, that both
$
\offd(\Psi_k)\:\mathrm{and} \: \mathrm{diag}(\Omega)
$
are identifiable, 
where we define the notation $\offd(M) = {M} - \mathrm{diag}({M})$. We can then write the identifiable decomposition
\begin{align}\label{eq:iddd}
\hat{\Omega}
&=\mathrm{diag}(\hat{\Omega}) + \offd(\hat{\Psi}_1) \oplus \dots \oplus \offd(\hat{\Psi}_K),
\end{align}
and correspondingly $\Omega_0=\mathrm{diag}(\Omega_0) + \offd({\Psi}_{0,1}) \oplus \dots \oplus \offd({\Psi}_{0,K})$. Note that while the offdiagonal factors can take on any values, $\mathrm{diag}(\Omega_0)$ is not completely free (for a fully orthogonal parameterization see Section \ref{app:DO} of the supplement).

\noindent\textbf{Interpretation of correlation coefficients.}
The quantities $\frac{[\Psi_k]_{ij}}{\sqrt{[\Psi_k]_{ii}[\Psi_k]_{jj}}}$ do not by themselves correspond to correlation coefficients. Due to the repeating structure of the Kronecker sum each element $[\Psi_k]_{ij}$ will appear in $m_k$ distinct $d_k \times d_k$ symmetric subblocks of $\Omega$, and in each ($\ell$th) subblock it will have a correlation coefficient uniquely defined for that subblock:
\[
\rho_{k,ij, \ell} = \frac{[\Psi_k]_{ij}}{\sqrt{([\Psi_k]_{ii} + c_\ell/d_k)([\Psi_k]_{jj} + c_\ell/d_k)}}
\]
where $c_\ell = \tr(\ell\mathrm{th \: subblock \: of \:} \Omega) - \tr(\Psi_k)$. The overall correlation structure is preserved across the $m_k$ blocks, simply the strength of the correlations are modulated by the contributions of the other $K-1$ additive factors in the block.\footnote{Recall that the $\Psi_k$ need not be positive definite and $c_\ell$ need not be $> 0$.}

\section{Models and Methods}
\label{sec:Models}

Let $X_1, \ldots, X_n$ be $n$ independent realizations of the $K$-way tensor $X$.
Define $\mathbf{x}_i = \mvec{{X}_i^{T}}$ for all $i=1, \ldots, n$.
Define ${\hat{S}} = \frac{1}{n} \sum_{i = 1}^n \mathbf{x}_i \mathbf{x}_i^T$
as the sample covariance. 
The mode-$k$ Gram matrix $S_k$ and factor-wise covariance $\Sigma^{(k)} = \mathbb{E}[S_k]$ are given by
\begin{align*}
S_k &= \frac{1}{n m_k} \sum_{i=1}^n X_{i,(k)} X_{i,(k)}^T 
\quad \mathrm{and}\quad
{\Sigma}^{(k)} =  \frac{1}{m_k}  \mathbb{E}[X_{(k)} X_{(k)}^T],\quad k = 1,\dots,K,
\end{align*}
noting that the elements of these matrices are effectively inner products between $(K-1)$ order tensors.
$S_k$ is the sample covariance of the data unfolded across the $k$th tensor axis, while $\Sigma^{(k)}$ denotes the population covariance matrix along the same axis.
These Gram matrices $S_k$ can
be represented as elementwise aggregations over entries in the full sample covariance~\eqref{eq::gram}, with locations indexed by $\Psi_{k, i,j}$ as:
\begin{eqnarray}
\label{eq:skk}
[S_k]_{ij} & = & \frac{1}{m_k}\ip{\hat{S}, I_{[d_{1:k-1}]} \otimes
\mathbf{e}_i \mathbf{e}_j^T \otimes I_{[d_{k+1:K}]}}.
\end{eqnarray}
In tensor covariance modeling when the dimension $p$ is much larger than the number of samples $n$, the Gram
matrices $S_k$ are often used to model the rows and columns separately, notably in
the matrix-variate estimation methods of \cite{zhou2014gemini} and \cite{kalaitzis2013bigraphical}. 
Observe that the TeraLasso estimator~\eqref{eq::lossfunc} of the precision matrix can be expressed as
\ben
\label{eq:nonconvobv}
\hat{\Omega} 
 =\arg\min_{\Omega \in \mathcal{K}_{\mathbf{p}}^\sharp, \|\Omega\|_2 \leq \kappa}
\left\{ -\log | \Omega| + \sum_{k=1}^K m_k \left(\langle S_k, \Psi_k \rangle + \sum_{i\neq j}g_{\rho_k}({[{\Psi}_k]_{ij}}) \right)\right\}
\een
where $\mathcal{K}_{\mathbf{p}}^\sharp$ is the set of positive semidefinite Kronecker sum matrices~\eqref{eq::kppintro}.

Ignoring regularization, the objective function in curly brackets can be written as $-\log p(\hat{S}|\Omega)$ where $p(\hat{S}|\Omega)= \alpha_\Omega \prod_{k=1}^K p(S_k|\Psi_k)$ and $p(S_k|\Psi_k) =  \exp\left(- \langle m_k S_k, \Psi_k\rangle \right) $, with $\alpha_\Omega$ a normalizing constant. 
The non-negativity of the Kullback-Liebler divergence $\int p(S|\Omega)\log \left(\frac{p(S|\Omega)}{\alpha_\Omega \prod_{k=1}^K p(S_k|\Psi_k)} \right)dS$ implies that the Kronecker sum model is a {\em maximum entropy} model, as previously pointed out for the case of $K=2$ by \cite{kalaitzis2013bigraphical}. Alternatively, Kronecker sum models can be characterized as regularizing the precision matrix estimation problem with a minimally informative prior over the set $\mathcal{K}_{\mathbf{p}}^{\sharp}$.


The class of Kronecker sum matrices is a highly structured, lower-dimensional subspace of $\mathbb{R}^{p\times p}$.
By definition of the Kronecker sum \eqref{Eq:model},
each entry of $\Psi_k$ appears in $m_k = p/d_k$ entries of $\Omega$. 
By imposing that the precision matrix have both Kronecker sum structure and sparse structure through the penalty $g_\rho$, TeraLasso is able to effectively regularize the precision estimation problem.

We assume the penalty $g_\rho$ is $(\mu,\gamma)$-amenable in the sense of \cite{LOH}.
\begin{definition}[$(\mu,\gamma)$ amenable regularizer]
A regularizer $g_\rho(t)$ is $(\mu,\gamma)$-amenable when $\mu \geq 0$ and $\gamma \in (0,\infty)$ if 
\begin{enumerate}
\item $g_\rho$ is symmetric around zero and $g_\rho(0) = 0$.
\item $g_\rho(t)$ and $g_\rho(t)/t$ are both nondecreasing on $\mathbb{R}^+$.
\item $g_\rho(t)$ is differentiable for all $t \neq 0$.
\item The function $g_\rho(t) + \frac{\mu}{2} t^2$ is convex.
\item $\lim_{t\rightarrow 0^+}g'_\rho (t) = \rho$.
\item $g'_\rho(t) = 0$ for all $t \geq \gamma \rho$. 
\end{enumerate}
\end{definition}
Note that the $\ell 1$ regularizer is $(0,\infty)$-amenable. Example nonconvex penalties in this class include the SCAD penalty \citep{fan2001variable} and the MCP penalty \citep{zhang2010nearly}, both defined in Appendix \ref{supp:regs} of the supplement.

{Observe that for nonzero $\mu$ (i.e. nonconvex $g_\rho$) the constraint on the spectral norm of $\Omega$ ($\|\Omega\|_2 \leq \kappa$) in the TeraLasso objective function \eqref{eq:nonconvobv} is necessary since without it a global minimum may not exist \citep{LOH}. }
For spectral norm constraint parameter set to $\kappa = \sqrt{2/\mu}$, we show (Lemma \ref{lem:ncconv} in the supplement) that \eqref{eq:nonconvobv} with $g_\rho$ $(\mu,\gamma)$-amenable is convex and has a unique global minimizer. For the $\ell 1$ penalty, the objective is always convex and $\kappa$ can be set to infinity.



\section{High Dimensional Consistency of the TeraLasso}
\label{sec:highdim}

%

Let $\mathbf{v} =  [v_{1},\dots,v_{p}]^T$ be an isotropic $\psi_2$-subgaussian random vector 
with independent entries $v_{j}$ satisfying $\E v_{j} = 0$, $1 = \E v_{j}^2 \le \norm{v_{j}}_{\psi_2} \leq K$. The $\psi_2$ condition on a scalar random variable $V$ is equivalent to
subgaussian decay of the tails of $V$, implying
$\prob{|V| >t} \leq 2 \exp(-t^2/c^2) \; \text{for all} \; t>0$.
The extension to random vectors is straightforward. Specifically, $\mathbf{x}$ is a subgaussian random vector with 
positive definite covariance $\Sigma \in \mathbb{R}^{p \times p}$ when
\begin{equation}
\label{eq::marginal}
\mathbf{x} = {\Sigma}^{1/2} \mathbf{v},
\end{equation}
where $\Sigma^{1/2}$ denotes a positive definite square root factor of $\Sigma$. We then call ${X}\in \mathbb{R}^{d_1 \times d_2 \times \dots \times d_K}$ to be an order-$K$ 
subgaussian random tensor with covariance $\Sigma$ when $\mathbf{x} = \mathrm{vec}({X}^{T})$ 
is a subgaussian random vector in $\mathbb{R}^p$ defined as in \eqref{eq::marginal}.


We assume the data $X_1, X_2, \ldots, X_n$ are independent and identically distributed subgaussian random tensors whose inverse covariance follows the Kronecker sum model \eqref{Eq:model}, namely, that 
 $\mathrm{vec}({X_i}^{T}) \sim \mathbf{x}$, 
where $\mathbf{x}$ is a subgaussian random vector in $\mathbb{R}^p$ as defined in \eqref{eq::marginal}. A special case of the subgaussian model is the Gaussian model, for which the zeros in the precision matrix define the conditional independencies among the variables $X_i$.
This conditional independence relation does not hold for the general subgaussian case, but nonetheless strong convergence of the TeraLasso precision matrix estimator is preserved. 

In addition to the subgaussian generative model given above, we make the following technical assumptions on the true model, guaranteeing sparsity in $\Omega$ and its eigenvalues being bounded away from zero and infinity.  
\begin{description}
\item[(A1)] Define the support set of the $k$th Kronecker sum component $\Psi_k$ of the precision matrix by 
$\mathcal{S}_k = \{(i,j): i \neq j, [\Psi_{k}]_{ij} \neq 0\}$ for $k = 1,\dots, K$. 
We assume $\mathcal{S}_k$ is sparse, i.e. $\mathrm{card}(\mathcal{S}_k) \leq s_k$. 

\item[(A2)]
The minimal eigenvalue satisfies $\phi_{\min}({\Omega}) = \sum_{k=1}^K \phi_{\min}({\Psi}_k) \geq\underline{k}_\Omega > 0$, and the maximum eigenvalue satisfies $\phi_{\max}({\Omega}) = \sum_{k=1}^K \phi_{\max}({\Psi}_k)  \leq\overline{k}_\Omega < \infty$.

\end{description}
Defining the support set of $\Omega$ as $\mathcal{S} = \{(i,j): i \neq j,\}$, (A1) implies $\mathrm{card}(\mathcal{S}) \leq s = \sum_{k=1}^K m_k s_k$.


\subsection{Regularization with $\ell 1$ penalty}
With $g_\rho(t) = \rho |t|$, the constraint on $\|\Omega\|_2$ is unnecessary, and \eqref{eq:nonconvobv} becomes
\ben
\label{Eq:objFun}
\hat{\Omega} 
& =&\arg\min_{\Omega \in \mathcal{K}_{\mathbf{p}}^\sharp}
\left\{ -\log | \Omega| + \sum_{k=1}^K m_k \left(\langle S_k, \Psi_k \rangle +  \rho_k  |{\Psi}_k|_{1,\off} \right)\right\}
\een
where
$\offone{\Psi_k} = \sum_{i\not=j} \abs{[\Psi_{k}]_{ij}}$ is the off diagonal $\ell_1$ norm.
The objective \eqref{Eq:objFun} is jointly convex, and its minimization over $\Omega \in \mathcal{K}_{\mathbf{p}}^\sharp$ has a unique solution (see Section \ref{App:Conv} of the supplement). 
We require an additional assumption
\begin{description}
\item[(A3)]
The sample size $n$ and the component dimensions $d_k$ satisfy the following condition:
\ben
n (\min_{k} m_k)^2 \ge C^2 \kappa(\Sigma_0)^4 (s+ p) (K+1)^2 \log p
\een
where $m_k = p/d_k$ and $\kappa(\Sigma_0) = \phi_{\max}(\Sigma_0)/\phi_{\min}(\Sigma_0)$
is the condition number of $\Sigma_0$.
\end{description}
Note this assumption holds for $n=1$ and sufficiently large $(\min_{k} m_k)^2 > O({p})$, which can hold for any $K >2$.
We obtain the following bounds on the Frobenius and operator norm error of the TeraLasso estimator \eqref{Eq:objFun}. 
The constants ($c,C_1, C_2, C_3$) are given in the proof (see the supplement), and do not depend on $K$, $n$, $s$, or $\mathbf{p}$.
\begin{theorem}[Frobenius error bound]
\label{Thm:1}
Suppose the assumptions (A1)-(A3) hold, and that $\hat{\Omega}$ is the minimizer of \eqref{Eq:objFun} with
$\rho_k  \asymp \frac{1}{\underline{k}_\Omega} \sqrt{\frac{\log p}{nm_k}}$. Then with probability at least $1-2(K+1) \exp(-c\log p)$
\begin{align*}
\|\hat{{\Omega}} - {\Omega}_0\|_F &\leq  \frac{2 C_1  \twonorm{\Sigma_0}}{\phi_{\min}^2(\Sigma_0)}
\sqrt{(K+1)\left(s + p\right) \frac{\log p}{n \min_k m_k}}.
\end{align*}

\end{theorem}

%




\begin{theorem}[Factorwise and L2 error bounds]
\label{Thm:1Spec}
Suppose the conditions of Theorem \ref{Thm:1} hold. Then with probability at least $1-2(K+1) \exp(-c\log p)$, 
\begin{align}\label{eq:fcbd}
\frac{\|\mathrm{diag}(\hat{\Omega}) -\mathrm{diag}(\Omega_0)\|_2^2}{(K+1)\max_k d_k} &+ \sum_{k=1}^K \frac{\|\offd(\hat{\Psi}_k - {\Psi}_{0,k})\|_F^2}{d_k}\nonumber\\ &\leq C_2 (K+1)\left(1 + \sum_{k=1}^K \frac{s_k}{d_k} \right)  \frac{\log p}{n\min_k  m_k}
\end{align}
and as a result
\begin{align*}
\|\hat{{\Omega}} - {\Omega}_0\|_2 \leq C_3(K+1) \sqrt{\left( \frac{p}{(\min_k m_k)^2}\right)\left( 1+ \sum_{k=1}^K \frac{s_k}{d_k} \right)\frac{\log p}{n}}.
\end{align*}

\end{theorem}
Theorems \ref{Thm:1} and \ref{Thm:1Spec} are proved in Section \ref{app:Pf} of the supplement.
Observe that the theorem predicts \eqref{eq:fcbd} that, for fixed $n$ and $K>2$, the estimation error of the parameters of $\Omega$ converges to zero as the dimensions $\{d_k\}$ go to infinity (recall that $p=\prod_{k=1}^K d_k$). This implies that for increasing dimensions the TeraLasso will converge even for a single sample $n=1$.  
Due to the repeating structure and increasing dimension of $\Omega$, the parameter estimates can converge without the overall Frobenius error $\|\hat{\Omega} - \Omega_0\|_F$ converging. 



\textbf{Comparison to GLasso.} The Frobenius norm bound in Theorem \ref{Thm:1} improves on the subgaussian GLasso rate of
\cite{rothman2008sparse,zhou2011high} by a factor of $\min_k m_k$.
If the dimensions are equal ($d_k = p^{1/K}$ and $s_k$ are constant over $k$) and $K$ is fixed, Theorem \ref{Thm:1Spec} implies $\|\Delta_k\|_F = O_p\left( \sqrt{\frac{ (d_k + s_k) \log p }{m_k n}}\right)$,
indicating that TeraLasso with $n$ replicates estimates the identifiable representation of $\Psi_{k}$ with an error rate equivalent to that of GLasso with $\Omega = \Psi_k$ and $nm_k$ available replicates. 

\textbf{Independence along an axis.} Suppose that the data tensor $X$ is i.i.d. along the first axis, i.e. $\Psi_1 = I_{d_1}$. Then instead of a $K$-way TeraLasso, a $K-1$ model with $n d_1$ replicates would suffice, yielding a factorwise error bound (Theorem \eqref{Thm:1Spec}) of $O\left(\sqrt{ \left(1 + \sum_{k=2}^K \frac{s_k}{d_k}\right) \frac{\log (p/d_1)}{n d_1 \min_{k>1} (m_k/d_1)}}\right)$, as compared to the factorwise error bound of $O\left(\sqrt{ \left(1 + \sum_{k=2}^K \frac{s_k}{d_k}\right) \frac{\log (p)}{n \min_{k} m_k}}\right)$ associated with the full $K$-way model (since $s_1 = 0$). Hence having a priori knowledge of independence (allowing the use of the $K-1$ model) does not meaningfully improve the rate over the the original $K$-way model so long as $\min_{k>1} m_k \approx \min_{k} m_k$. A similar satisfying result holds for the Frobenius error bound in Theorem \ref{Thm:1}.

\subsection{Nonconvex Regularizers and Single Sample Support Recovery}\label{sec:Nonconv}
Nonconvex regularization will provide nonasymptotic guarantees on the elementwise estimation error, implying strong, single sample support recovery guarantees when the smallest nonzero element of $\Omega_0$ is bounded from below. On the other hand, these stronger results require more restrictive assumptions on sparsity of the precision matrix and its smallest nonzero element.
Specifically, we will require the following:
\begin{description}
\item[(A4)] The degree (maximum number of nonzero edges connected to a node) of the sparsity graph of each factor $\Psi_k$ is bounded by a constant $d$.
\item[(A5)] The sample size satisfies: $n \min_k m_k \geq c_0 d^2 \log p$ for some $c_0$ large enough.
\item[(A6)] There exist constants $c_\infty, c_3$ such that
$\|(\Omega_0 \otimes \Omega_0)_{\mathcal{S}\mathcal{S}}\|_\infty \leq c_\infty$
and
\[
\min_{[i,j] \in \mathcal{S}} |[\Omega_{0}]_{ij}| \geq \rho(\gamma + 2c_\infty) + c_3 \sqrt{\frac{\log p}{n\min_k m_k}}.
\]
\end{description}
In (A6) the notation $A_{\mathcal S \mathcal S}$ denotes the submatrix of $A$ formed by extracting the rows and columns corresponding to the index set $\mathcal{S}$.
Under these assumptions we have the following result. 
\begin{theorem}[Nonconvex Regularizers]\label{thm:NonCon}
Suppose the regularizer $g_\rho$ in \eqref{eq:nonconvobv} is $(\mu,\gamma)$-amenable, and $\kappa = \sqrt{2/\mu}$. Then with probability at least $1-2(K+1) \exp(-c\log p)$ as in Theorem \ref{Thm:1}, \eqref{eq:nonconvobv} has a unique stationary point $\hat{\Omega}$ (given by the oracle estimator defined in the supplement), with (for all $k$)
\[
\|\mathrm{offd}(\hat{\Psi}_k - \Psi_{0,k})\|_{\max} \leq \|\hat{\Omega} - \Omega_0\|_{\max} \leq c_3 (K + 1) \sqrt{\frac{\log p}{n \min_k m_k}},
\]
\[
\|\mathrm{offd}(\hat{\Psi}_k - \Psi_{0,k})\|_{F} \leq c_3 (K + 1) \sqrt{\frac{s_k \log p }{n \min_k m_k}},
\]
\[
\|\hat{\Omega} - \Omega_0\|_{F} \leq c_3 (K + 1) \sqrt{\frac{(s+p)\log p }{n \min_k m_k}},
\]
\[
\|\hat{\Omega} - \Omega_0\|_{2} \leq c_3 d(K + 1) \sqrt{\frac{\log p }{n \min_k m_k}}.
\]
\end{theorem}
The proof is given in Section \ref{app:NonConv} in the supplement, and uses arguments analogous to those of \cite{LOH} along with concentration inequalities arising from the structure of the TeraLasso model.

Theorem \ref{thm:NonCon} implies that the elements (of both $\Omega$ and the offdiagonals of $\Psi_k$), and thus the support (of both $\Omega$ and the $\Psi_k$) can be estimated using a single sample ($n=1$) provided $\min_k m_k$ is large enough. The Frobenius norm convergence rates (both factorwise and overall) for the convex and nonconvex regularizers remain effectively the same (comparing Theorem \ref{thm:NonCon} to Theorems \ref{Thm:1} and \ref{Thm:1Spec}), hence the primary benefit of the nonconvex bound is the ability to guarantee support recovery in exchange for additional assumptions.

\section{TG-ISTA Algorithm}\label{sec:alg}


In this section, we introduce an iterative soft thresholding (ISTA) method, restricted to the convex set $\mathcal{K}_{\mathbf{p}}^\sharp$ of possible positive semidefinite Kronecker sum precision matrices, to implement the TeraLasso optimization \eqref{eq:nonconvobv}. We call this implementation Tensor Graphical Iterative Soft Thresholding (TG-ISTA). 

\subsection{Composite gradient descent and proximal first order methods}

Our goal is to solve the objective \eqref{eq:nonconvobv}.
This objective function can be decomposed into the sum of a differentiable function $f$ and a lower semi-continuous but nonsmooth function $g$: for $\Omega \in \mathcal{K}_{\mathbf{p}}$:
\begin{align}
\label{eq:oj2}
Q(&{\Psi}_1,\dots,{\Psi}_K) = f(\Omega) + g(\Omega),\: \mathrm{where\: for\:}  \langle \hat{S}, \Omega\rangle  = \sum\nolimits_{k=1}^K m_k \langle S_k, \Psi_k \rangle, \nonumber\\
f(\Omega) 
&= \left. -\log|\Omega| + \langle \hat{S}, \Omega\rangle \right|_{\Omega \in \mathcal{K}_{\mathbf{p}}},\:\:
g(\Omega) = \sum\nolimits_{k=1}^K m_k  \sum_{i\neq j} g_{\rho_k}(  [{\Psi}_k]_{ij}).
\end{align}
For objectives of this form, \cite{nesterov2007gradient} proposed a first order method called composite gradient descent. Composite gradient descent has been specialized to the case of $g = |\cdot|_1$ and is widely known as Iterative Soft Thresholding (ISTA) (see for example \cite{tseng2010approximation,combettes2005signal,beck2009fast,nesterov1983method,nesterov2004introductory}). An extension to nonconvex regularizers $g$ is given in \cite{loh2013regularized}.

The linearity of the constraint set $\mathcal{K}_{\mathbf{p}}$ suggests the use of gradient descent where the gradients are projected onto the associated $1 - K + \sum_{k=1}^K d_k^2$ dimensional \emph{linear subspace}. The positive definite restriction can then be handled in a similar way as \cite{GISTA} did for the vanilla GLasso. We therefore derive composite gradient descent in the linear subspace $\mathcal{K}_{\mathbf{p}}$ of $\mathbb{R}^{p^2}$, creating a positive definite sequence of iterates $\{\Omega_t\}$ given by the recursion
\begin{align}\label{eq:cmpgrd}
\Omega_{t+1} \in \arg \min_{\Omega \in \mathcal{K}^\sharp_{\mathbf{p}}} \left\{ \frac{1}{2}\left\| \Omega - \left(\Omega_t - \zeta_t \mathrm{Proj}_{\mathcal{K}_{\mathbf{p}}}(\nabla f(\Omega_t))\right) \right\|_F^2 + \zeta_t g(\Omega) \right\},
\end{align}
where the initial matrix $\Omega_0 \in \mathcal{K}^\sharp_{\mathbf{p}}$ can be chosen as the identity. We enforce the positive semidefinite constraint at each step by performing backtracking line search to find a suitable stepsize $\zeta_t$ (see Algorithm \ref{alg:highlevel}) \citep{GISTA}. We decompose and solve the problem \eqref{eq:cmpgrd} for the case of the TeraLasso objective in Section \ref{sec:decc} below.

\subsection{TG-ISTA implementation of TeraLasso}
\label{sec:decc}

To apply this form of composite gradient descent to the TeraLasso objective, the projected gradient of $f(\Omega)$ is required for \eqref{eq:oj2}. For simplicity, consider the $\ell 1$ regularized case. The general nonconvex case is described in the next section and the supplement.
Since the gradient of $\langle \hat{S}, \Omega\rangle$ with respect to $\Omega$ is $\hat{S}$ 
(Lemma \ref{lem:Projj} in the supplementary material) 
\begin{align}\label{eq:ingrad}
\nabla_{\Omega \in \mathcal{K}_{\mathbf{p}}} (\langle \hat{S}, \Psi_1 \oplus \dots \oplus \Psi_k\rangle) =& \mathrm{Proj}_{\mathcal{K}_{\mathbf{p}}}(\hat{S})
=  \tilde{S}_1 \oplus \dots \oplus \tilde{S}_K = \tilde{S} \quad \mathrm{where}\nonumber\\
\tilde{S}_k &= S_k - \frac{K-1}{K} \frac{\tr(S_k)}{d_k} I_{d_k}.
\end{align}
While many different conventions for parameterizing the projection using the $\tilde{S}_k$ are possible, the projection remains unique. Alternate parameterizations will not affect the convergence or output of the algorithm. Since the gradient of $-\log|\Omega|$ with respect to $\Omega$ is $\Omega^{-1}$ \citep{boyd2009convex}, the projected gradient takes the form 
\begin{align}\label{eq:logGrad}
\nabla_{\Omega \in \mathcal{K}_{\mathbf{p}}} & \left(-\log | \Omega |\right)
= \mathrm{Proj}_{\mathcal{K}_{\mathbf{p}}}\left(\Omega^{-1}\right) = G_1^t \oplus \dots \oplus G_K^t
\end{align}
The matrices $G_k^t \in \mathbb{R}^{d_k \times d_k}$ are computed via the expressions given in Lemma \ref{lem:Projj} in the supplement. Combining \eqref{eq:ingrad} and \eqref{eq:logGrad}, the projected gradient of the objective $f(\Omega_t)$ is 
\begin{align}\label{eq:allGrad}
\mathrm{Proj}_{\mathcal{K}_{\mathbf{p}}}(\nabla f(\Omega_t)) = \tilde{S} - (G_1^t \oplus \dots \oplus G_K^t).
\end{align}
\begin{lemma}[Decomposition of objective]\label{lem:decomp}
For $\Omega_t, \Omega \in \mathcal{K}_{\mathbf{p}}$ of the form
\begin{align*}
\Omega_t &= \Psi^t_1 \oplus \dots \oplus \Psi^t_K\quad \mathrm{and} \quad
\Omega = \Psi_1 \oplus \dots \oplus \Psi_K,
\end{align*}
the unique solution to \eqref{eq:cmpgrd} with $g_\rho = |\cdot|_1$ is given by $\Omega_{t+1} = \Psi_1^{t+1}\oplus \dots \oplus \Psi_K^{t+1}$
where
\begin{equation}\label{eq:newpsi}
\Psi_k^{t+1} = \arg \min_{\Psi_k \in \mathbb{R}^{d_k \times d_k}} \frac{1}{2}\left\| \Psi_k - (\Psi_k^t - \zeta_t (\tilde{S}_k - G_k^t)) \right\|_F^2 + \zeta_t \rho_k  |{\Psi}_k|_{1,\off}.
\end{equation}

\end{lemma}
The proof is in supplement Section \ref{app:decomp}.
The right hand side of \eqref{eq:newpsi} is the proximal operator of the $\ell_1$ penalty on the off diagonal entries. The solution has closed form, as given in \cite{beck2009fast},
\begin{align}\label{eq:upp}
\Psi_k^{t+1} = \mathrm{shrink}_{\zeta_t\rho_k}^- (\Psi_k^t - \zeta_t (\tilde{S}_k - G_k^t)),
\end{align}
where we define the off diagonal shrinkage operator $\mathrm{shrink}^-_{\rho} (\cdot)$ as
\begin{align}
[\mathrm{shrink}^-_{\rho}({M})]_{ij} = \left\{\begin{array}{ll} \sign(M_{ij})(|M_{ij}| - \rho)_+& i \neq j \\ M_{ij}& \mathrm{o.w.} \end{array} \right.
\end{align}

The composite gradient descent algorithm is given in Algorithm \ref{alg:highlevel}. 
In Section \ref{supp:conv} of the supplement, a scalable geometric rate of convergence of TG-ISTA to the global minimum is derived (Theorem \ref{thm:conv}).
In Section \ref{supp:compcomp} of the supplement we show that each iteration can be computed in $
O\left(pK + \sum_{k=1}^K d_k^3\right)$ floating point operations.

\begin{algorithm}[h]
\caption{TG-ISTA implementation of TeraLasso (high level)}
\label{alg:highlevel}
\begin{algorithmic}[1]
\STATE Input: SCM factors $S_k$, regularization parameters $\rho_i$, backtracking constant $c \in (0,1)$, initial step size $\zeta_{1,0}$, initial iterate $\Omega_{\mathrm{init}} = I \in \mathcal{K}_{\mathbf{p}}^\sharp$.

\WHILE{ not converged}

\STATE Compute the subspace gradient $\mathrm{proj}_{\mathcal{K}_{\mathbf{p}}}\left(\Omega^{-1}_t\right) = G_1^t \oplus \dots \oplus G_K^t$.
\STATE \emph{Line search}: Let stepsize $\zeta_t$ be the largest element of $\{c^j \zeta_{t,0}\}_{j =1,\dots}$ such that the following are satisfied for $\Psi_k^{t+1} = \mathrm{shrink}^-_{\zeta_t \rho_k}(\Psi_k^t - \zeta_t (\tilde{S}_k - G_k^t))$
\[
\Psi_1^{t+1} \oplus \dots \oplus \Psi_K^{t+1} \succ 0 \quad \mathrm{and}\quad f(\{\Psi_k^{t+1}\}) \leq \mathcal{Q}_{\zeta_t}(\{\Psi_k^{t+1}\}, \{\Psi_k^{t+1}\})
\]

\FOR{$k = 1,\dots,K$}

\STATE \emph{Composite objective gradient update}: 
\[
\Psi_k^{t+1} \gets \mathrm{shrink}^-_{\zeta_t \rho_k}(\Psi_k^t - \zeta_t (\tilde{S}_k - G_k^t)).
\]
\ENDFOR
\STATE Compute Barzilai-Borwein stepsize $\zeta_{t+1,0}$ via \eqref{eq:BarBor} in Supplement \ref{supp:stepsz}.

\ENDWHILE
\STATE Return $\{\Psi_k^{t+1}\}_{k =1 }^K$.

\end{algorithmic}
\end{algorithm}




\subsection{TG-ISTA for a nonconvex regularizer}



The estimation algorithm is largely the same as Algorithm \ref{alg:highlevel}, except with an additional term added to the gradient. Specifically, the updates are of the form
\begin{equation}\label{eq:ncalg}
\Omega^{t + 1} = \mathrm{shrink}^-_{\zeta\rho}\left(\Omega^t - \zeta\nabla \bar{\mathcal{L}}_n(\Omega^t)\right)
\end{equation}
where $\zeta$ is the step size and 
\[
\bar{\mathcal{L}}_n(\Omega) = -\log | \Omega| + \langle \hat{S}, \Omega\rangle + \sum_{k=1}^K m_k \sum_{i\neq j}\left(g_\rho({[{\Psi}_k]_{ij}}) - \rho |{\Psi}_k]_{ij}|\right).
\]
The update \eqref{eq:ncalg} can be decomposed into the factorwise updates
\[
\Psi_k^{t + 1} = \mathrm{shrink}^-_{\zeta\rho}\left(\Psi_k^t - \zeta (\tilde{S}_k - G^t_k + q'_\rho(\Psi_k)) \right)
\]
where $q'_\rho(t) = \frac{d}{dt} (g_\rho(t) - \rho |t|)$ for $t \neq 0$ and $q'_\rho(0) = 0$. 
These updates can be inserted into the framework of Algorithm \ref{alg:highlevel}, with an added step of enforcing the $\|\Omega\|_2 \leq \kappa$ constraint, e.g. via step size line search. The algorithm is summarized in Algorithm \ref{alg:highlevelNC} in Supplement \ref{supp:nonconvalg}.

\begin{theorem}[Convergence of Algorithm \ref{alg:highlevelNC}]
Algorithm \ref{alg:highlevelNC} will converge to the global optimum when the norm constraint parameter $\kappa$ is chosen to be less than or equal to $\sqrt{2/\mu}$.
\end{theorem}
\begin{proof}
Follows since for $\kappa \leq \sqrt{2/\mu}$ the objective \eqref{eq:nonconvobv} is convex on the convex constraint set $\{\Omega \in \mathcal{K}_{\mathbf{p}} | \Omega \succ 0, \|\Omega\|_2 \leq \kappa\}$ (Lemma \ref{lem:ncconv}, supplement).
\end{proof}

\section{Validation on synthetic data}
\label{sec:exp}

Random graphs were created for each factor $\Psi_k$ using both an Erdos-Renyi (ER) topology and a random grid graph topology\footnote{Code for experiments is included in the supplementary material and can be found at https://github.com/kgreenewald/teralasso.}. These ER type graphs were generated according to the method of \cite{zhou2010time}. Initially we set $\Psi_k = 0.25I_{n \times n}$, where $n = 100$, and randomly select $q$ edges and update $\Psi_k$ as follows: for each new edge $(i,j)$, a weight $a>0$ is chosen uniformly at random from $[0.2,0.4]$; we subtract $a$ from $[\Psi_k]_{ij}$ and $[\Psi_k]_{ji}$, and increase $[\Psi_k]_{ii},[\Psi_k]_{jj}$ by $a$. This keeps $\Psi_k$ positive definite. We repeat this process until all edges are added. Finally, we form $\Omega = \Psi_1 \oplus \dots \oplus \Psi_K$. An example 25-node, $q = 25$ ER graph and precision matrix are shown in Figure \ref{Fig:GER}. 
\begin{figure}[h]
\begin{subfigure}[b]{\textwidth}
\centering
\includegraphics[width=2.8in]{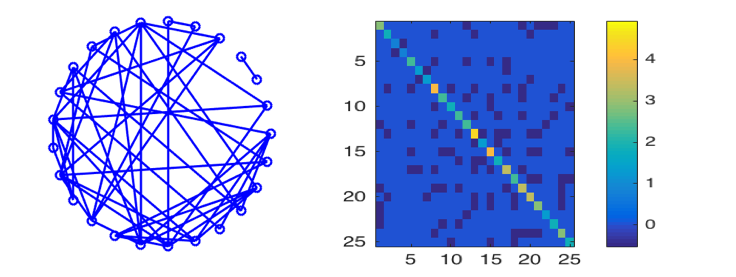}
\caption{Random Erdos-Renyi (ER) graph with 25 nodes and 50 edges}
\end{subfigure}
\hfill
\begin{subfigure}[b]{\textwidth}
\centering
\includegraphics[width=2.8in]{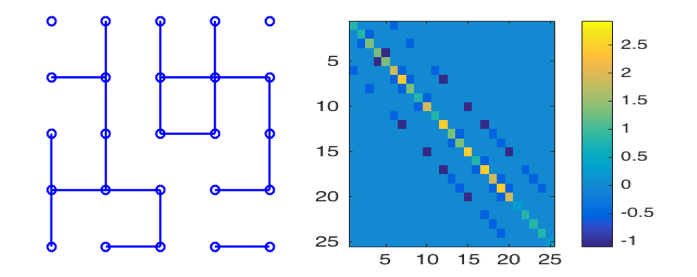}
\caption{Random grid graph (square) with 25 nodes and 26 edges}
\end{subfigure}
\caption{Example Erdos-Renyi and random grid graphs. Left: Graphical representation. Right: Corresponding precision matrix $\Psi$. }
\label{Fig:GER}
\end{figure}
The random grid graph is produced in a similar way, with the exception that edges are only allowed between adjacent nodes, where the nodes are arranged on a square grid (Figure \ref{Fig:GER}(b)). Algorithm 1 in Section \ref{App:Alg} of the supplement describes how the random vector $\mathbf{x} = \mathrm{vec}(X^T)$ is generated under the Kronecker sum model. 

\subsection{Validation of theoretical algorithmic convergence rates}

To verify the geometric convergence of the TG-ISTA implementation (Theorem \ref{thm:conv} in the supplement), we generated Kronecker sum inverse covariance graphs and plotted the Frobenius norm between the inverse covariance iterates $\Omega_t$ and the optimal point $\Omega^*$. We set the $\Psi_k$ to be random ER graphs with $d_k$ edges where $d_1 = \dots = d_K$, and determined the value for $\rho_k = \rho$ using cross validation. Figure \ref{Fig:Conv} shows the results as a function of iteration, for a variety of $d_k$ and $K$ configurations and the $\ell$1 convex regularization. Figure \ref{Fig:ConvNon} in Supplement \ref{supp:nonconvalg} repeats these experiments with the nonconvex SCAD and MCP penalties, using the same random seed.
For comparison, the statistical error of the optimal point is also shown, as optimizing beyond this level provides reduced benefit. As predicted, linear or better convergence to the global optimum is observed. 
The small number of iterations combined with the low computational cost per iteration confirm the algorithmic efficiency of the TG-ISTA implementation of TeraLasso. Additional numerical experiments demonstrating fast convergence on larger scale problems are given in Section \ref{supp:compcomp} of the supplement.

\begin{figure}[h]
\begin{subfigure}[b]{\textwidth}
\centering
\begin{minipage}[c]{0.3\textwidth}
\caption{$\ell$1 penalty,\\ $n=100$ sample size}
\end{minipage}\hfill
\begin{minipage}[l]{0.7\textwidth}
\includegraphics[width=1.7in]{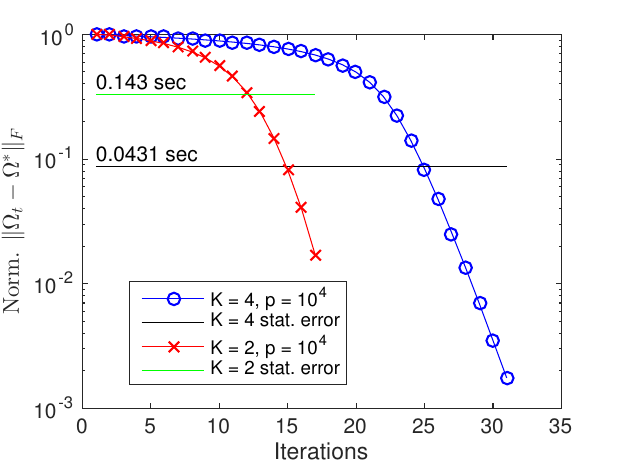}\includegraphics[width=1.7in]{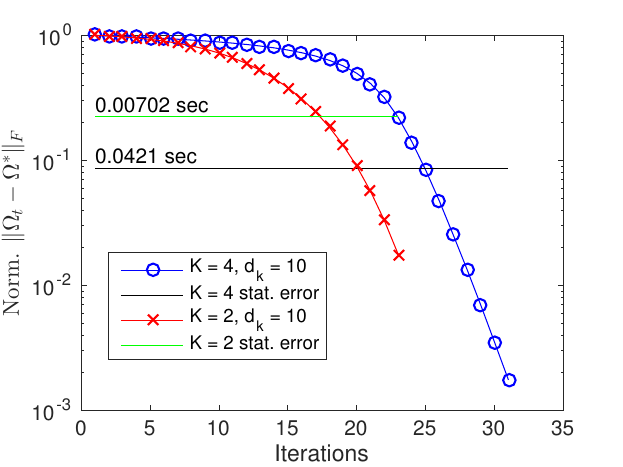}
\end{minipage}
\end{subfigure}
\vfill
\begin{subfigure}[b]{\textwidth}
\centering
\begin{minipage}[c]{0.3\textwidth}
\caption{$\ell$1 penalty,\\ $n=1$ sample size}
\end{minipage}\hfill
\begin{minipage}[l]{0.7\textwidth}
\includegraphics[width=1.7in]{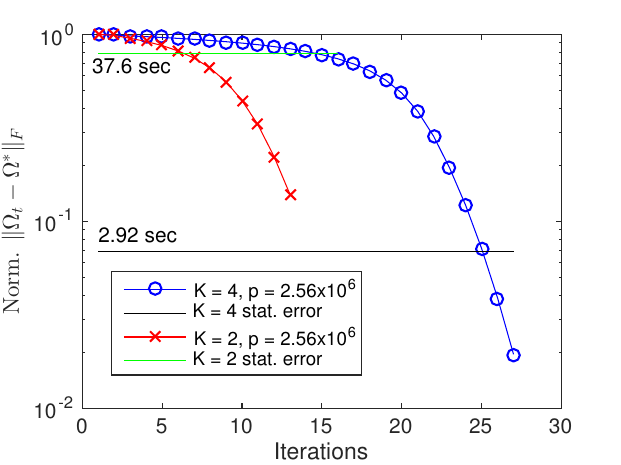}\includegraphics[width=1.7in]{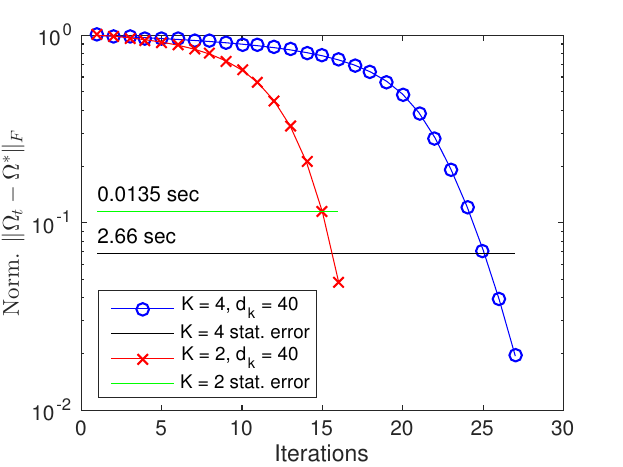}
\end{minipage}
\end{subfigure}
\caption{Linear geometric convergence of the convex ($\ell$1 penalized) TG-ISTA implementation of TeraLasso. Shown is the normalized Frobenius norm $\|{\Omega}_t - \Omega^*\|_F$ of the difference between the estimate at the $t$th iteration and the optimal $\Omega^*$. On the left are results comparing $K = 2$ and $K=4$ on the same data with the same value of $p$ (different $d_k$), on the right they are compared for the same value of $d_k$ (different $p$). Also included are the statistical error levels, and the computation times required to reach them. Observe the consistent and rapid linear convergence rate, with logarithmic dependence on $K$ and dimension $d_k$.  }
\label{Fig:Conv}
\end{figure}

\subsection{Regularization with $\ell 1$ penalty}
In the TeraLasso objective \eqref{Eq:objFun}, the sparsity of the estimate is controlled by $K$ distinct tuning parameters $\rho_k$ for $k=1,\dots, K$. The convergence condition on $\rho_k$ in Theorem \ref{Thm:1} suggests that the $\rho_k$ can be set as
$
\rho_k = \bar{\rho} \sqrt{\frac{\log p}{n m_k}}
$
with $\bar{\rho}$ being a single scalar tuning parameter, depending on absolute constants and $\|\Sigma\|_2$. Below, we experimentally validate the reliability of this tuning strategy.

\begin{figure}[h]
\centering
\includegraphics[width=3.80in,height = 1.2in]{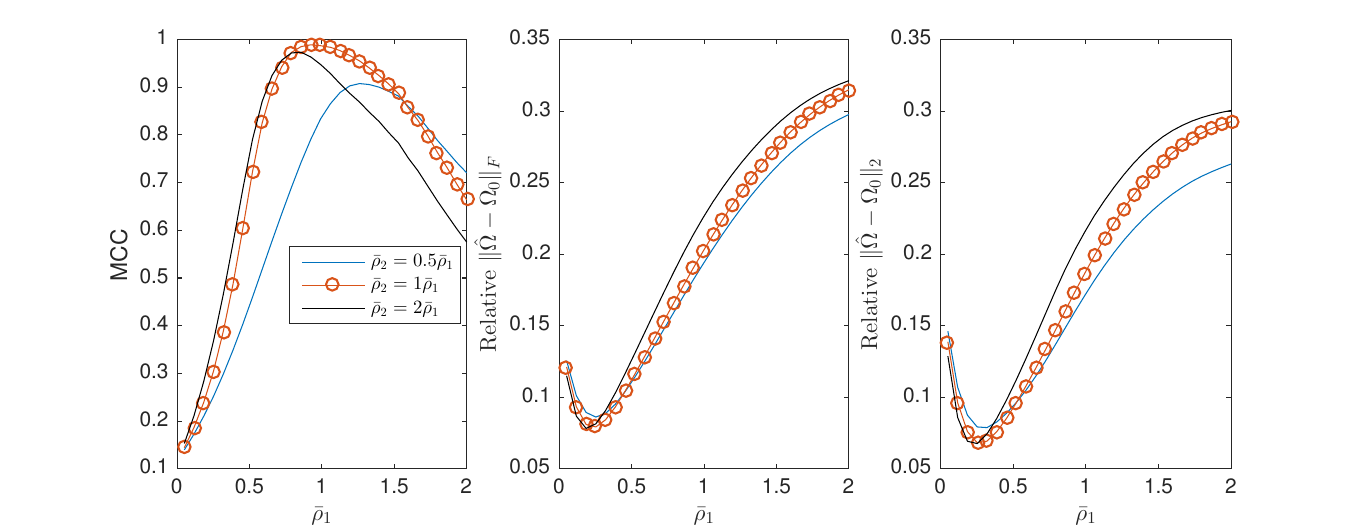}\\
\includegraphics[width=3.8in,height = 1.2in]{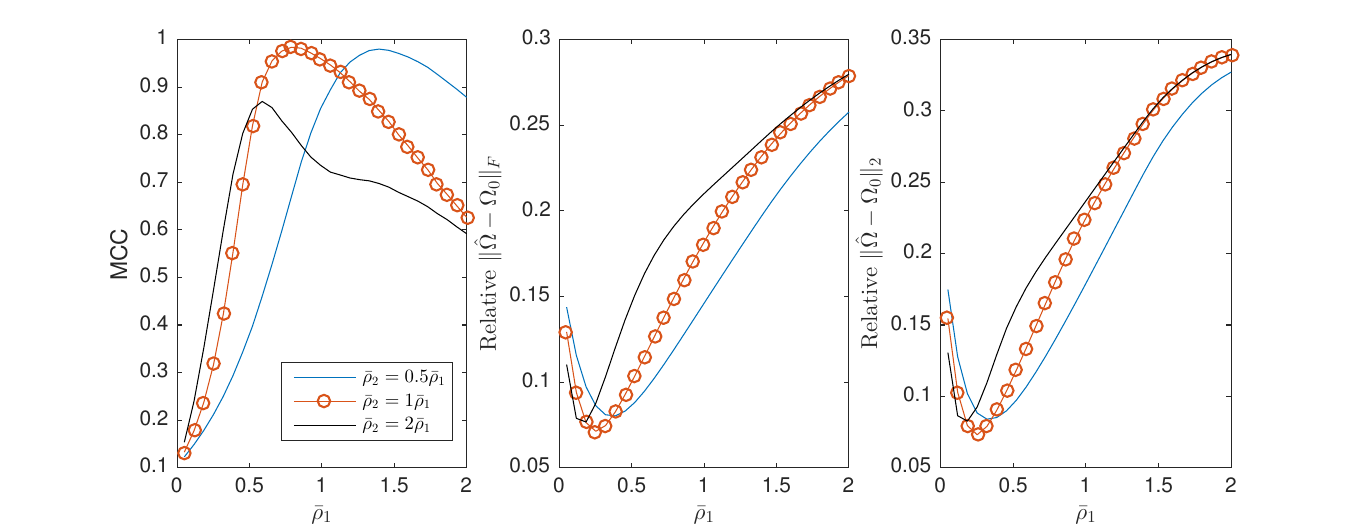}
\caption{Setting tuning parameters with $K=3$, $n = 1$, and $d_1 = d_3 = 64$. Shown are the MCC, relative Frobenius error, and relative L2 error of the TeraLasso estimate as the scaled tuning parameters $\rho_i$ are varied. Shown are deviations of $\bar{\rho}_2$ from the theoretically dictated $\bar{\rho}_2 = \bar{\rho}_1 = \bar{\rho}_3$. 
Top: Equal dimensions, $d_1 = d_2 = d_3$. First and third factors are random ER graphs with $d_k$ edges, and the second factor is random grid graph with $d_k/2$ edges. Bottom: Dimensions $d_2 = 2 d_1$, each factor is a random ER graph with $d_k$ edges. Notice in these scenarios that using $\bar{\rho}_1 = \bar{\rho}_2$ is near optimal, as theoretically predicted.  }
\label{Fig:Tune}
\end{figure}

The performance is empirically evaluated using several metrics including: the Frobenius norm ($\|\hat{\Omega}-\Omega_0\|_F$) and spectral norm ($\|\hat{\Omega}-\Omega_0\|_2$) error of the precision matrix estimate $\hat{\Omega}$ and the Matthews correlation coefficient to quantify the edge misclassification error. Let the number of true positive edge detections be TP, true negatives TN, false positives FP, and false negatives FN. The Matthews correlation coefficient is defined as \citep{matthews1975comparison}
\[
\mathrm{MCC} = \frac{\mathrm{TP \cdot TN - FP \cdot FN}}{\sqrt{\mathrm{(TP + FP)(TP + FN)(TN+FP)(TN + FN)}}},
\]
where each nonzero off diagonal element of $\Psi_k$ is considered as a single edge.
Larger values of MCC imply better edge estimation performance, with $\mathrm{MCC} = 0$ implying complete failure and $\mathrm{MCC} = 1$ perfect edge set estimation.

Shown in Figure \ref{Fig:Tune} are the MCC, normalized Frobenius error, and spectral norm error as functions of $\bar{\rho}_1$ and $\bar{\rho}_2$ where the $\bar{\rho}_k$ constants giving
${\rho}_k = \frac{\bar{\rho}_k}{\sqrt{(\log p)/(n m_k)}}.
$ 
Note $\bar{\rho}_1 = \bar{\rho}_2 = \bar{\rho}_3$ achieves near optimal results. 



Having verified the single tuning parameter approach, hereafter we will cross-validate only $\bar{\rho}$. In supplement Section \ref{supp:convVerif}, we provide experimental verification in a wide variety of experimental settings (including varying the relative size of the tensor dimensions $d_k$) that our bounds on the rate of convergence for the $\ell$1 regularized model are tight. 
Figure \ref{Fig:None1} illustrates how increasing dimension $p$ and $K$ improves single sample performance. Shown are the average TeraLasso edge detection precision and recall values for different values of $K$ in the single and 5-sample regimes, all increasing to $1$ (perfect structure estimation) as $p$, $K$, and $n$ increase. 
\begin{figure}[h]
\begin{subfigure}[b]{\textwidth}
\centering
\begin{minipage}[c]{0.15\textwidth}
\caption{$n=1$}
\end{minipage}\hfill
\begin{minipage}[l]{0.85\textwidth}
\includegraphics[width=3.6in]{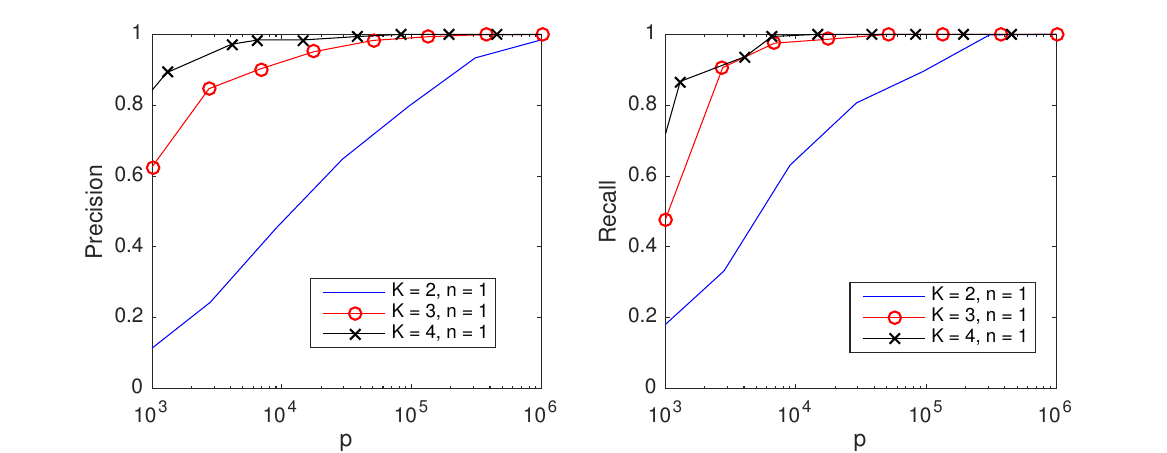}
\end{minipage}
\end{subfigure}
\vfill
\begin{subfigure}[b]{\textwidth}
\centering
\begin{minipage}[c]{0.15\textwidth}
\caption{$n=5$}
\end{minipage}\hfill
\begin{minipage}[l]{0.85\textwidth}
\includegraphics[width=3.6in]{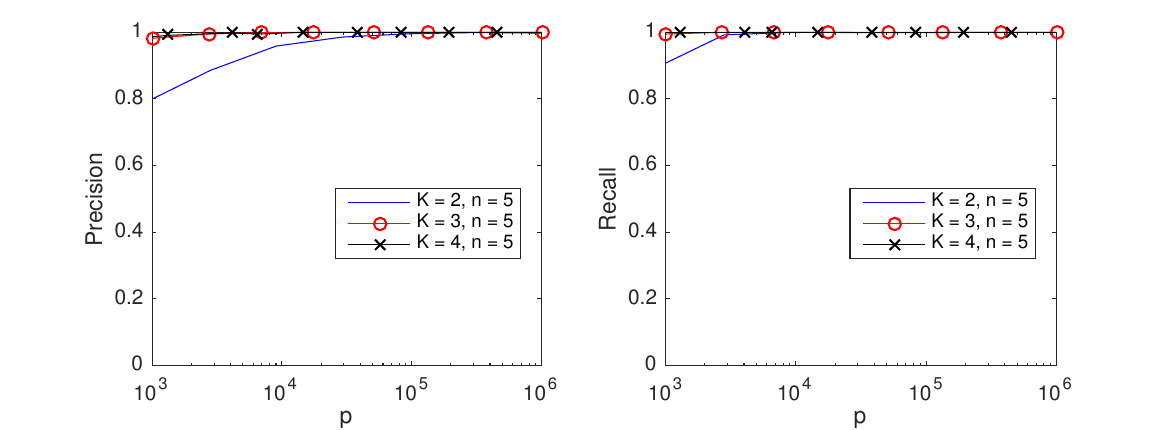}
\end{minipage}
\end{subfigure}
\caption{
Edge support estimation on random ER graphs, with the $\rho_k$ set according to Theorem 1. Graphical model edge detection precision and recall curves are shown as a function of data dimension $p = \prod_{k = 1}^Kd_k$. For each value of the tensor order $K$, we set $d_k=p^{1/K}$. Observe single sample convergence as the dimension $p$ increases and as increasing $K$ creates additional structure. }
\label{Fig:None1}
\end{figure}

\subsection{Nonconvex Regularization}
Here the $\ell 1$ penalized TeraLasso is compared to TeraLasso with nonconvex regularization \eqref{eq:nonconvobv}.
Shown in Figure \ref{Fig:TuneNonConv} are the MCC, normalized Frobenius error, and spectral norm error for estimating $K=2$ and $K=3$ Erdos-Renyi graphs as functions of regularization parameter ${\rho}$ for each of $\ell$1, SCAD \eqref{eq:SCAD}, and MCP \eqref{eq:MCP} regularizers in a variety of configurations. Figure \ref{Fig:TuneNonConvSpiked} shows similar results for $\Psi_k$ a variant of the spiked identity model of \cite{LOH}. Observe that nonconvex regularization improves performance slightly, not only for structure estimation (MCC) but for the Frobenius norm error (due to the reduction in bias) as well. This improvement is increased in the spiked identity case. 
\begin{figure}[htb]
\centering
\begin{subfigure}[b]{\textwidth}
\centering
\includegraphics[width=4in,height = 1.1in]{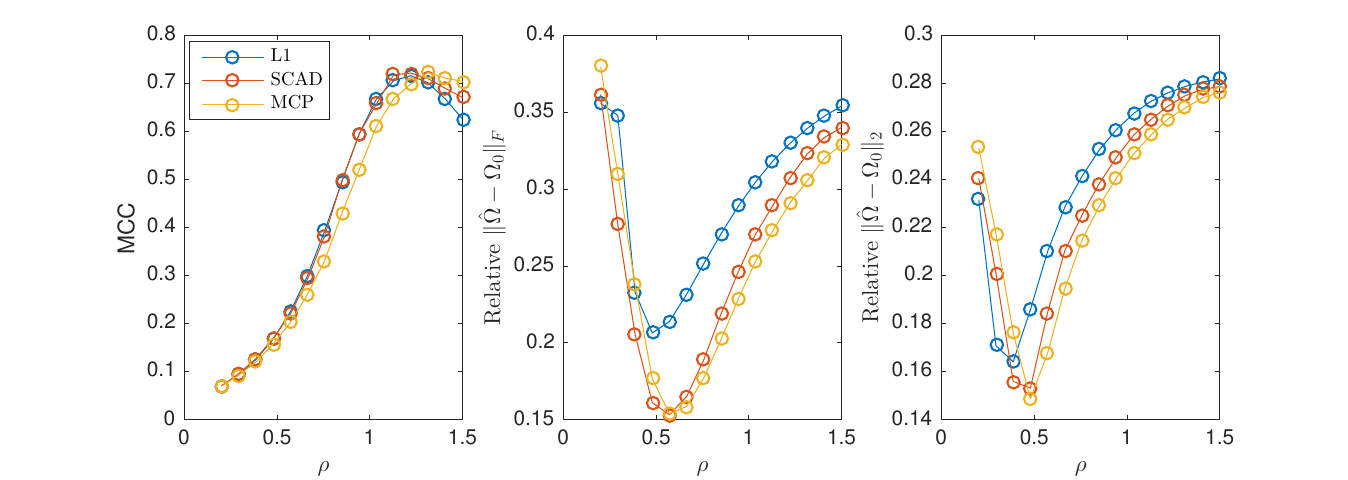}
\caption{$K=2$, $d_1 = d_2 = 1024$}
\end{subfigure}
\begin{subfigure}[b]{\textwidth}
\centering
\includegraphics[width=4in,height = 1.1in]{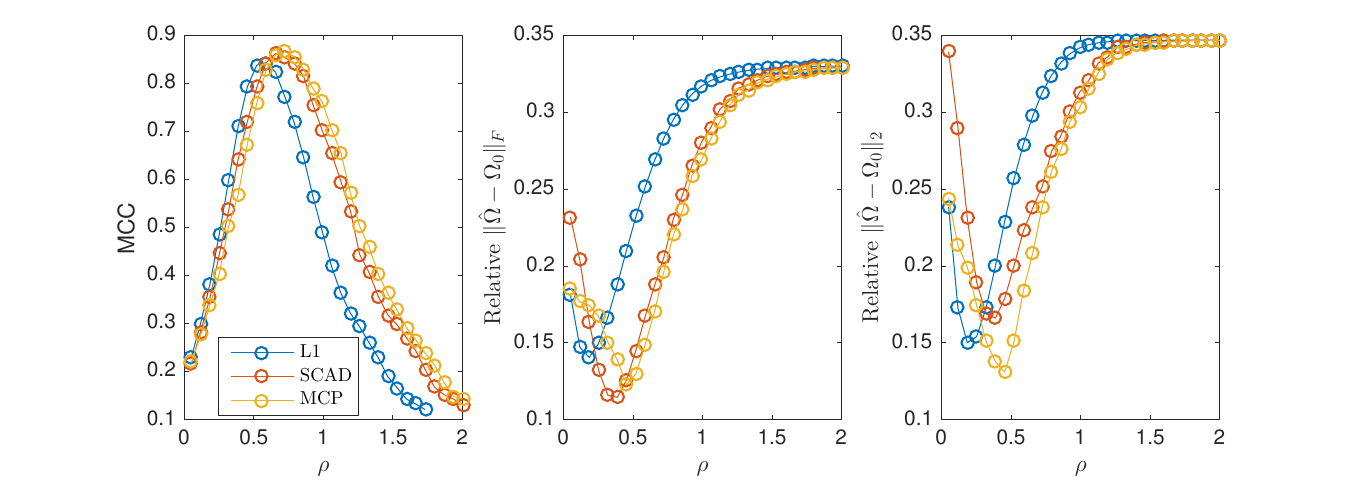}
\caption{$K=3$, $d_1 = d_2 = d_3 = 32$}
\end{subfigure}
\caption{Nonconvex regularizers in the single sample regime ($n=1$, $\Psi_k$ ER with $d_k$ edges). Shown are the MCC, relative Frobenius error, and relative L2 error as a function of ${\rho}$. Note nonconvex regularization improves performance. 
}
\label{Fig:TuneNonConv}
\end{figure}

\begin{figure}[htb]
\centering
\begin{subfigure}[b]{\textwidth}
\centering
\includegraphics[width=5in]{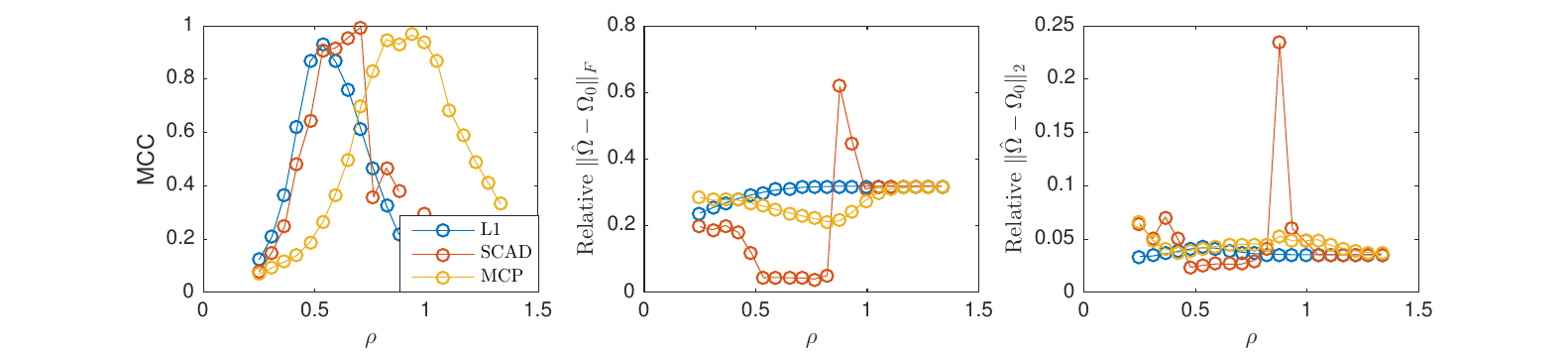}
\caption{$K=2$, $d_1 = d_2 = 256$, $n = 10$, $\Psi_k = 0.5I_{d_k} + 0.5\left[{1}_{8} ; {0}_{248} \right]\left[{1}_{8} ; {0}_{248} \right]^T$}
\end{subfigure}
\begin{subfigure}[b]{\textwidth}
\centering
\includegraphics[width=5in]{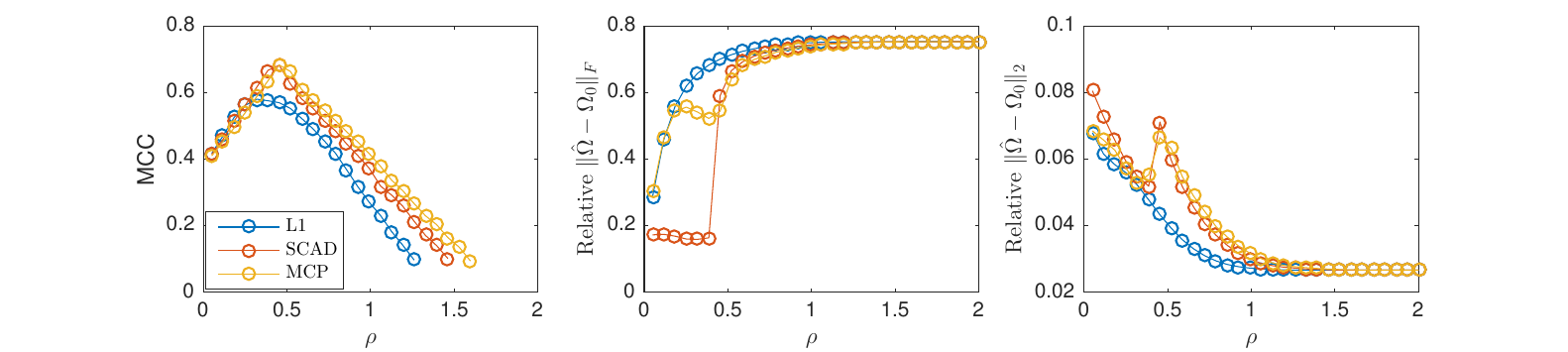}
\caption{$K=3$, $d_1 = d_2 = d_3 = 32$, $n=1$, $\Psi_k = 0.5I_{d_k} + 0.5\left[{1}_{16} ; {0}_{16} \right]\left[{1}_{16} ; {0}_{16} \right]^T$}
\end{subfigure}
\caption{Nonconvex regularizers with spiked identity factors $\Psi_k$. Shown are the MCC and relative Frobenius error as a function of ${\rho}$. Note nonconvex regularization improves performance when $\rho$ is chosen correctly. 
}
\label{Fig:TuneNonConvSpiked}
\end{figure}



\section{NCEP Windspeed Data}\label{sec:real}
The TeraLasso model is illustrated on a meteorological dataset. The US National Center for Environmental Prediction (NCEP) maintains records of average daily wind velocities in the lower troposphere, with daily readings beginning in 1948. The data is available online at ftp://ftp.cdc.noaa.gov/Datasets/ ncep.reanalysis.dailyavgs/surface. Velocities are recorded globally, in a $144 \times 73$ latitude-longitude grid with spacings of 2.5 degrees in each coordinate. Over bounded areas, the spacing is approximately a rectangular grid, suggesting a $K=2$ model (latitude vs. longitude) for the spatial covariance, and a $K=3$ model (latitude vs. longitude vs. time) for the full spatio-temporal covariance. 

\begin{figure}[h]
\centering
\includegraphics[height=1.3in]{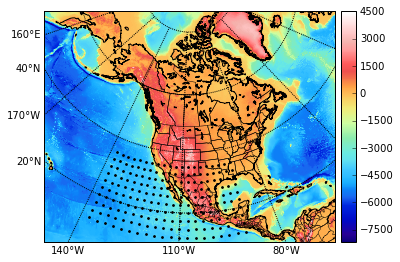}\includegraphics[height=1.3in]{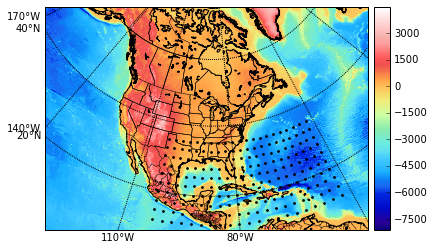}
\caption{Rectangular $10\times 20$ latitude-longitude grids of windspeed locations shown as black dots. Elevation colormap shown in meters. Left: ``Western grid", Right: ``Eastern grid".  }
\label{Fig:Nation}
\end{figure}

\begin{figure}[h]
\begin{subfigure}[b]{\textwidth}
\centering
\includegraphics[width=.06in]{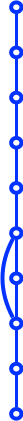}\includegraphics[width=3.5in]{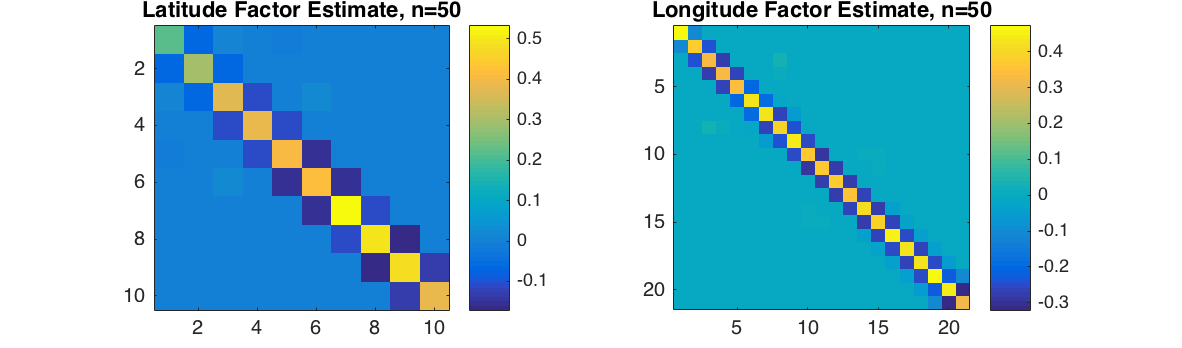}\\\includegraphics[width=3in]{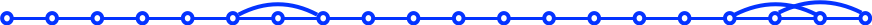}
\caption{Eastern grid. Graphical representation of latitude (left, 10 nodes) and longitude factors (bottom, 20 nodes) with the corresponding precision estimates. Note the simple AR-1 type structure of the longitude graph. 
}
\end{subfigure}
\vfill
\begin{subfigure}[b]{\textwidth}
\centering
\includegraphics[width=.06in]{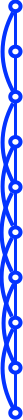}\includegraphics[width=3.5in]{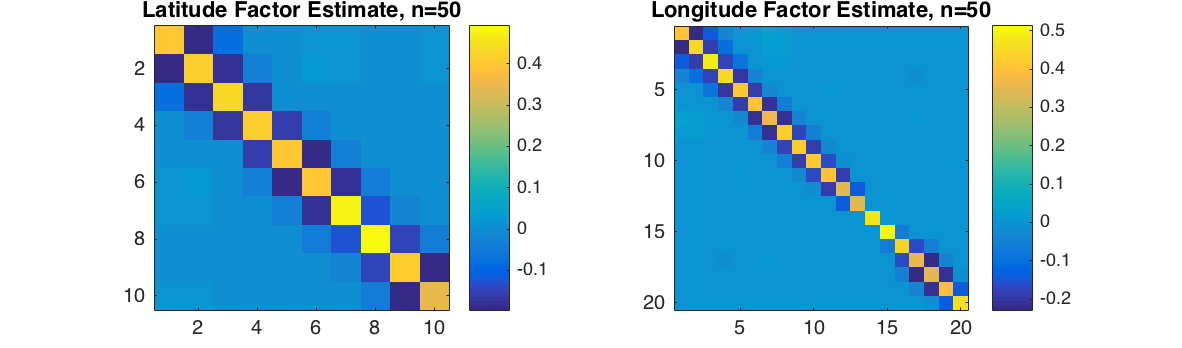}\\\includegraphics[width = 3in]{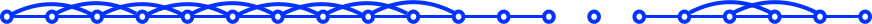}
\caption{Western grid. Graphical representation of latitude (left) and longitude factors (bottom) with the corresponding precision estimates. Observe the decorrelation (longitude factor entries connecting nodes 1-13 to nodes 14-20 are essentially zero) in the Western longitudinal factor, corresponding to the high-elevation line of the Rocky Mountains.}
\end{subfigure}
\caption{TeraLasso estimate factors, $K=2$.  }
\label{Fig:WFac}
\end{figure}

Consider the time series of daily-average wind speeds. Following \cite{tsiliArxiv}, we regress out the mean for each day in the year via a 14-th order polynomial regression on the entire history from 1948-2015. 
We extract two $20\times 10$ spatial grids, one from eastern North America, and one from western North America 
(Figure \ref{Fig:Nation}). 
Figure \ref{Fig:WFac} shows the TeraLasso estimates for latitude and longitude factors 
using time samples from January in $n$ years following 1948, for both the eastern and western grids. Observe the approximate AR structure, and the break in correlation (Figure \ref{Fig:WFac} (b), longitude factor) in the Western Longitude factor. The location of this break corresponds to the high elevation line of the Rocky Mountains.
In the supplement, we compare the TeraLasso estimator to the unstructured shrinkage estimator, the non-sparse Kronecker sum estimator (TeraLasso estimator with sparsity parameter $\rho = 0$), and the Gemini sparse Kronecker product estimator of \cite{zhou2014gemini}. It is shown that the TeraLasso provides a significantly better fit to the data.

To illustrate the utility of the estimated precision matrices, we use them to construct a season classifier. NCEP
windspeed records are taken from the 51-year span from 1948-2009. We
estimate spatial precision matrices on $n$ consecutive days in
January and June of a training year respectively, and running anomaly
detection on $m = 30$-day sequences of observations in the remaining 50
testing years. We report average classifier performance by averaging over
all 51 possible partitions of the 51-year data into 1 training and 50
testing years. The sequences are labeled as summer (June), and winter
(January), and we compute the classification error rate for the winter vs.
summer classifier obtained by choosing the season associated with the
larger of the likelihood functions
\begin{align*}
 \log |\hat{\Omega}_{\mathrm{summer}}| -\sum\nolimits_{i = 1}^m (\mathbf{x}_i-\mathbf{\mu}_i)^T \hat{\Omega}_{\mathrm{summer}} (\mathbf{x}_i-\mathbf{\mu}_i)\\
 \log |\hat{\Omega}_{\mathrm{winter}}| -\sum\nolimits_{i = 1}^m (\mathbf{x}_i-\mathbf{\mu}_i)^T \hat{\Omega}_{\mathrm{winter}} (\mathbf{x}_i-\mathbf{\mu}_i).
 \end{align*}
We consider the $K=3$ spatial-temporal precision matrix for a spatial-temporal array of size $10\times 20 \times T$, with the first ($10\times 10$) factor corresponding to the latitude axis of the spatial array, the second a $20\times 20$ factor corresponding to the longitude axis, and the third factor a $T \times T$ factor corresponding to a temporal axis of length $T$. The spatial-temporal array is created by concatenating $T$ temporally consecutive $10 \times 20$ spatial samples. We use $\ell 1$ regularization.

Results for different sized temporal covariance extents ($T = d_3$) are shown in Figure \ref{Fig:WAnom3} for TeraLasso, with unregularized TeraLasso (ML Kronecker Sum) and maximum likelihood Kronecker product estimator \citep{werner2008estimation,tsiligkaridis2013convergence} results shown for comparison. In this experiment, we use the ML Kronecker product estimator instead of the Gemini, as for this maximum-likelihood classification task the maximum-likelihood based approach performs significantly better than the factorwise objective approach of the Gemini estimators, which is not surprising as the Kronecker product is not a good fit for this data (Section \ref{supp:fit} of the supplement). Note the superior performance and increased single sample robustness of the proposed ML Kronecker Sum and TeraLasso estimates as compared to the Kronecker product estimate, confirming the better fit of TeraLasso. In each case, the nonmonotonic behavior of the Kronecker product curves is due partly to randomness associated with the small test sample size, and partly due to the fact that the Kronecker product in $K=3$ has overly strong coupling across tensor directions, giving large bias. 

\begin{figure}[h]
\begin{subfigure}[b]{\textwidth}
\centering
\includegraphics[width=5in]{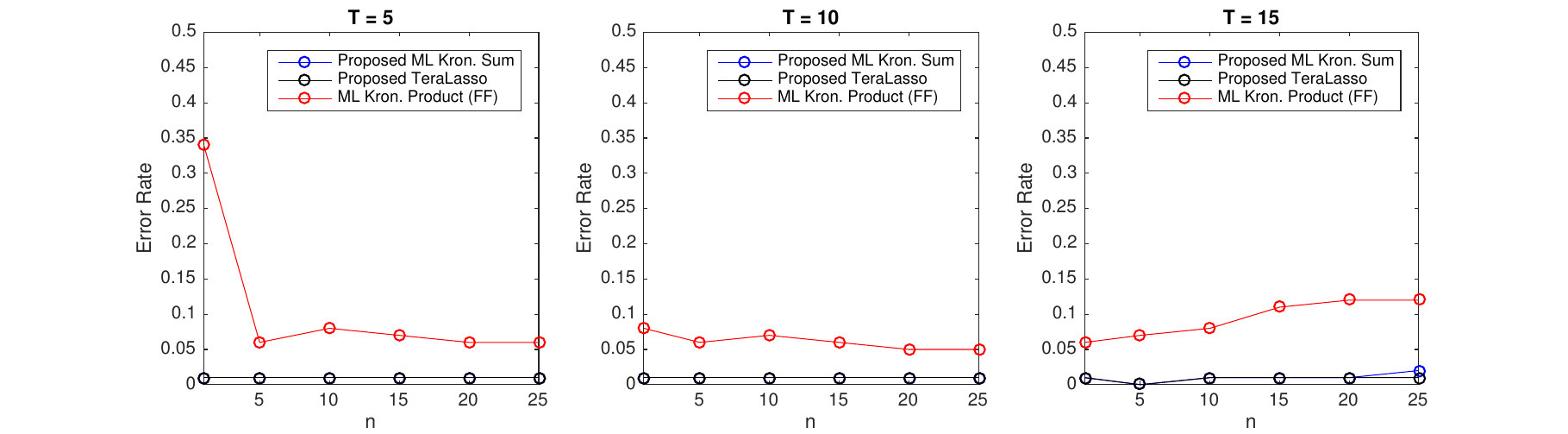}
\caption{Eastern grid}
\end{subfigure}
\vfill
\begin{subfigure}[b]{\textwidth}
\centering
\includegraphics[width=5in]{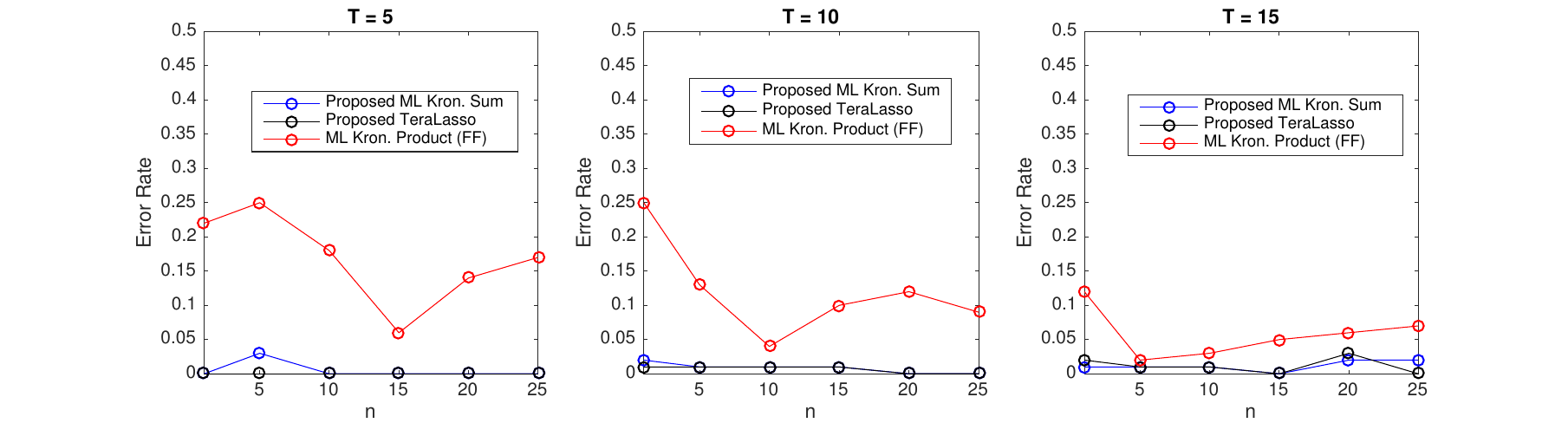}
\caption{Western grid}
\end{subfigure}
\caption{Classification using Gaussian loglikelihood and estimated spatio-temporal ($K=3$) precision matrices for each season, where $T$ is the temporal dimension in days. Shown is windspeed summer vs. winter classification error rate as a function of sample size $n$ and length of temporal window $T$. Note the stability of the Kronecker sum estimate in the $n=1$ case with low error rate. }
\label{Fig:WAnom3}
\end{figure}

\section{Conclusion}
\label{sec:conc}

A factorized model, called the TeraLasso, is proposed for the precision matrix of tensor-valued data that uses Kronecker sum structure and sparsity to regularize the precision matrix estimate. An ISTA-like optimization algorithm is presented that scales to high dimensions. Statistical and algorithmic convergence are established for the TeraLasso that quantify performance gains relative to other structured and unstructured approaches. Numerical
results demonstrate single-sample convergence as well as tightness of the
bounds. Finally, an application to real tensor-valued ($K=3$) meteorological data is considered,
where the TeraLasso model is shown to fit the data well and enable
improved single-sample performance for estimation and anomaly
detection. Future work includes combining first moment tensor representation methods for mean estimation such as PARAFAC \citep{harshman1994parafac} with the second order TeraLasso method introduced in this paper for estimating the covariance. 

\section{Acknowledgement}
The research reported in this paper was partially supported by US Army Research Office grant W911NF-15-1-0479, US Department of Energy grant DE-NA0002534, NSF grant DMS-1316731, and
the Elizabeth Caroline Crosby Research Award from the Advance Program at the
University of Michigan.


\bibliographystyle{rss}
\bibliography{KronML_bib}

\begin{appendices}

\section{Appendix outline}\label{supp:outline}
This supplement is organized as follows. Sections \ref{supp:algo}-\ref{supp:addlExp} focus on the implementation and numerical convergence of the TeraLasso algorithm and Sections \ref{app:DO}-\ref{supp:conv} focus on theory and proofs of convergence. 
Section \ref{supp:algo} presents the algorithm for TeraLasso with nonconvex regularization and describes additional properties of the TeraLasso algorithm, including a discussion of the choice of
step size, decomposition of the gradient update, and proof of joint
convexity of the objective. Section \ref{supp:addlExp} presents additional numerical experiments, including convergence of the nonconvex algorithm, larger scale TG-ISTA convergence experiments, 
additional discussion comparing the fit of the TeraLasso model to the wind speed data, and a discussion of the geometric differences between the Gemini and TeraLasso objectives.  

We then proceed to the convergence analysis. Section \ref{app:DO} describes properties
of the Kronecker sum and the Kronecker sum subspace
$\mathcal{K}_{\mathbf{p}}$ that are needed for the remainder of the discussion.
Proof of the main Frobenius norm theorem and of the spectral norm
theorem are in Section \ref{app:Pf},
with the
concentration bounds proven in Section \ref{App:Conc}. Section \ref{app:NonConv} proves the result on nonconvex regularization, and
Section \ref{supp:conv} presents and proves theorems on the geometrical convergence of the TG-ISTA algorithm. Relevant properties and identities relating to the space $\mathcal{K}_{\mathbf{b}}$ spanned by Kronecker sum matrices are contained in Appendix \ref{App:Ident}, and a discussion of the case where the diagonal elements of $\Omega$ are known is given in Appendix \ref{app:B}.

\section{TeraLasso algorithm step size and numerical convergence proofs}
\label{supp:algo}

\subsection{Convergence of nonconvex regularization algorithm}
\label{supp:nonconvalg}
The TG-ISTA implementation of the TeraLasso algorithm for nonconvex regularizers is shown in Algorithm \ref{alg:highlevelNC}. The primary differences from the $\ell$1 regularized case are (a) the addition of the norm constraint, and (b) the use of the nonconvex regularizer in the gradient computation.

\begin{algorithm}[h]
\caption{TG-ISTA implementation of TeraLasso with nonconvex regularization}
\label{alg:highlevelNC}
\begin{algorithmic}[1]
\STATE Input: SCM factors $S_k$, regularization parameter $\rho$, regularizer $g_{\rho}(\cdot)$ and associated $q'_{\rho}(\cdot)$, backtracking constant $c \in (0,1)$, initial step size $\zeta_{1,0}$, initial iterate $\Omega_{\mathrm{init}} = I \in \mathcal{K}_{\mathbf{p}}^\sharp$.

\WHILE{ not converged}

\STATE Compute the subspace gradient $\mathrm{Proj}_{\mathcal{K}_{\mathbf{p}}}\left(\Omega^{-1}_t\right) = G_1^t \oplus \dots \oplus G_K^t$.
\STATE \emph{Line search}: Let stepsize $\zeta_t$ be the largest element of $\{c^j \zeta_{t,0}\}_{j =1,\dots}$ such that the following are satisfied for $\Psi_k^{t+1} = \mathrm{shrink}^-_{\zeta_t \rho}(\Psi_k^t - \zeta_t (\tilde{S}_k - G_k^t+q'_{\rho}(\Psi_k)))$:
\begin{enumerate}
\item $\|\Psi_1^{t+1} \oplus \dots \oplus \Psi_K^{t+1}\|_2 \leq \kappa$,
\item $ \Psi_1^{t+1} \oplus \dots \oplus \Psi_K^{t+1}\succ 0$, 
\item $f(\{\Psi_k^{t+1}\}) \leq \mathcal{Q}_{\zeta_t}(\{\Psi_k^{t+1}\}, \{\Psi_k^{t+1}\})$.
\end{enumerate}

\FOR{$k = 1,\dots,K$}

\STATE \emph{Composite objective gradient update}: 
\[
\Psi_k^{t+1} \gets \mathrm{shrink}^-_{\zeta_t \rho}\left(\Psi_k^t - \zeta_t (\tilde{S}_k - G^t_k + q'_{\rho}(\Psi_k)) \right).
\]
\ENDFOR
\STATE Compute next Barzilai-Borwein stepsize $\zeta_{t+1,0}$ via \eqref{eq:BarBor} in supplement \ref{supp:stepsz}.

\ENDWHILE
\STATE Return $\{\Psi_k^{t+1}\}_{k =1 }^K$.

\end{algorithmic}
\end{algorithm}

\subsection{Choice of step size $\zeta_t$}
\label{supp:stepsz}

Here we propose a method \eqref{eq:stpq} for selecting the stepsize parameter $\zeta_t$ at each step $t$ that ensures convergence of the algorithm. We follow the approach of \cite{beck2009fast} and \cite{GISTA}. Since $\Omega_t \succ 0$ and the the positive definite cone is an open set, there will always exist a $\zeta_t$ small enough such that $\Omega_{t+1} \succ 0$. We prove geometric convergence when $\zeta_t$ is chosen such that $\Omega_{t+1} \succ 0$ and
\begin{align}
\label{eq:stpq}
f(\Omega_{t+1}) = - \log|\Omega_{t+1}| + \langle \hat{S}, \Omega_{t+1}\rangle \leq \mathcal{Q}_{\zeta_t}(\Omega_{t+1},\Omega_t)
\end{align}
where $\mathcal{Q}_{\zeta_t}$ is a quadratic approximation to $f$ given by
\begin{align}\label{eq:qq}
\mathcal{Q}_{\zeta_t}& (\Omega_{t+1},\Omega_t) 
\\\nonumber&= - \log |\Omega_t| + \langle \hat{S},\Omega_t\rangle + \langle \Omega_{t+1} - \Omega_t,\nabla f(\Omega_t)\rangle + \frac{1}{2\zeta_t}\|\Omega_{t+1} - \Omega_t\|_F^2.
\end{align}

At each iteration $t$, we thus perform a line search to select an appropriate $\zeta_t$. We first select an initial stepsize $\zeta_{t,0}$ and compute the update \eqref{eq:upp}. If the resulting $\Omega_{t+1}$ is not positive definite or does not decrease the objective sufficiently according to \eqref{eq:stpq}, we decrease the stepsize $\zeta_t$ to $c\zeta_{t,0}$ for $c \in (0,1)$ and re-evaluate if the resulting $\Omega_{t+1}$ satisfies the conditions. This backtracking process is repeated (setting stepsize equal to $c^j \zeta_{t,0}$ where $j$ is incremented) until the resulting $\Omega_{t+1}$ satisfies the conditions. Since by construction $\Omega_t$ is positive definite, and the positive definite cone is an open set, there will be a step size small enough such that the conditions are satisfied. In practice, if after a set number of backtracking steps the conditions are still not satisfied, we can always take the safe step
\[
\zeta_t = \lambda_{\min}^2(\Omega_t) = \sum_{k=1}^K \min_i [\mathbf{s}_k]_i^2.
\]

As the safe stepsize often leads to slower convergence, we use the more aggressive Barzilai-Borwein step to set a starting $\zeta_{t,0}$ at each time. 
The Barzilai-Borwein stepsize presented in \cite{barzilai1988two} creates an approximation to the Hessian, in our case given by 
\begin{align}
\label{eq:BarBor}
\zeta_{t+1,0} &= \frac{\|\Omega_{t+1} - \Omega_t\|_F^2}{\langle \Omega_{t+1} - \Omega_t,\nabla f(\Omega_t) - \nabla f(\Omega_{t+1})\rangle} 
\end{align}
We derive the gradient $\nabla f(\Omega_t)$ in the next section. The norms and inner products in \eqref{eq:BarBor} and \eqref{eq:qq} can be efficiently computed factorwise (using the $\Psi_k$ and $S_k$ only) using the formulas in Appendix \ref{App:BasicProp}.
\subsection{Generation of Kronecker Sum Random Tensors}\label{App:Alg}
Generating random tensors given a Kronecker sum precision matrix can be made efficient by exploiting the Kronecker sum eigenstructure.
Algorithm \ref{alg:Gen} allows efficient generation of data following the TeraLasso model.
\begin{algorithm}[h]
\caption{Generation of subgaussian tensor $X \in \mathbb{R}^{d_1 \times \dots \times d_K}$ under TeraLasso model. }
\label{alg:Gen}
\begin{algorithmic}[1]
\STATE Assume $\Sigma^{-1} = \Psi_1 \oplus \dots \oplus \Psi_K$.
\STATE Input precision matrix factors $\Psi_k \in \mathbb{R}^{d_k \times d_k}$, $k = 1,\dots, K$.
\FOR{$k = 1,\dots, K$}
\STATE $U_k, \Lambda_k \gets \mathrm{EIG}(\Psi_k)$ eigendecomposition of $\Psi_k$.
\ENDFOR
\STATE $\mathbf{v} = [v_1,\dots, v_p] \gets \mathrm{diag}(\Lambda_1) \oplus \dots \oplus \mathrm{diag}(\Lambda_K) \in \mathbb{R}^p$.
\STATE Generate isotropic subgaussian random vector $z \in \mathbb{R}^{p}$.
\STATE $\tilde{x}_i \gets v_i^{-1/2} z_i$, for $i  = 1,\dots, p$.
\FOR{$k = 1,\dots,K$}
\STATE $\tilde{\mathbf{x}} \gets (I_{[d_{1:k-1}]} \otimes U_k \otimes I_{[d_{k+1:K}]})\tilde{
\mathbf{x}}$.
\ENDFOR
\STATE Reshape $\tilde{x}$ into $X \in \mathbb{R}^{d_1 \times \dots \times d_K}$.
\end{algorithmic}
\end{algorithm}
\subsection{Detailed TeraLasso Algorithm}
Algorithm \ref{alg:lowlevel} shows additional details of the implementation of Algorithm \ref{alg:highlevel} in the main text.
\begin{algorithm}[h]
\caption{TG-ISTA Implementation of TeraLasso (Detailed)}
\label{alg:lowlevel}
\begin{algorithmic}[1]
\STATE Input: SCM factors $S_k$, regularization parameters $\rho_i$, backtracking constant $c \in (0,1)$, initial step size $\zeta_{1,0}$, initial iterate $\Omega_{\mathrm{init}} = \Psi_1^0\oplus \dots \oplus \Psi_K^0$.

\FOR{$k = 1,\dots, K$}
\STATE $\mathbf{s}_k$, $U_k \gets$ Eigen-decomposition of $\Psi_k^0 = U_k \mathrm{diag}(\mathbf{s}_k) U_k^T$.
\STATE $\tilde{S}_k \gets S_k - I_{d_k} \frac{\mathrm{tr}(S_k)}{d_k}\frac{K-1}{K}$.
\ENDFOR
\WHILE{ not converged}

\STATE $\{\tilde{\mathbf{s}}\}_{k = 1}^K \gets \mathrm{Proj}_{{\mathcal{K}}_\mathbf{p}}\left(\mathrm{diag}\left(\frac{1}{\mathbf{s}_1 \oplus \dots \oplus \mathbf{s}_K}\right)\right)$.
\FOR{$k = 1 \dots K$}
\STATE $G_k^t \gets U_k \mathrm{diag}(\mathbf{s}_k) U_k^T$.
\ENDFOR
\FOR{$j = 0,1,\dots$}
\STATE $\zeta_t \gets c^j \zeta_{t,0}$.
\FOR{$k = 1,\dots,K$}
\STATE $\Psi_k^{t+1} \gets \mathrm{shrink}^-_{\zeta_t \rho_k}(\Psi_k^t - \zeta_t (\tilde{S}_k - G_k^t))$. 
\STATE Compute eigen-decomposition $U_k \mathrm{diag}(\mathbf{s}_k) U_k^T = \Psi_k^{t+1}$.

\ENDFOR
\STATE Compute $\mathcal{Q}_{\zeta_t}(\{\Psi_k^{t+1} \},\{\Psi_k^{t} \})$ via \eqref{eq:qq}.
\IF{ $f(\{\Psi_k^{t+1} \}) \leq \mathcal{Q}_{\zeta_t}(\{\Psi_k^{t+1} \},\{\Psi_k^{t} \})$ as in \eqref{eq:qq} \AND $\min_i([\mathbf{s}_1 \oplus \dots \oplus \mathbf{s}_K]_i) >0$ }
 \STATE Stepsize $\zeta_t$ is acceptable; \textbf{break}
\ENDIF
\ENDFOR
\STATE Compute Barzilai-Borwein stepsize $\zeta_{t+1,0}$ via \eqref{eq:BarBor}

\ENDWHILE
\STATE Return $\{\Psi_k^{t+1}\}_{k =1 }^K$.

\end{algorithmic}
\end{algorithm}

\subsection{Decomposition of Objective: Proof of Lemma \ref{lem:decomp}}
\label{app:decomp}
\begin{proofof2}

For simplicity of notation define $G_t$ to be the projection of $\Omega^{-1}$ onto the cone $\mathcal{K}_{\mathbf{p}}$ of positive definite Kronecker sum matrices: 
\[
G_t = G_1^t \oplus \dots \oplus G_K^t = \mathrm{Proj}_{\mathcal{K}_{\mathbf{p}}}(\Omega^{-1}_t).
\] 
Using this notation and substituting in \eqref{eq:allGrad} from the main text, the objective \eqref{eq:cmpgrd} becomes
\begin{align}\label{eq:Ogstp}
\Omega_{t+1} \in \arg \min_{\Omega \in \mathcal{K}_{\mathbf{p}}} \left\{ \frac{1}{2}\left\| \Omega - \left(\Omega_t - \zeta_t\left(\tilde{S} - G_t\right)\right) \right\|_F^2 + \zeta_t \sum_{k=1}^K m_k  \rho_k  |{\Psi}_k|_{1,\off} \right\}
\end{align}

Expanding out the Kronecker sums, for 
\begin{align*}
\Omega_t &= \Psi^t_1 \oplus \dots \oplus \Psi^t_K,\qquad
\Omega = \Psi_1 \oplus \dots \oplus \Psi_K,
\end{align*}
the Frobenius norm term in the objective \eqref{eq:Ogstp} 
can be decomposed into a sum of a diagonal portion and a factor-wise sum of the off diagonal portions. This holds by Property \ref{prop:disj} in Appendix A which states the off diagonal factors $\Psi_k^-$ have disjoint support in $\Omega$. Thus,
\begin{align*}
&\left\| \Omega - \left(\Omega_t - \zeta_t\left((\tilde{S}_1 - G_1^t) \oplus \dots \oplus (\tilde{S}_K - G_K^t)\right)\right) \right\|_F^2\\ &= 
\left\| \left(\Psi_1 - (\Psi_1^t - \zeta_t (\tilde{S}_1 - G_1^t))\right) \oplus \dots \oplus \left(\Psi_K - (\Psi_K^t - \zeta_t (\tilde{S}_K - G_K^t))\right) \right\|_F^2\\
&= \left\|\diag( \Omega) - \left(\diag(\Omega_t) - \zeta_t\diag\left(\tilde{S} - G_t\right)\right) \right\|_F^2\\ &+ \sum_{k = 1}^K m_k \left\| \offd\left(\Psi_1 - (\Psi_1^t - \zeta_t (\tilde{S}_1 - G_1^t))\right) \right\|_F^2.
\end{align*}
Substituting into the objective \eqref{eq:Ogstp}, we obtain
\begin{align*}
\Omega_{t+1} \in &
\arg \min_{\Omega \in \mathcal{K}_{\mathbf{p}}} \left\{ \frac{1}{2}\left\| \diag(\Omega) - \left(\diag(\Omega_t) - \zeta_t\diag\left(\tilde{S} - G_t]\right)\right) \right\|_F^2\right. \\ &+ \left. \sum_{k=1}^K m_k  \left(\frac{1}{2}\left\| \offd\left(\Psi_k - (\Psi_k^t - \zeta_t (\tilde{S}_k - G_k^t))\right) \right\|_F^2 + \zeta_t \rho_k  |{\Psi}_k|_{1,\off}\right) \right\}.
\end{align*}
This objective is decomposable into a sum of terms each involving either the diagonal $\Omega^+$ or one of the off diagonal factors $\Psi_k^-$. Thus, we can solve for each portion of $\Omega$ independently, giving 

\begin{align}\label{eq:DigO}
\offd(\Psi_k^{t+1}) &= \arg \min_{\offd(\Psi_k)} \frac{1}{2}\left\| \offd(\Psi_k) - \offd(\Psi_k^t - \zeta_t (\tilde{S}_k - G_k^t)) \right\|_F^2 + \zeta_t \rho_k  |{\Psi}_k|_{1,\off}\\\nonumber
\diag(\Omega_{t+1}) &= \arg \min_{\diag(\Omega)}\frac{1}{2}\left\| \diag(\Omega) - \diag\left(\Omega_t - \zeta_t\left(\tilde{S} - G_t\right)\right) \right\|_F^2.
\end{align}
Since the diagonal $\diag(\Omega)$ is not regularized in \eqref{eq:DigO}, we have
\[
\diag(\Omega_{t+1}) = \diag(\Omega_t) - \zeta_t\diag(\tilde{S} - G_t),
\]
i.e.
\begin{align}\label{eq:digg}
\diag(\Psi_k^{t+1}) = \diag(\Psi_k^{t}) - \zeta_t\diag(\tilde{S}_k - G_k^t).
\end{align}
This means we can equivalently obtain the solution of the problem \eqref{eq:DigO} by solving
\[
\Psi_k^{t+1} = \arg \min_{\Psi_k} \frac{1}{2}\left\| \Psi_k - (\Psi_k^t - \zeta_t (\tilde{S}_k - G_k^t)) \right\|_F^2 + \zeta_t \rho_k  |{\Psi}_k|_{1,\off},
\]
completing the proof.

\end{proofof2}

\subsection{Proof of Joint Convexity} 
\label{App:Conv}

Our objective function is 
\begin{align}
\label{Eq:objFunsupp}
Q(\{{\Psi}_k\}) &=  -\log | {\Psi}_1 \oplus \dots \oplus {\Psi}_K| + \langle \hat{S}, {\Psi}_1 \oplus\dots \oplus {\Psi}_K\rangle + \sum_k \rho_k d_k|{\Psi}_k|_{1,\off}.
\end{align}
We have the following theorem. 
This theorem proves the joint convexity of the objective function \eqref{Eq:objFunsupp} and the uniqueness of the minimizer $\hat{\Omega}$.
\begin{theorem}
\label{Thm:Conv}
The objective function \eqref{Eq:objFunsupp} is jointly convex in $\{{\Psi}_k\}_{k=1}^K$. Furthermore, define the set $\mathcal{A} = \{\{{\Psi}_k\}_{k=1}^K | Q(\{{\Psi}_k\}_{k=1}^K) = Q^*\}$ where the global minimum $Q^* = \min_{\{{\Psi}_k\}_{k=1}^K} Q(\{{\Psi}_k\}_{k=1}^K)$. There exists a unique ${\Omega}_* \in \mathcal{K}_{\mathbf{p}}^\sharp$, defined in \eqref{eq::kppintro}, that achieves the minimum of $Q$ such that
\begin{align}
{\Psi}_1 \oplus\dots \oplus {\Psi}_K = {\Omega}_* \quad \forall \: \{{\Psi}_k\}_{k=1}^K \in \mathcal{A}.
\end{align}
\end{theorem}

\begin{proof}

By definition,
\begin{align}
{\Psi}_1 \oplus \dots \oplus {\Psi}_K =&{\Psi}_1 \otimes {I}_{m_1} + \dots +  {I}_{m_K} \otimes {\Psi}_K
\end{align}
is an affine function of $\mathbf{z} = [\mathrm{vec}({\Psi}_1); \dots ; \mathrm{vec}({\Psi}_K)]$. Thus, since $\log |{A}|$ is a concave function on the space of positive definite matrices \citep{boyd2009convex}, all the terms of $Q$ are convex since convex functions of affine functions are convex and the elementwise $\ell_1$ norm is convex. Hence $Q$ is jointly convex in $\{{\Psi}_k\}_{k=1}^K$ on $\mathcal{K}_{\mathbf{p}}^\sharp$. Hence, every local minima is also global. Furthermore, for positive $\rho_k$ at least one global minimum must exist since $|\cdot|_1$ has a global minimum at zero.

We show that a nonempty set of $\{{\Psi}_k\}_{k=1}^K$ such that $Q(\{{\Psi}_k\}_{k=1}^K)$ is minimized maps to a unique ${\Omega} = {\Psi}_1 \oplus \dots \oplus {\Psi}_K$. If only one point $\{{\Psi}_k\}_{k=1}^K$ exists that achieves the global minimum, then the statement is proved. Otherwise, suppose that two distinct points $\{{\Psi}_{k,1}\}_{k=1}^K$ and $\{{\Psi}_{k,2}\}_{k=1}^K$ achieve the global minimum $Q^*$. Then, for all $k$ define
\begin{align}
{\Psi}_{k,\alpha} = \alpha{\Psi}_{k,1} + (1-\alpha) {\Psi}_{k,2}
\end{align}
By convexity, $Q(\{{\Psi}_{k,\alpha}\}_{k=1}^K) = Q^*$ for all $\alpha \in [0,1]$, i.e. $Q$ is constant along the specified affine line segment. This can only be true if (up to an additive constant) the first two terms of $Q$ are equal to the negative of the second two terms along the specified segment. Since 
\begin{align}
-\log |{A}| + \langle {\hat{S}}, {A}\rangle
\end{align}
is strictly convex and smooth on the positive definite cone (i.e. the second derivative along any line never vanishes) \citep{boyd2009convex} and the sum of the two elementwise $\ell$1 norms along any affine combination of variables is at most piecewise linear when smooth, this cannot hold when ${\Omega}_\alpha = {\Psi}_{1,\alpha} \oplus \cdots \oplus {\Psi}_{K,\alpha}$ varies with $\alpha$. Hence, ${\Omega}_\alpha$ must be a constant ${\Omega}^*$ with respect to $\alpha$. Thus, the minimizing ${\Omega}^*$ is unique and Theorem \ref{Thm:Conv} is established. 

\qed \end{proof}

\section{Additional experiments}
\label{supp:addlExp}

\subsection{Convergence of nonconvex regularization algorithm}

Figure \ref{Fig:ConvNon} illustrates the convergence of the nonconvex Algorithm \ref{alg:highlevelNC} (experiment described more thoroughly in the main text).

\begin{figure}[h]
\centering
\begin{subfigure}[b]{0.4\textwidth}
\centering
\includegraphics[width=1.8in]{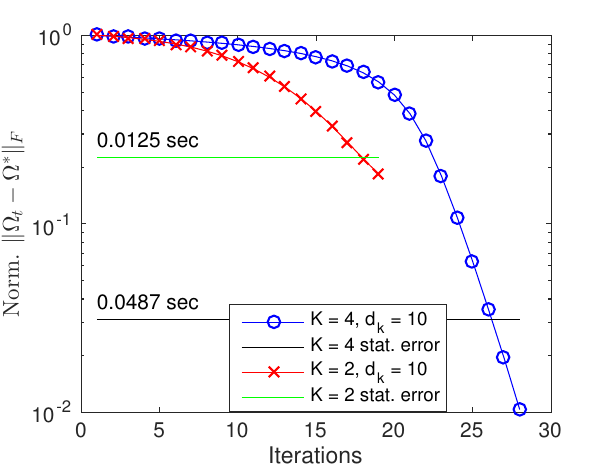}
\caption{SCAD penalty, $n=100$}
\end{subfigure}
\begin{subfigure}[b]{0.4\textwidth}
\centering
\includegraphics[width=1.8in]{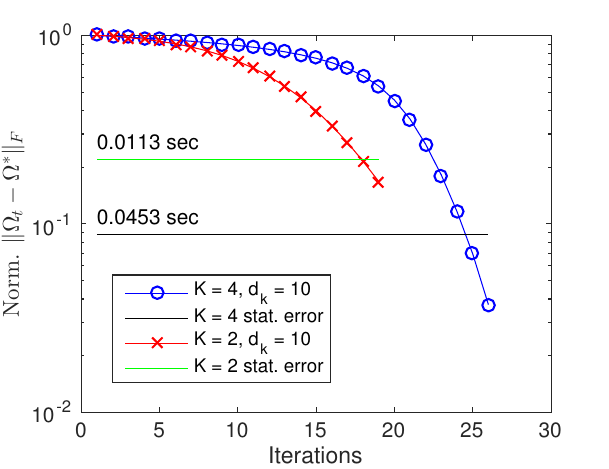}
\caption{MCP penalty, $n=100$}
\end{subfigure}
\vfill
\begin{subfigure}[b]{0.4\textwidth}
\centering
\includegraphics[width=1.8in]{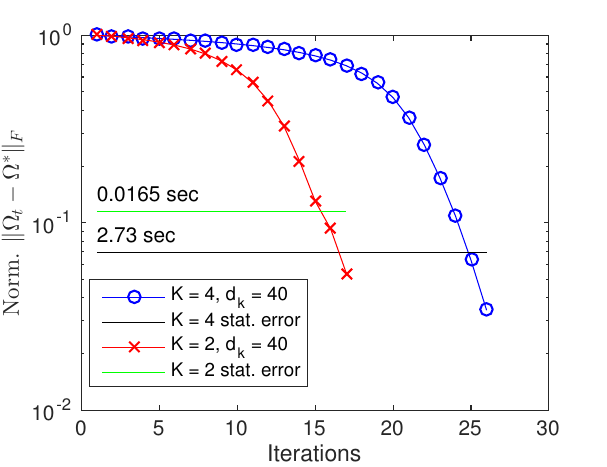}
\caption{SCAD penalty, $n=1$}
\end{subfigure}
\begin{subfigure}[b]{0.4\textwidth}
\centering
\includegraphics[width=1.8in]{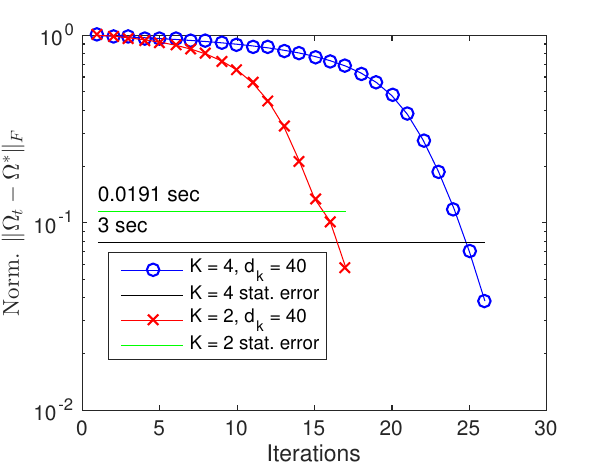}
\caption{MCP penalty, $n=1$}
\end{subfigure}
\caption{Geometric convergence of the nonconvex TG-ISTA implementation of TeraLasso. Shown is the normalized Frobenius norm $\|{\Omega}_t - \Omega^*\|_F$ of the difference between the estimate at the $t$th iteration and the optimal $\Omega^*$. On the left are results comparing $K = 2$ and $K=4$ on the same data with the same value of $p$ (different $d_k$), on the right they are compared for the same value of $d_k$ (different $p$). Also included are the statistical error levels, and the computation times required to reach them. Observe the consistent and rapid linear convergence rate, with logarithmic dependence on $K$ and dimension $d_k$.  }
\label{Fig:ConvNon}
\end{figure}

\subsection{Computational Complexity of TG-ISTA}
\label{supp:compcomp}
In Section \ref{supp:conv}, we show that TG-ISTA reaches the statistical error floor in 
\begin{align*}
T = O_p\left(\frac{2 \log K + \log(s+p) + \log \log p - \log (n \min_k m_k)}{\log \left(1 - \frac{2}{1 + K^2}\right)}\right)
\end{align*}
iterations. 

Each TG-ISTA iteration is also computationally efficient. Due to the representation \eqref{Eq:objFun}, the TG-ISTA implementation of TeraLasso never needs to form the full $p \times p$ covariance. The memory footprint of the proposed implementation is $O(p + \sum_{k = 1}^K d_k^2)$ as opposed to the $O(p^2)$ storage required by BiGLasso and GLasso. Since the training data itself requires $O(n p)$ storage, the storage footprint of the TG-ISTA implementation of TeraLasso is scalable to large values of $p = \prod_{k = 1}^K d_k$ when the $d_k/p$ decrease in $p$, e.g. $d_k = p^{1/K}$.
The computational cost per iteration is dominated by the computation of the gradient, which is performed by doing $K$ eigendecompositions of size $d_1,\dots,d_K$ respectively and then computing the projection of the inverse of the Kronecker sum of the resulting eigenvalues. The former step costs $O(\sum_{k=1}^K d_k^3)$, and the second step costs $O(pK)$, giving a cost per iteration of $
O\left(pK + \sum_{k=1}^K d_k^3\right)$.
For $K > 1$ and $d_k/p \ll 1$, this gives a dramatic improvement on the $O(p^3) = O(\prod_{k=1}^K d_k^3)$ cost per iteration of unstructured Graphical Lasso algorithms \citep{GISTA,hsieh2014quic}. In addition, for $K \leq 3$ the cost per iteration is comparable to the $O(d_1^3 + d_2^3 + d_3^3)$ cost per iteration of the most efficient ($K=3$) Kronecker product GLasso methods such as \cite{zhou2014gemini}.

Figure \ref{Fig:ConvBIG} shows convergence speeds on various random ER graph estimation scenarios, with the BiGLasso of \cite{kalaitzis2013bigraphical} shown for comparison. Note that the BiGLasso algorithm only applies when the diagonal elements of $\Omega$ are known, so it cannot be considered to solve the general BiGLasso or TeraLasso objectives. Observe that TeraLasso's ability to efficiently exploit the Kronecker sum structure to obtain computational and memory savings allows it to quickly converge to the optimal solution, while the alternating-minimization based BiGLasso algorithm is impractically slow. All computation was timed on a 4-core, 64 bit, 2.5GHz CPU system using Matlab 2016b.
\begin{figure}[h]
\centering
\scriptsize{
\begin{tabular}{c|*{5}{c}}
K              & $p$ & $d_k$ & $n$ & TeraLasso Runtime (s)  & BiGLasso Runtime (s)  \\
\hline
2 & $100$ & 10 & 10 & .0131 & .84\\
2 & $625$ & 25 & 10 & .0147 & 6.81\\
2 & $2500$ & 50 & 10 & .0272 & 161\\
2 & $5625$ & 75 & 10 & .0401 & 1690\\
2 & $10^4$ & 100 & 10 & .0664 &  \\
2 & $2.5\times 10^5$ & 500 & 10 & 1.62 &   \\
2 & $10^6$ & 1000 & 10 & 23.2 &   \\
2 & $4\times 10^6$ & 2000 & 10 & 427 &   \\
\hline
3 & $10^6$ & 100 & 10 & 3.52 & NA  \\
3 & $8\times 10^6$ & 200 & 10 & 11.2 & NA  \\
3 & $1.25 \times 10^8$ & 500 & 10 & 32.6 & NA  \\
3 & $1\times 10^9$ & 1000& 10 & 70.0 & NA\\
\hline
4 & $10^4$ & 10 & 10 & .281 & NA  \\
4 & $1.6\times 10^5$ & 20 & 10 & .649 & NA  \\
4 & $6.25 \times 10^6$ & 50 & 10 & 10.8 & NA  \\
4 & $1.00 \times 10^9$ & 178 & 10 & 88.4 & NA\\
\hline
5 & $1.16 \times 10^9$ & 65 & 10 & 124 & NA\\
\end{tabular}
}
\caption{Run times for the BiGLasso algorithm \citep{kalaitzis2013bigraphical} and the proposed TG-ISTA on a $K=2$ Kronecker sum model where the ground-truth edge topology follows a Kronecker sum Erd\"{o}s-R\'{e}nyi graphs for various values of the total dimension $p = d_1 d_2$ with $d_1  = d_2$. Also shown are TeraLasso results for $K=3, 4, 5$, for which BiGLasso is not applicable. 
Note the $10^2$ - $10^4$ magnitude speed up of TeraLasso (increasing with $p$), allowing estimation of billion-variable covariances ($10^{18}$ elements). 
}
\label{Fig:ConvBIG}
\end{figure}

\subsection{Convergence rate verification}
\label{supp:convVerif}
\begin{figure}[h]
\centering
\includegraphics[width=4in]{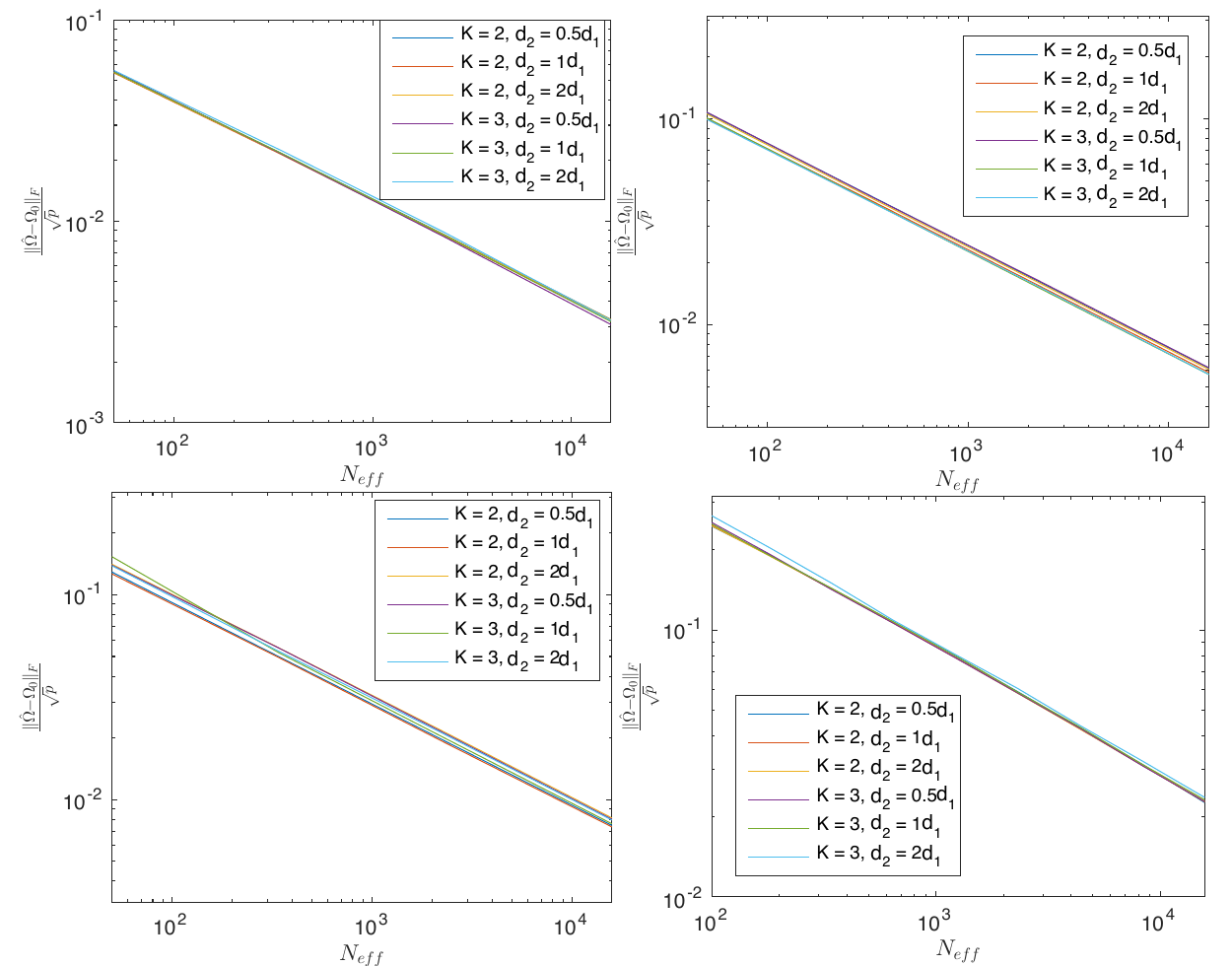}
\caption{
Frobenius norm convergence rate for the proposed TeraLasso. Shown (ordered by increasing difficulty) are results for AR graphs with $d_1 = 40$ (top left), random ER graphs with $d_1 = 10$ (top right), $d_1 = 40$ (bottom left), and random grid graphs with $d_1 = 36$ (bottom right). For each covariance model, 6 different combinations of $d_2$ and $K$ are considered, and the resulting Frobenius error is plotted versus the effective sample size $N_{\mathrm{eff}}$ \eqref{eq:Neff}.
}
\label{Fig:Rates}
\end{figure}

In this section, we verify that our bounds on the rate of convergence are tight in the case of $\ell$1 regularization. 
We will hold $\|\Sigma_0\|_2$ and $s/p$ constant. We set $\rho_k$ as in Theorem \ref{Thm:1}. 
By Lemma \ref{lem:DO} in the supplement, this implies an ``effective sample size" proportional to the inverse of the bound on $\|\hat{\Omega} - \Omega_0\|_F^2/p$:
\begin{align}
\label{eq:Neff}
N_{\mathrm{eff}} \propto (\log p)^{-1}n \min_k m_k. 
\end{align}



For each experiment below, we varied $K$ and $d_2$ over 6 scenarios. To ensure that the constants in the bound were minimally affected, we held $\Psi_1$ constant over all $(K,d_2)$ scenarios, and let $\Psi_3 = 0$ and $d_3 = d_1$ when $K=3$. 
We let $d_2$ vary by powers of 2, i.e. $d_2(c_d) = 2^{c_d} d_{2,\mathrm{base}}$ where $d_{2,\mathrm{base}}$ is a constant, allowing us to create a fixed matrix $B$ and set $\Psi_2 = I_{d_2/d_{2,\mathrm{base}}} \otimes B$ to ensure the eigenvalues of $\Psi_2$ and thus $\|\Sigma_0\|_2$ remain unaffected as $d_2$ ($c_d$) changes. 

Results averaged over random training data realizations are shown in Figure \ref{Fig:Rates} for ER ($d_k/2$ edges per factor), random grid ($d_k/2$ edges per factor), and AR-1 graphs (AR parameter $.5$ for both factors). Observe that in each case, the curves for all scenarios are very close despite the wide variation in dimensionality, indicating that our bound on the rate of convergence in Frobenius norm is tight. 

\subsection{Additional details for wind speed data experiments}
\label{supp:fit}
For the wind speed data example in the main text, we first regressed out the mean for each day in the year via a 14-th order polynomial regression on the entire history from 1948-2015. 
As in the main text, we extracted two $20\times 10$ spatial grids, one from eastern North America, and one from Western North America, with the latter including an expansive high-elevation area and both Atlantic and Pacific oceans (Figure \ref{Fig:Nation}). 
We compare the TeraLasso estimator to the unstructured shrinkage
estimator, the non-sparse Kronecker sum estimator (TeraLasso estimator
with sparsity parameter $\rho = 0$), 
and the Gemini sparse Kronecker product estimator of
\cite{zhou2014gemini}. 
Figure \ref{Fig:WA} shows the estimated precision matrices trained on the eastern grid, using time samples from January in $n$ years following 1948. 
Note the graphical structure reflects approximate auto-regressive (AR) spatial and temporal structure in each dimension. The TeraLasso estimation is much more stable than the Kronecker product estimation for small sample size $n$. 

\begin{figure}[h]
\centering
\includegraphics[width=5.1in]{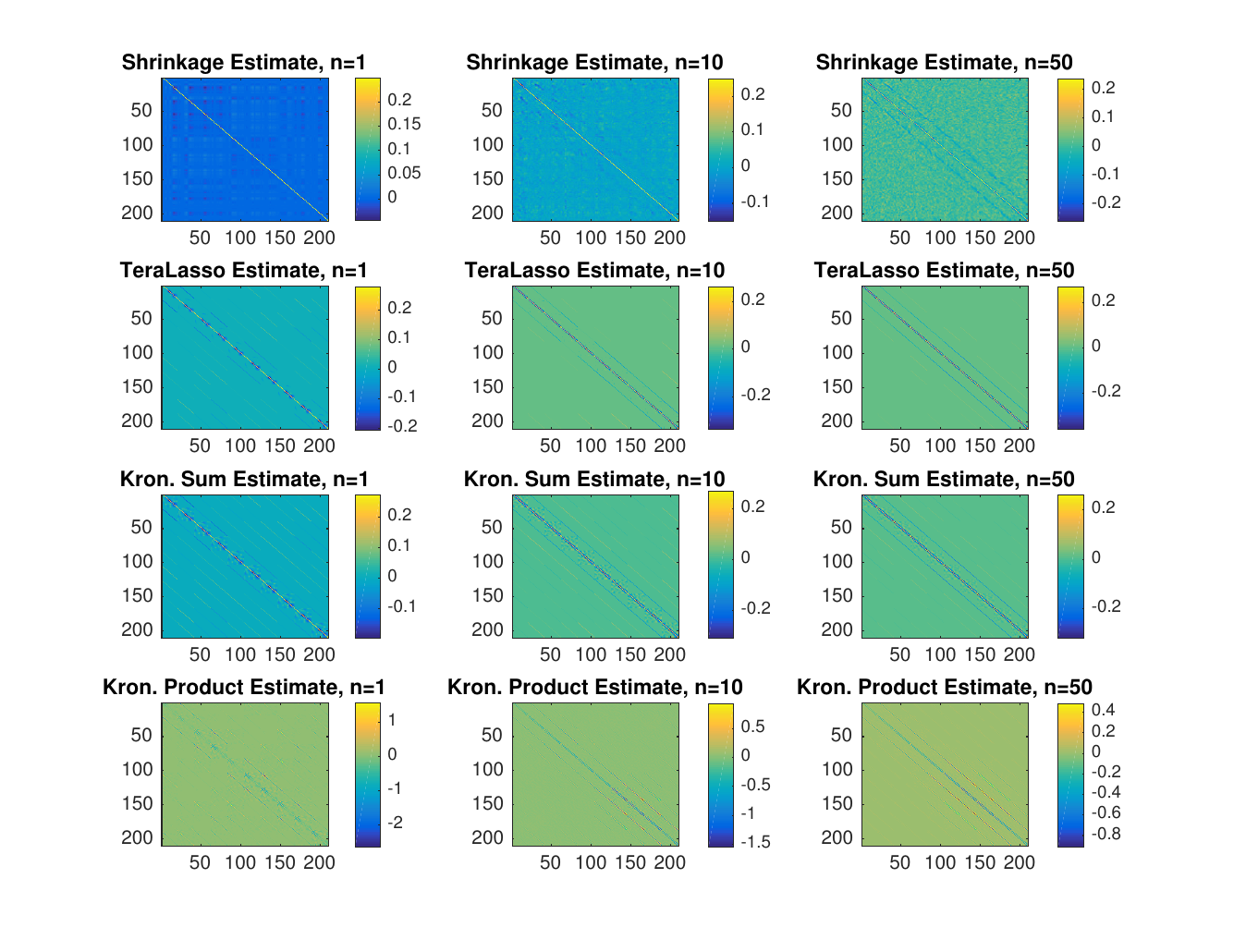}
\caption{Windspeed data, eastern grid. Spatial ($K=2$) precision matrix estimation, comparing TeraLasso to unstructured and sparse Kronecker product (Gemini) techniques, using $n=1$, $10$, and $50$. Observe the increasing sparsity and structure with increasing $n$, and TeraLasso's consistent structure even from one sample up to $n=50$. For improved contrast, the diagonal elements have been reduced in the plot. }
\label{Fig:WA}
\end{figure}

To quantify the fit of the estimated precision matrices to the observed wind data, we compare to an unstructured estimator in a higher sample regime. After training each estimated precision matrix (TeraLasso, Gemini, and ML Kronecker Product) on a 30-day summer interval from 1 year, as in the main experiment, we create a sample covariance $\hat{S}_{\mathrm{test}}$ from the same 30-day summer intervals in the remaining 50 years. We evaluate the precision matrices estimated by TeraLasso, Gemini, and ML Kronecker product using a normalized Frobenius error metric:
\[
\arg \min_{\delta \in [0, 1]} \|\hat{\Omega} - (\hat{S}_{\mathrm{test}} + \delta I_p)^{-1}\|_F/\|(\hat{S}_{\mathrm{test}} + \delta I_p)^{-1}\|_F.
\]
If this metric is small, the structured $\hat{\Omega}$ is close to the unstructured $(\hat{S}_{\mathrm{test}} + \delta I_p)^{-1}$, indicating a good fit to the data. The small ridge $\delta$ is included to ensure that the unstructured inverse estimator $(\hat{S}_{\mathrm{test}} + \delta I_p)^{-1}$ is well-conditioned, with the minimum taken over $\delta$ to present the most optimistic view of Gemini and the ML Kronecker product. 
The results for each precision matrix are TeraLasso: 0.0728, Gemini: 0.903, and ML Kronecker Product: 0.76, confirming the superior performance of the TeraLasso estimator.

\subsection{Comparison between TeraLasso and Gemini (Kronecker product) log determinant geometry}
\label{supp:logdet}

In this section, we present further analysis of the relation of the performance of TeraLasso in this wind data setting to its inherently more robust eigenstructure.

Recall the $\ell 1$ TeraLasso objective
\begin{equation}
\label{eq:tlobj}
-\log |\Psi_1 \oplus \dots \oplus \Psi_K| + \langle \hat{S}, \Psi_1 \oplus \dots \oplus \Psi_K\rangle + \sum_{k=1}^K \rho_k m_k |{\Psi}_k|_{1,\off}.
\end{equation}
where $m_k=p/d_k$. The Gemini Kronecker product algorithm \cite{zhou2014gemini} uses a similar objective function to estimate the Kronecker product covariance, which can be shown to be equivalent to
\begin{equation}
\label{eq:gkobj}
-\log |\Psi_1 \otimes  \Psi_2| + \langle \hat{S}, \Psi_1 \oplus  \Psi_2\rangle + \sum_{k=1}^2 \rho_k m_k |{\Psi}_k|_{1,\off}.
\end{equation}
Observe that, for $K=2$, the Gemini objective function \eqref{eq:gkobj} is the same as in TeraLasso objective function \eqref{eq:tlobj} except for the log determinant term. Figure \ref{fig:misp} (a) compares the Kronecker product Gemini estimator to TeraLasso on data generated using precision matrix $\Psi_1 \oplus \Psi_2$, and again on data generated using the Kronecker sum precision matrix $\Psi_1 \otimes \Psi_2$, where $\Psi_1, \Psi_2$ are each $10\times 10$ random ER graphs (generated as in the main text) with 5 nonzero edges. In all cases, we used the theoretically dictated optimal $\ell_1$ penalty for TeraLasso from Theorem \ref{Thm:1} in the main text and for Gemini from Theorem 3.1 in \cite{zhou2014gemini}. Note that both methods perform well in the single sample regime, even under model misspecification. This apparent symmetricity is very different from the relation of the ML Kronecker sum (TeraLasso with zero penalty) and the ML Kronecker product (not directly related to Gemini), whose results on the same data are also shown in Figure \ref{fig:misp} (b). In this case, the ML Kronecker product performs poorly in the single sample regime, whereas the ML Kronecker sum performs well in all regimes, surpassing the ML Kronecker product method in the low sample regime even when the data is generated under the Kronecker product model. 

This seems to indicate that the Gemini estimator leverages some of the inherent stability of the ML Kronecker sum objective (TeraLasso) to solve the more unstable Kronecker product covariance estimation problem.
\begin{figure}[h]
\centering
\begin{subfigure}[b]{\textwidth}
\centering
\includegraphics[width=1.83in]{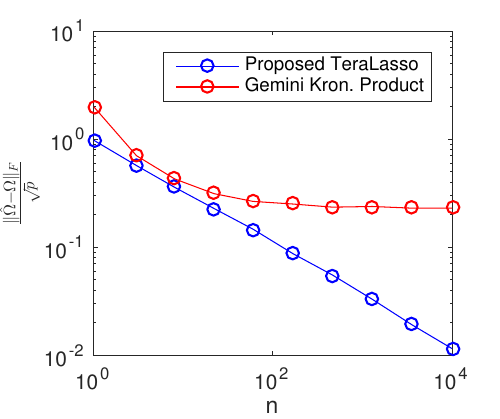}\includegraphics[width=2in]{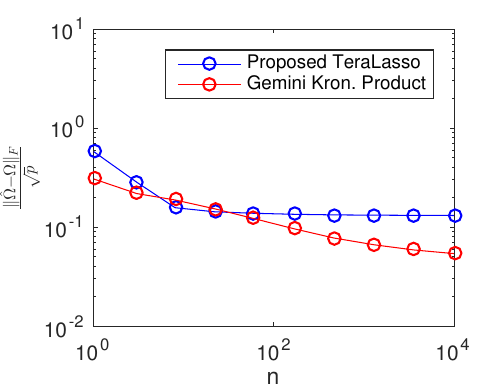}
\caption{TeraLasso (proposed Kronecker sum) and Gemini (Kronecker product) estimators, using optimal $\ell_1$ penalties, under model misspecification. Note the largely symmetric performance under model misspecification (TeraLasso on right, Gemini on left).}
\end{subfigure}
\hfill
\begin{subfigure}[b]{\textwidth}
\centering
\includegraphics[width=2in]{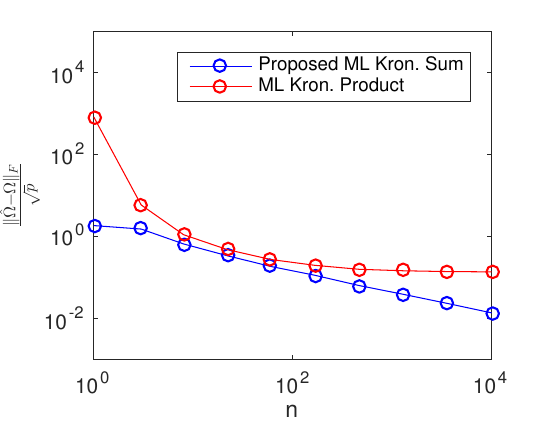}\includegraphics[width=2in]{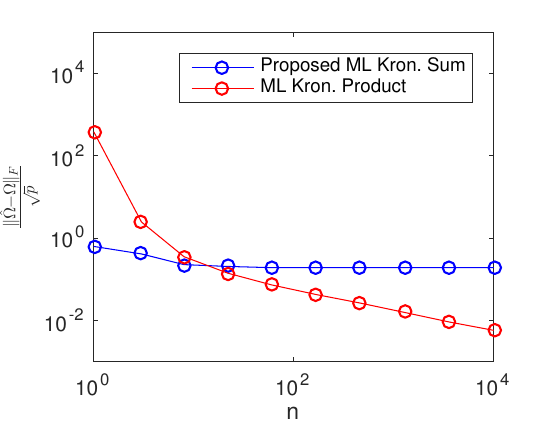}
\caption{Gaussian maximum likelihood estimators under model misspecification. Note the significant low-sample advantage of our proposed ML Kronecker Sum estimator even under model misspecification (right).}
\end{subfigure}
\caption{Kronecker sum and Kronecker product estimators under model misspecification. Left-hand plots were generated using Kronecker sum precision matrix $\Omega = \Psi_1 \oplus \Psi_2$, and right-hand plots were generated using Kronecker product precision matrix $\Omega = \Psi_1 \otimes \Psi_2$.}
\label{fig:misp}
\end{figure}

\begin{figure}[h]
\centering
\begin{subfigure}[b]{\textwidth}
\centering
\includegraphics[width=2in]{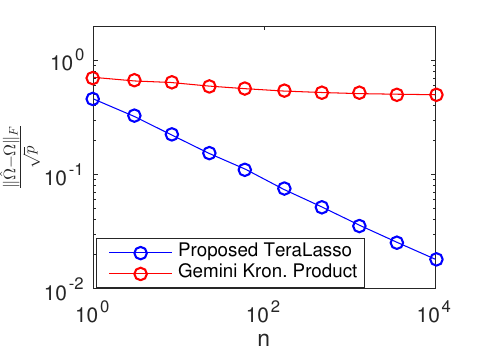}\includegraphics[width=2in]{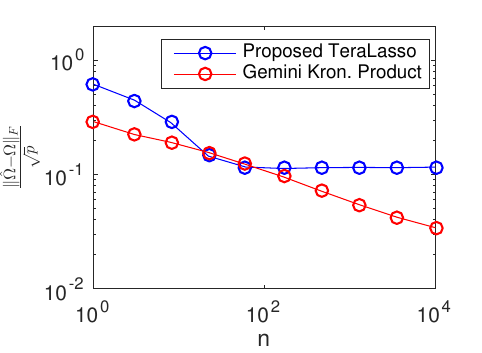}
\caption{TeraLasso (proposed Kronecker sum) and Gemini (Kronecker product) estimators, using optimal $\ell_1$ penalties, under model misspecification. Note the largely symmetric performance under model misspecification (TeraLasso on right, Gemini on left).}
\end{subfigure}
\hfill
\begin{subfigure}[b]{\textwidth}
\centering
\includegraphics[width=2in]{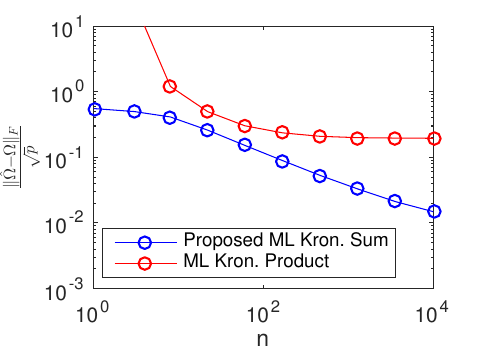}\includegraphics[width=2in]{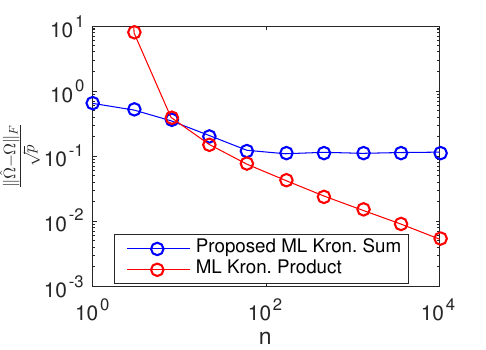}
\caption{Gaussian maximum likelihood estimators under model misspecification. Note the significant low-sample advantage of our proposed ML Kronecker Sum estimator even under model misspecification (right).}
\end{subfigure}
\caption{Kronecker sum and Kronecker product estimators under model misspecification, using the wind data Kronecker sum precision matrix $\Omega = \Psi_1 \oplus \Psi_2$ shown in Figure \ref{Fig:WFac} (a). Left-hand plots were generated using Kronecker sum precision matrix $\Omega = \Psi_1 \oplus \Psi_2$, and right-hand plots were generated using Kronecker product precision matrix $\Omega = \Psi_1 \otimes \Psi_2$.}
%
\label{fig:misp2}
\end{figure}

To further illuminate the connection between TeraLasso and Gemini, we now examine the relationship of the geometry of the differing log determinant terms. Let the eigenvalues of $\Psi_k$ be denoted as $\lambda_{k, 1},\dots, \lambda_{k,d_k}$, and suppose that $\Psi_1 \oplus \dots \oplus \Psi_K \succ 0$ so we can assume all the $\lambda_{k,i} \geq 0$. Using the properties of determinants and the additivity of the eigenvalues in a Kronecker sum we can write
\[
\log | \Psi_1 \oplus \dots \oplus \Psi_K| = \sum_{i_1 = 1}^{d_1} \dots \sum_{i_K = 1}^{d_K} \log|\lambda_{1,i_1} + \dots + \lambda_{K,i_K}|.
\]
Observe that the partial derivative of the log determinant with respect to any one eigenvalue $\lambda_{k,i_k}$ is $\sum_{i_1,\dots,i_{k-1},i_{k+1},\dots i_K}1/|\lambda_{1,i_1} + \dots + \lambda_{K,i_K}| \leq m_k/|\lambda_{k,i_k}|$.

Correspondingly, the log determinant of a Kronecker product is
\[
\log | \Psi_1 \otimes \dots \otimes \Psi_K| =\sum_{k=1}^K m_k \sum_{i_k = 1}^{d_k} \log|\lambda_{k,i_k}|.
\]
Observe that the partial derivative of the log determinant with respect to any one eigenvalue $\lambda_{k,i_k}$ is $m_k/|\lambda_{k,i_k}|$. 

Thus, the geometry of the Kronecker sum log determinant term is significantly flatter than the Kronecker product log determinant, especially for larger $K$, indicating that the Kronecker sum estimator (TeraLasso) will enjoy more flexibility when matching the sample covariances than a Kronecker product method will. 

A parallel interpretation can be obtained by recalling that the Kronecker sum of two sparse graphs is significantly sparser than the Kronecker product of the same two graphs, as discussed in the introduction of the main text.


\section{Identifiable Parameterization of $\mathcal{K}_{\mathbf{p}}$}\label{app:DO}

Observe that for any scalar $c$ 
\[
A \oplus B = A \otimes I + I \otimes B = A \otimes I - cI  + cI + I\otimes B = (A - cI) \oplus (B+cI),
\]
and thus the trace of each factor is non-identifiable, and we can write
\begin{align}\label{eq:eqq}
\Psi_1 \oplus \dots \oplus \Psi_K &= ({{\Psi}}_1+ c_1 {I}_{d_1}) \oplus \dots \oplus ({{\Psi}}_K + c_K {I}_{d_K})\\\nonumber &= (\Psi_1 \oplus \dots \oplus \Psi_K) + \left(\sum_{k=1}^K c_k\right) I_{p}\\\nonumber
&= \Psi_1 \oplus \dots \oplus \Psi_K,
\end{align}
where $c_k$ are any scalars such that $\sum_{k=1}^K c_k = 0$. 

The following lemma addresses this trace ambiguity, and creates an orthogonal, identifiable decomposition of $\Omega$ into factors.

Based on the original parameterization 
\[
B = A_1 \oplus \dots \oplus A_K,
\]
we know that the number of degrees of freedom in $B$ is much smaller than the number of elements $p^2$. We thus seek a lower-dimensional parameterization of $B$. The Kronecker sum parameterization
is not identifiable on the diagonals, so we seek a representation of $B$ that is identifiable. In the main text, we noted that $\mathrm{diag}(B) + \mathrm{offd}(A_1) \oplus \dots \oplus \mathrm{offd}(A_K)$ is identifiable (where $\mathrm{offd(A)} = A - \mathrm{diag}(A)$), but $\mathrm{diag}(B)$ cannot be a parameter of the model since not all diagonal vectors can be expressed as a Kronecker sum. Hence while this diagonal-based decomposition is useful for stating identifiable factorwise error bounds, it is does not truly serve as a parameterization.
We show in Lemma \ref{lem:DO} that the space $\mathcal{K}_{\mathbf{p}}$ is linearly, identifiably, and orthogonally parameterized by the quantities $\left(\tau_B \in \mathbb{R}, \left\{\tilde{A}_k \in \{A \in \mathbb{R}^{d_k\times d_k}|\tr(A) \equiv 0 \}\right\}_{k=1}^K\right)$. Specifically,
\begin{lemma}
\label{lem:DO}
Let $B \in \mathcal{K}_{\mathbf{p}}$ and  $B = A_1 \oplus \dots \oplus A_K \in \mathcal{K}_{\mathbf{p}}$.
Then $B$ can be identifiably written as
\begin{equation}\label{eq:decompNew}
B = \tau_B I_p + (\tilde{A}_1 \oplus \dots \oplus \tilde{A}_K)
\end{equation}
where $\tr(\tilde{A}_k) \equiv 0$ and the identifiable parameters $(\tau_B, \{\tilde{A}_k\}_{k=1}^K)$ can be computed as
\ben
\label{eq::tildeA}
\tau_B &= \frac{\tr(B)}{p},\qquad \tilde{A}_k = A_k - \frac{\tr(A_k)}{d_k} I_{d_k}.
\een
By orthogonality, the Frobenius norm can be decomposed as
\begin{align*}
\|B\|_F^2 &=  p \tau_B^2 + \sum_{k=1}^K m_k \|\tilde{A}_k\|_F^2
\geq \sum_{k=1}^K m_k \left\|\frac{\tau_B}{K}I_{d_k}+\tilde{A}_k\right\|_F^2,
\end{align*}
noting that
\[
B = \left(\frac{\tau_B}{K}I_{d_1}+\tilde{A}_1\right) \oplus \dots \oplus \left(\frac{\tau_B}{K}I_{d_K}+\tilde{A}_K\right).
\]

\label{eq:constru}

\end{lemma}






\begin{proof}

\textbf{Part I: Identifiable Parameterization.}
Let $B \in {\mathcal K}_{\mathbf{p}}$. By definition, there exists $A_1,\dots, A_K$ such that
\begin{align*}
B = A_1 \oplus \dots \oplus A_K
&= \sum_{k = 1}^K I_{[d_{1:k-1}]} \otimes A_k \otimes I_{[d_{k+1:K}]}\\
&= \sum_{k = 1}^K \left(I_{[d_{1:k-1}]} \otimes (A_k - \tau_k I_{d_k}) \otimes I_{[d_{k+1:K}]} + \tau_k I_{p}\right)\\
&= \left(\sum_{k = 1}^K \tau_k\right)I_p + ((A_1 - \tau_1 I_{d_1}) \oplus \dots \oplus (A_K - \tau_K I_{d_K})).
\end{align*}
where $\tau_k = \tr(A_k)/d_k$. Observe that $\tr(A_k - \tau_k I_{d_k}) = 0$ by construction, so we can set $\tilde{A}_k = A_k - \tau_k I_{d_k}$, creating
\begin{align*}
B &= \left(\sum_{k = 1}^K \tau_k\right)I_p + (\tilde{A}_1 \oplus \dots \oplus \tilde{A}_K).
\end{align*}
Note that in this representation, $\tr(\tilde{A}_1 \oplus \dots \oplus \tilde{A}_K) = 0$, so letting $\tau_B = \tr(B)/p$,
\[
\tau_B = \sum_{k = 1}^K \tau_k,
\]
and \eqref{eq:decompNew} in the Lemma results. It is easy to verify any $B$ expressible in the form \eqref{eq:decompNew} is in ${\mathcal K}_{\mathbf{p}}$.


Thus, $(\tau_B, \{\tilde{A}_k\}_{k=1}^K)$ parameterizes ${\mathcal K}_{\mathbf{p}}$.
It remains to show that this parameterization is identifiable.

%


\textbf{Part II: Orthogonal Parameterization.}
We will show that under the linear parameterization of $\mathcal{K}_{\mathbf{p}}$ by $(\tau_B, \{\tilde{A}_k\}_{k=1}^K)$, each of the $K+1$ components are linearly independent of the others. 

To see this, we compute the inner products between the components:
\begin{align*}
\langle \tau_B I_p,  I_{[d_{1:k-1}]} \otimes \tilde{A}_k \otimes& I_{[d_{k+1:K}]}\rangle = \tau_B m_k \tr(\tilde{A}_k) \equiv 0\\
\langle I_{[d_{1:k-1}]} \otimes \tilde{A}_k \otimes I_{[d_{k+1:K}]},\:&I_{[d_{1:\ell-1}]} \otimes \tilde{A}_\ell \otimes I_{[d_{\ell +1:K}]} \rangle \\
&= \tr\left( I_{[d_{1:k-1}]} \otimes \tilde{A}_k \otimes I_{[d_{k+1:\ell-1}]} \otimes \tilde{A}_\ell \otimes I_{[d_{\ell + 1:K}]}  \right) \\ &= \frac{p}{d_k d_\ell} \tr(\tilde{A}_k)\tr(\tilde{A}_\ell) \equiv 0,
\end{align*}
for all $k \neq \ell$. We have recalled that by definition, $\tr(\tilde{A}_k) \equiv 0$ for all $k$. Since all the inner products are identically zero, the components are orthogonal, thus they are linearly independent. Hence, by the definition of linear independence, this linear parameterization $(\tau_B, \{\tilde{A}_k\}_{k=1}^K)$ is uniquely determined by $B \in \mathcal{K}_{\mathbf{p}}$ (i.e. it is identifiable).


\textbf{Part III: Decomposition of Frobenius norm.}
Using the identifiability and orthogonality of this parameterization, we can find a direct factorwise decomposition of the Frobenius norm on $\mathcal{K}_{\mathbf{p}}$.

By orthogonality (cross term inner products equal to zero) 
\begin{align}\label{eq:ddecp}
\|B\|_F^2 &= \|\tau_B I_p\|_F^2 + \sum_{k = 1}^K \|I_{[d_{1:k-1}]} \otimes \tilde{A}_k \otimes I_{[d_{k+1:K}]}\|_F^2\\\nonumber
&= p \tau_B^2 + \sum_{k=1}^K m_k \|\tilde{A}_k\|_F^2.
\end{align}
This completes the first decomposition, representing the squared Frobenius norm as weighted sum of the squared Frobenius norms on each component.

For convenience, we also observe that given any $B \in \mathcal{K}_{\mathbf{p}}$ with identifiable parameterization
\[
B = \tau_B I_p +  (\tilde{A}_1 \oplus \dots \oplus \tilde{A}_K),
\]
we can absorb the scaled identity into the Kronecker sum and still bound the Frobenius norm decomposition. Specifically, observe that
\[
p\tau_B^2 = p K \sum_{k = 1}^K \left(\frac{\tau_B}{K}\right)^2 \geq p \sum_{k = 1}^K \left(\frac{\tau_B}{K}\right)^2.
\]
Substituting this into \eqref{eq:ddecp},
\begin{align*}
\|B\|_F^2 =  p \tau_B^2 + \sum_{k=1}^K m_k \|\tilde{A}_k\|_F^2
&\geq  p \sum_{k = 1}^K \left(\frac{\tau_B}{K}\right)^2 + \sum_{k=1}^K m_k \|\tilde{A}_k\|_F^2\\
&= \sum_{k=1}^K m_k \left(\left\|\frac{\tau_B}{K}I_{d_k}\right\|_F^2 + \|\tilde{A}_k\|_F^2\right)\\
&=\sum_{k=1}^K m_k \left\|\frac{\tau_B}{K}I_{d_k}+\tilde{A}_k\right\|_F^2,
\end{align*}
where the last term follows because $\tr(\tilde{A}_k) \equiv 0$ implies that $\langle I_{d_k},\tilde{A}_k\rangle \equiv 0$.

Observe that
\[
B = \left(\frac{\tau_B}{K}I_{d_1}+\tilde{A}_1\right) \oplus \dots \oplus \left(\frac{\tau_B}{K}I_{d_K}+\tilde{A}_K\right),
\]
hence Lemma \ref{lem:DO} is proved.

%

\qed \end{proof}

The identifiable parameterization of $\mathcal{K}_{\mathbf{p}}$ in Lemma \ref{lem:DO} will provide a way to bound the spectral norm relative to the Frobenius norm. This is used to form the spectral norm bound in Theorem \ref{Thm:1Spec}.

The following lemma is also used in the proof of Theorem~\ref{Thm:1} (cf. Proposition \ref{prop:posi-def-interval}).

\begin{lemma}[Spectral Norm Bound]
\label{cor:l2bd}
For all $B \in \mathcal{K}_{\mathbf{p}}$, 
\[
\|B\|_2 \leq \sqrt{\frac{K+1}{\min_k m_k}} \|B\|_F.
\]
\end{lemma}
\begin{proof}
Using the identifiable parameterization of $B$
\[
B = \tau_B I_p + (\tilde{A}_1 \oplus \dots \oplus \tilde{A}_K),
\]
and the triangle inequality, we have
\begin{align*}
\|B\|_2 &\leq | \tau_B| + \sum_{k=1}^K \|\tilde{A}_k\|_2\leq | \tau_B| + \sum_{k=1}^K \|\tilde{A}_k\|_F\leq \sqrt{K+1} \sqrt{\tau_B^2 + \sum_{k=1}^K \|\tilde{A}_k\|_F^2}\\
&\leq \sqrt{\frac{K+1}{\min_k m_k}} \sqrt{p \tau_B^2 + \sum_{k=1}^K m_k\|\tilde{A}_k\|_F^2}\\
&\leq \sqrt{\frac{K+1}{\min_k m_k}} \|B\|_F.
\end{align*}
\qed \end{proof}

\subsection{Inner Product in $\mathcal{K}_{\mathbf{p}}$}
\label{app:cnv}
\begin{lemma}[Kronecker sum inner Products]\label{lem:InnCov}
Suppose $B \in \mathbb{R}^{p \times p}$. Then for any $A_k \in \mathbb{R}^{d_k \times d_k}$, $k = 1, \dots, K$, 
\[
\langle B, A_1 \oplus\dots \oplus A_K\rangle = \sum_{k=1}^K m_k \langle B_k, A_k \rangle.
\]
\end{lemma}
\begin{proof}
\begin{align*}
\langle B, A_1 \oplus\dots \oplus A_K\rangle &= \sum_{k = 1}^K \langle B, I_{[d_{1:k-1}]} \otimes A_k \otimes I_{[d_{k+1:K}]}\rangle\\
&= \sum_{k = 1}^K \sum_{i = 1}^{m_k}  \langle B(i,i|k) ,  A_k \rangle\\
&= \sum_{k = 1}^K   \left\langle \sum_{i = 1}^{m_k} B(i,i|k) ,  A_k \right\rangle\\
&= \sum_{k = 1}^K m_k  \langle B_k ,  A_k \rangle.
\end{align*}
where we have used the definition of the submatrix notation $B(i,i|k)$ and the matrices $B_k = \frac{1}{m_k} \sum_{i = 1}^{m_k} B(i,i|k)$. See Appendix~\ref{App:Ident} for the notation being used here.
\qed \end{proof}

\newcommand{\ksump}{\mathcal{K}_{\mathbf{p}}}
\section{Proof of Theorems \ref{Thm:1} and \ref{Thm:1Spec} ($\ell 1$ regularized case)}
\label{app:Pf}

%

Let $\Omega_0$ be the true value of the precision matrix $\Omega$. 
Since $\Omega, \Omega_0 \in \mathcal{K}_{\mathbf{p}}$ and
$\mathcal{K}_{\mathbf{p}}$ is convex, $\Delta_\Omega = {\Omega} - \Omega_0 \in \mathcal{K}_{\mathbf{p}}$ and we can decompose $\Delta_\Omega$ into diagonal and Kronecker sum off diagonal components:
\begin{align}
\label{eq:DOM}
\Delta_\Omega &= {\Omega} - \Omega_0 
= \diag(\Delta_\Omega) + (\offd(\Delta_{\Psi,1}) \oplus \dots \oplus \offd(\Delta_{\Psi,K})),
\end{align}
where $\diag(\Delta_\Omega) = \diag({\Omega} - \Omega_0)$ and $\offd(\Delta_{\Psi,k}) = \offd({\Psi}_k - {\Psi}_{0,1}) $. Recall that the $\diag(\Delta_\Omega)$ and $\offd(\Delta_{\Psi,k})$ terms are all identifiable given $\Delta_\Omega \in \mathcal{K}_{\mathbf{p}}$. Similarly, we can write 
\begin{align*}
{\Omega} &=\diag(\Omega) + (\offd({\Psi}_1) \oplus \dots \oplus \offd({\Psi}_K))\\
{\Omega}_0 &= \diag(\Omega_0) + (\offd({\Psi}_{0,1}) \oplus \dots \oplus\offd( {\Psi}_{0,K})).
\end{align*}

Let $I(\cdot)$ be the indicator function. For an index set $\mathcal{A}$ and a matrix $M = [m_{ij}]$, define the operator $\mathcal{P}_{\mathcal{A}}(M) \equiv [m_{ij}I((i,j) \in \mathcal{A})]$ that projects $M$ onto the set $\mathcal{A}$. Let $\Delta_{k,S} = \mathcal{P}_{\mathcal{S}_k}(\offd(\Delta_{\Psi,k}))$ be the projection of $\offd(\Delta_{\Psi,k})$ onto the true sparsity pattern of $\Psi_k$. Let ${\mathcal S}_{k}^c$ be the complement of $\mathcal{S}_{k}$, and $\Delta_{k,S^c}=\mathcal{P}_{{\mathcal S}^c_{k}}(\offd(\Delta_{\Psi,k}))$. 
Furthermore, let
\bens
\Delta_{S} &= & (\Delta_{1,S} \oplus \dots \oplus \Delta_{K,S})
\text{ and } \; \; \\
\Delta_{S^c} & = & \Delta_{1,S^c} \oplus \dots \oplus \Delta_{K,S^c}
\eens
be the projection of $\Delta_\Omega$ onto the sparsity set
$\mathcal{S}$  and its complement. 
Recall neither $\mathcal{S}$ nor $\mathcal{S}^c$ includes the diagonal. 

We now provide a deterministic bound on the difference in the penalty terms.
\begin{lemma}
\label{lemma::geometry}
Denote by
\bens
\Delta_g := \sum_k \rho_k m_k (|\Psi_{k,0}+ \Delta_{\Psi,k}|_{1 ,\off}-
|\Psi_{k,0}|_{1,\off}),
\eens
Then
\ben
\label{eq::DDG}
\Delta_g &\geq & 
\sum_k \rho_k m_k  (\onenorm{\Delta_{k, \Sc}} -\onenorm{\Delta_{k,S}} )
\een
\end{lemma}

\begin{proofof}{Lemma~\ref{lemma::geometry}}
By the decomposability of the $\ell_1$ norm and the reverse triangle inequality $ |A + B|_1 \geq |A|_1 - |B|_1$, we have
\begin{align}\label{eq:decc}
|\Psi_{k,0} + \Delta_{\Psi,k}|_{1,\off} &- |\Psi_{k,0}|_{1,\off}\\\nonumber
&= |\Psi_{k,0} + \Delta_{k,S}|_{1,\off} +  | \Delta_{k,S^c}|_1 - |\Psi_{k,0}|_{1,\off}\\\nonumber
&\geq |\Psi_{k,0}|_{1,\off} - |\Delta_{k,S}|_1+  | \Delta_{k,S^c}|_1 - |\Psi_{k,0}|_{1,\off}\\\nonumber
&\geq |\Delta_{k, S^c}|_1 - |\Delta_{k,S}|_1
\end{align}
since $\Psi_{k,0}$ is assumed to follow sparsity pattern
$\mathcal{S}_k$ by (A1).
\end{proofof}

Let $\mathcal{A}_0$ be the event that for some constant $C_0$,
\begin{equation}\label{eq:a0}
 \frac{|\mathrm{tr}(\hat{S})- \mathrm{tr}(\Sigma_0)|}{p}  \leq C_0 \|\Sigma_0\|_2 \sqrt{\frac{\log p}{p n}};
\end{equation}
and for each $k  = 1,\dots, K$, denote by $\mathcal{A}_k$ the event
such that 
\begin{equation}\label{eq:aa}
\max_{ij} \left|[S_k - \Sigma_0^{(k)}]_{ij}\right| \leq C_0 \|\Sigma_0\|_2 \sqrt{\frac{\log p}{m_k n}}
\end{equation}
holds for some absolute constant $C_0$ which is chosen such that 
probability statement in Lemma~\ref{lemm:concNew} holds:
\begin{lemma}
\label{lemm:concNew}
Let $\mathcal{A} = \cap_{k=0}^K \mathcal{A}_k$ as in \eqref{eq:aa}, \eqref{eq:a0}. 
Then $\mathbb{P}(\mathcal{A}) \geq 1-2(K+1)\exp(-c\log p)$.
\end{lemma}
Lemma~\ref{lemm:concNew} is proved in Section~\ref{App:Conc}. Using the definition of event $\mathcal{A}$, in Section \ref{supp:pfoffd} we prove the following lemma.
\begin{lemma}
\label{lemma::offd}
Denote by $\delta_{n,k} = C_1 \twonorm{\Sigma_0} \sqrt{\frac{\log p}{n m_k}}$.  
Then on event $\mathcal A$ the following holds: 
for all $\Delta_{\Omega}$ as in \eqref{eq:DOM}
\ben
\label{eq::offdiag}
\abs{\ip{\offd(\Delta_{\Omega}), \hat{S} - \Sigma_0} } 
& \le & \sum_{k=1}^K m_k \offone{\Delta_{\Psi, k}} \delta_{n,k}
\een
where $C_0$ are some absolute constants.
\end{lemma}

We then have the following lemma, which we prove in Section \ref{sec:digpf}.
\begin{lemma}
\label{lemma::diagdecomp}
On event $\mathcal{A}$, we have for $\Delta_{\Omega} \in \mathcal{K}_{\mathbf{p}}$,
\bens
\abs{\ip{\diag(\Delta_{\Omega}), \hat{S} - \Sigma_0} }
& \le & 
 C_1 \twonorm{\Sigma_0} \sqrt{\frac{\log p}{n \min_k m_k }} 
\sqrt{(K  + 1) p}\fnorm{\diag(\Delta_{\Omega})} \\
& \le &
\max_{k} \delta_{n,k}\sqrt{(K  + 1) p}\fnorm{\diag(\Delta_{\Omega})} \\
& \asymp &
\twonorm{\Sigma_0} \fnorm{\diag(\Delta_{\Omega})}
\max_{k} \sqrt{d_k} \sqrt{\frac{\log p}{n}}
\eens
where $C_1$ is an absolute constant.
\end{lemma}

\subsection{Proof of Theorem \ref{Thm:1}}
\begin{proofof2}
Let
\begin{align}
G(\Delta_\Omega)=&Q(\Omega_0 + \Delta_\Omega) -
Q(\Omega_0)\label{eq:Qbar}
\end{align}
be the difference between the objective function \eqref{Eq:objFun} at
$\Omega_0 + \Delta_\Omega$ and at $\Omega_0$. 
Clearly $\hat{\Delta}_\Omega = \hat{\Omega}-\Omega_0$ minimizes
$G(\Delta_\Omega)$, which is a convex function  with a unique
minimizer on $\mathcal{K}_{\mathbf{p}}^\sharp$ (cf. Theorem
\ref{Thm:Conv}). 
Define
\begin{equation}\label{eq:rrnM}
\T_n =  \left\{\Delta_\Omega \in \mathcal{K}_{\mathbf{p}}: \Delta_\Omega = \Omega - \Omega_0, \Omega,\Omega_0 \in \mathcal{K}_{\mathbf{p}}^\sharp, \|\Delta_\Omega\|_F = M r_{n, \mathbf{p}}\right\}
\end{equation}
where for some large enough absolute constant $C$ to be specified, 
\ben
\label{eq:Mdef}
r_{n,\mathbf{p}} &  = & \frac{C \twonorm{\Sigma_0}}{M} \sqrt{\left(s + p\right)  (K+1) } \sqrt{\frac{\log p}{n  \min_k m_k}}
\; \; \text{ where } \\\nonumber
&& M   =  \frac{1}{2}{\phi^2_{\max}(\Omega_0)}
= \inv{2 \phi^2_{\min}(\Sigma_0)};
\een
In particular, we set $C > 9 (\max_{k} \inv{\ve_k} \vee  C_1)$ for $C_1$ as in
Lemma~\ref{lemma::diagdecomp}.

Proposition~\ref{prop:cnv} follows from~\cite{zhou2010time}.
\begin{proposition}\label{prop:cnv}
If $G(\Delta) > 0$ for all $\Delta \in \mathcal{T}_n$ as defined in \eqref{eq:rrnM}.
 then $G(\Delta) > 0$ for all $\Delta$ in
\bens
\mathcal{V}_n = \{\Delta \in \mathcal{K}_{\mathbf{p}}: \Delta = \Omega - \Omega_0, \Omega,\Omega_0 \in \mathcal{K}_{\mathbf{p}}^{\sharp}, \|\Delta\|_F > M r_{n, \mathbf{p}}\}
\eens
for  $r_{n, \mathbf{p}}$~\eqref{eq:Mdef}.
Hence if $G(\Delta) > 0$ for all $\Delta \in \T_n$, then $G(\Delta) > 0$ for all $\Delta \in \T_n \cup \mathcal{V}_n$.
\end{proposition}

\begin{proof}
By contradiction, suppose $G(\Delta')\leq 0$ for some $\Delta' \in
\mathcal{V}_n$. Let $\Delta_0 = \frac{M
  r_{n,\mathbf{p}}}{\|\Delta'\|_F}\Delta'$. Then $\Delta_0 = \theta
\mathbf{0} + (1-\theta)\Delta'$, where $0 < 1-\theta = \frac{M
  r_{n,\mathbf{p}}}{\|\Delta'\|_F} < 1$ by definition of
$\Delta_0$. Hence $\Delta_0 \in \mathcal{T}_n$ since by the convexity
of the positive definite cone $\Omega_0 + \Delta_0 \succ 0$ because
$\Omega_0 \succ 0$ and $\Omega_0 + \Delta' \succ 0$. By the convexity
of $G(\Delta)$, we have that $G(\Delta_0) \leq \theta G(\mathbf{0}) +
(1-\theta) G(\Delta') \leq 0$, 
contradicting our assumption that $G(\Delta_0) > 0$ for $\Delta_0 \in \T_n$.
\qed \end{proof}

\begin{proposition}
\label{prop:bnd}
Suppose $G(\Delta_\Omega) > 0$ for all $\Delta_\Omega \in \mathcal{T}_n$. We then have that
\[
\|\hat{\Delta}_\Omega\|_F < M r_{n,\mathbf{p}}.
\]
\end{proposition}
\begin{proof}
By definition, $G(0) = 0$, so $G(\hat{\Delta}_\Omega) \leq G(0) = 0$. Thus if $G(\Delta_\Omega) > 0$ on $\mathcal{T}_n$, then by Proposition \ref{prop:cnv}  (section \ref{app:cnv}), $\hat{\Delta}_\Omega \notin \mathcal{T}_n \cup \mathcal{V}_n$ where $\mathcal{V}_n$ is defined therein. The proposition results.
\qed \end{proof}

\begin{lemma}
\label{cor:LD} 
Under (A1) -  (A3), for all $\Delta \in {\mathcal{T}_n}$ for
which $r_{n,\mathbf{p}} = o\left(\sqrt{\frac{\min_k
      m_k}{K+1}}\right)$,
\bens
\log|\Omega_0 + \Delta| - \log|\Omega_0| \leq \langle \Sigma_0, \Delta
\rangle -
\frac{2}{9\|\Omega_0\|_2^2}\fnorm{\Delta}^2.
\eens
\end{lemma}
The proof is in Section \ref{app:LD}.

By Proposition \ref{prop:bnd}, it remains to show that $G(\Delta_\Omega) > 0$ on 
$\mathcal{T}_n$ under event $\mathcal{A}$. We show this indeed holds.
\begin{lemma}
\label{lemm:18}
On event $\mathcal{A}$, we have $G(\Delta) > 0$ for all $\Delta \in \mathcal{T}_n$.
\end{lemma}

\begin{proof}
Throughout this  proof, we assume that event $\A$ holds.
By Lemma \ref{cor:LD}, if $r_{n,\mathbf{p}} \leq \sqrt{\min_k m_k/(K+1)}$, 
we can write \eqref{eq:Qbar} using the objective \eqref{Eq:objFun},
\ben
\label{eq:G}
\lefteqn{G(\Delta_\Omega)   =
\langle\Omega_0 + \Delta_\Omega,\hat{S}\rangle - \log|\Omega_0 +
\Delta_\Omega| - \langle\Omega_0,\hat{S}\rangle + \log|\Omega_0|  } \\
&& 
\nonumber
+ \sum_k \rho_k m_k|\offone{\Psi_{k,0} + \Delta_{\Psi,k}}
- \sum_k \rho_k m_k \offone{\Psi_{k,0}} \\
& \geq&  
\nonumber
\langle\Delta_\Omega,\hat{S}\rangle - \langle
\Delta_\Omega,\Sigma_0\rangle + \frac{2}{9\|\Omega_0\|_2^2}
\|\Delta_\Omega\|_F^2 \\
\nonumber
&& 
+\sum_k \rho_k m_k (|\Psi_{k,0}+ \Delta_{\Psi,k}|_{1,\off} -
|\Psi_{k,0}|_{1,\off}) 
\\\nonumber
& =& \ip{\diag(\Delta_\Omega),\hat{S}- \Sigma_0 } +
\ip{\offd(\Delta_\Omega),\hat{S}- \Sigma_0 } + 
\frac{2}{9\|\Omega_0\|_2^2} \|\Delta_\Omega\|_F^2  \\
\nonumber
& & \quad+\underbrace{\sum_k \rho_k m_k (|\Psi_{k,0}+ \Delta_{\Psi,k}|_{1 ,\off}- |\Psi_{k,0}|_{1,\off})}_{\Delta_g}.
\een
We next bound the inner product term under event $\A$.
Substituting the bound of Lemma \ref{lemma::offd} and \eqref{eq::DDG} into \eqref{eq:G}, under event $\mathcal{A}$,
we have  by choice of $\rho_{k} = \delta_{n,p}/\ve_k$ where $0 < \ve_k < 1$ for all $k$,
\bens
\lefteqn{\sum_{k=1}^K m_k \rho_{k}
\left(\offone{\Psi_k + \Delta_{\Psi, k}} - \offone{\Psi_k}\right) + 
\ip{\offd(\Delta_{\Omega}), \hat{S} - \Sigma_0} } \\
 & \ge & 
\sum_{k=1}^K m_k \rho_{k}
\left(\onenorm{\Delta_{k, \Sc}} - \onenorm{ \Delta_{k, S}}\right) -  \sum_{k=1}^K m_k 
\offone{\Delta_{\Psi, k}} \delta_{n,k} \\
 & \ge & 
\sum_{k=1}^K m_k \rho_{k}
\left(\onenorm{\Delta_{k, \Sc}} - \onenorm{ \Delta_{k,  S}}\right) -  \sum_{k=1}^K m_k \delta_{n,k} 
\left(\onenorm{\Delta_{k, \Sc}} + \onenorm{\Delta_{k,  S}}\right) \\
 & \ge & 
-2 \max_{k} \rho_{k} \sum_{k=1}^K m_k \onenorm{\Delta_{k, S}}  = 
-2 \max_{k} \rho_{k} \onenorm{\Delta_{\Omega, S}} 
\eens
For the diagonal part, we have by Lemma~\ref{lemma::diagdecomp}
\bens
\abs{\ip{\diag(\Delta_\Omega), \hat{S}-\Sigma_0} }
& \leq &  C_1 \max_{k} \delta_{n,k} \sqrt{p} \sqrt{K  +1}\fnorm{\diag(\Delta_{\Omega})}
\eens
we have  for all $\Delta_\Omega \in \T_n$, and $C'' = \max_{k}(\frac{2}{\ve_k}) \vee \sqrt{2} C_1$, 
and for  $K \ge 1$, 
\bens
\label{eq:G3}
&G&(\Delta_\Omega) 
\geq 
\ip{\diag(\Delta_\Omega),\hat{S}- \Sigma_0 } -2 \max_{k} \rho_{n,k}
\onenorm{\Delta_{\Omega, S}} 
+ \frac{2}{9\|\Omega_0\|_2^2} \|\Delta_\Omega\|_F^2 \\
& > & 
\frac{2}{9\|\Omega_0\|_2^2} \|\Delta_\Omega\|_F^2 \\
&-&  \max_{k} \delta_{n,k} 
\left(\sqrt{p} \sqrt{K  +  1}\fnorm{\diag(\Delta_{\Omega})} + 2
  \max_{k} \inv{\ve_k} \onenorm{\Delta_{\Omega, S}} \right) \\
&\geq & 
\frac{2}{9\|\Omega_0\|_2^2} \|\Delta_\Omega\|_F^2  - 
 C' \twonorm{\Sigma_0} \sqrt{\frac{\log p}{n \min_k m_k }} \cdot \\
 && \left(\sqrt{(K + 1) p} \fnorm{\diag(\Delta_{\Omega})} + \sqrt{2 s}
\|\Delta_{\Omega,S}\|_F\right) \\
& \geq & 
\frac{2}{9\|\Omega_0\|_2^2} \fnorm{\Delta_\Omega}^2 \\&-&   
C' \left( \sqrt{2}\sqrt{(K+1)(p + s)} \fnorm{\Delta_{\Omega}}\right)  
\twonorm{\Sigma_0} \sqrt{\frac{\log p}{n \min_k m_k }} \\
 & = &
\fnorm{\Delta}^2\left(\frac{2}{9\twonorm{\Omega_0}^2} - 
C'' \twonorm{\Sigma_0} \sqrt{\frac{\log p}{n \min_k m_k }}
\frac{\sqrt{(K  + 1)(p + s)} }{M r_{n,\mathbf{p}}}\right)  > 0
\eens
which holds  for all $\Delta_\Omega \in \T_n$,
where we use the following bounds: for all $K \ge 1$.
\bens 
\sqrt{(K  + 1) p}\fnorm{\diag(\Delta_{\Omega})} + \sqrt{2s}
\|\Delta_{\Omega,S}\|_F &  \leq &  
 \sqrt{2}\sqrt{(K+1)(p +  s)}\fnorm{\Delta_{\Omega}}
\eens
and 
\bens
\lefteqn{C'' \twonorm{\Sigma_0} \sqrt{\frac{\log p}{n \min_k m_k }} \sqrt{(K  +
  1)(p + s)} \inv{M r_{n,\mathbf{p}}} } \\
& & =   
\frac{C''}{ C M} =   \frac{2 C''}{C} \phi_{\min}^2(\Sigma_0) < \frac{2}{9\twonorm{\Omega_0}^2}
\eens
where  $M = \inv{2 \phi_{\min}^2(\Sigma_0)}$,  which holds so long as
$C$ is  chosen to be large enough in
$$r_{n,\mathbf{p}} = C \twonorm{\Sigma_0} \sqrt{(s + p)
  (K+1)\frac{\log p}{n  \min_k m_k}};$$
For example, we set $C > 9 C'' = 9 (\max_{k}(\frac{2}{\ve_k}) \vee
\sqrt{2} C_1)$.
\end{proof}

Theorem \ref{Thm:1} follows from Proposition \ref{prop:bnd} immediately.
\end{proofof2}

\subsection{Proof of Lemma \ref{lemma::offd}}
\label{supp:pfoffd}
\begin{proofof2}
Assume that the event $\A$ of Lemma \ref{lemm:concNew} holds.
Using the definition of $\Delta_\Omega$ \eqref{eq:DOM}, the projection
operator $\mathrm{Proj}_{\tilde{K}_{\mathbf{p}}}(\cdot)$, 
 and letting $\tau_\Sigma = (K-1) \frac{\tr(\hat{S}-\Sigma_0)}{p}$, we have
\begin{align}\label{eq:two}
&\left| \langle \offd(\Delta_\Omega), 
\hat{S}-\Sigma_0 \rangle\right| = |\langle \Delta_\Omega, \mathrm{Proj}_{\tilde{K}_{\mathbf{p}}}(\hat{S} - \Sigma_0) \rangle |  \\\nonumber
&= \left|\langle \offd(\Delta_\Omega), ({S}_1 - {\Sigma}_0^{(1)})
  \oplus \dots \oplus ({S}_K - {\Sigma}_0^{(K)}) - \tau_\Sigma I_p
  \rangle\right| 
\\\nonumber
&=\left|\langle \offd(\Delta_{\Psi,1}) \oplus \dots \oplus
  \offd(\Delta_{\Psi,K}), ({S}_1 - {\Sigma}_0^{(1)}) \oplus \dots
  \oplus ({S}_K - {\Sigma}_0^{(K)})  \rangle\right|, 
\end{align}
where we have used the fact that $\offd(\Delta_{\Psi,1}) \oplus \dots
\oplus \offd(\Delta_{\Psi,K})$ is zero along the diagonal and thus has zero inner product with $I_p$.
Substituting Lemma \ref{lemma::diagdecomp} and the definitions of
subevents under $\mathcal{A}$, we have by \eqref{eq:III-IVN} and Lemma \ref{lem:InnCov},
\ben
\label{eq:III-IVN}
\abs{ \langle \offd(\Delta_\Omega), \hat{S}-\Sigma_0 \rangle}
 &= & \sum_{k = 1}^K m_k |\langle \offd(\Delta_{\Psi,k}), {S}_k -{\Sigma}_0^{(k)} \rangle|\\\nonumber
 &\leq  & \sum_{k = 1}^K m_k \sum_{i,j =1}^{d_k}
 |[\offd(\Delta_{\Psi,k})]_{ij}|\cdot \max_{ij} \left|[{S}_k
   -{\Sigma}_0^{(k)}]_{ij}\right|  \\ \nonumber
& \le &  C
\|\Sigma_0\|_2\sum_{k=1}^K m_k  |\Delta_{\Psi,k}|_{1,\off} \sqrt{\frac{\log p}{m_k n}}. 
\een
\end{proofof2}

\subsection{Proof of Lemma \ref{lemma::diagdecomp}: Bound on Inner Product for Diagonal}
\label{sec:digpf}
\begin{proofof2}
Let $\tilde{\Delta}_\Omega =\Delta_\Omega - \tau_{\Omega} I_p$.
Recall the identifiable parameterization of $\Delta_\Omega$ (Lemma \ref{lem:DO})
\bens
\Delta_\Omega = \tau_\Omega I_p + \tilde{\Delta}_{\Psi,1} \oplus \dots \oplus \tilde{\Delta}_{\Psi,K}
\eens
where $\tau_\Omega = \tr(\Delta_\Omega)/p$ and $\tilde{\Delta}_{\Psi,k}$ are given in the lemma.
We then have  $\tr(\tilde{\Delta}_{\Psi,j})= 0 $ and
\ben
\label{eq::sumnorm}
\sum_{k=1}^K \fnorm{\diag(\tilde{\Delta}_{\Psi,k})}^2 m_k  + p \tau_{\Omega}^2
= \fnorm{\diag(\Delta_\Omega)}^2
\een
by othogonality of the decomposition. By Lemma \ref{lem:InnCov}, we can write
\bens
\abs{\ip{\diag(\tilde{\Delta}_{\Omega}), \hat{S} - \Sigma_0} }
& \leq &\sum_{k=1}^K m_k |\ip{S_k - \Sigma_0^{(k)}, \diag(\tilde{\Delta}_{\Psi_k})}| \\
& \le & 
C \twonorm{\Sigma_0} \sum_{k=1}^K m_k
\onenorm{\diag(\tilde{\Delta}_{\Psi,k})}\sqrt{\frac{\log p}{n m_k}} \\
& \le &
C \twonorm{\Sigma_0} \sum_{k=1}^K \sqrt{m_k} \sqrt{d_k}
\fnorm{\diag(\tilde{\Delta}_{\Psi,k})}\sqrt{\frac{\log p}{n}}.
\eens
Moreover,  under $\mathcal{A}_0$, we have
\bens
\label{eq::diagsum}
\abs{\ip{ \tau_{\Omega} I_p, \hat{S} - \Sigma_0} } & \le &
C |\tau_{\Omega}| \sqrt{p} \twonorm{\Sigma_0}
\sqrt{\frac{\log p}{n}}.
\eens
Summing these terms together, we have
\ben
\nonumber
\lefteqn{\abs{\ip{\diag(\Delta_{\Omega}), \hat{S} - \Sigma_0} }} \\
& \le &
C_0 \twonorm{\Sigma_0} \sqrt{\frac{\log p}{n}}
\left( \sum_{k=1}^K \sqrt{m_k} \sqrt{d_k}
\fnorm{\diag(\tilde{\Delta}_{\Psi,k})} + |\tau_{\Omega}| \sqrt{p} \right) \nonumber\\
& \le &
\label{eq:dcmpp}
C_0 \max_{k} \sqrt{d_k} \twonorm{\Sigma_0} \sqrt{\frac{\log p}{n}} \sqrt{K + 1}\fnorm{\diag(\Delta_{\Omega})} \\
\nonumber
& = &
C_0 \max_{k} \left( \sqrt{\frac{\log p}{n m_k}}  \twonorm{\Sigma_0} \right)
\sqrt{(K  + 1)p}\fnorm{\diag(\Delta_{\Omega})} \\
\nonumber
& = &
C_0 \sqrt{\frac{\log p}{n \min_k m_k}}  \twonorm{\Sigma_0}
\sqrt{(K  + 1)p}\fnorm{\diag(\Delta_{\Omega})} \\
\nonumber
& \asymp &
\max_{k} \delta_{n, k} \sqrt{(K  + 1)p}\fnorm{\diag(\Delta_{\Omega})}
\een
where in \eqref{eq:dcmpp},  we have used the following inequality in view of \eqref{eq::sumnorm}:
\bens
\left(\sum_{k=1}^K \sqrt{m_k} 
\fnorm{\diag(\tilde{\Delta}_{\Psi,k})} + |\tau_{\Omega}| \sqrt{p}
\right)
\le \sqrt{K  + 1}\fnorm{\diag(\Delta_{\Omega})}.
\eens
\end{proofof2}


\subsection{Proof of Lemma \ref{cor:LD}}\label{app:LD}
\begin{proofof2}
We first state Proposition \ref{prop:posi-def-interval}
\begin{proposition}
\label{prop:posi-def-interval}
Under (A1)-(A3), for all $\Delta \in {\mathcal{T}_n}$,
\ben
\label{eq::eigen-bound}
\phi_{\min} (\Omega_0) > 2M r_{n,\mathbf{p}} \sqrt{\frac{K+1}{\min_k
    m_k}} \ge \twonorm{\Delta}/2
\een
so that $\Omega_0 + v \Delta \succ 0, \forall v \in I \supset [0, 1]$,
where $I$ is an open interval containing $[0, 1]$.
\end{proposition}
\begin{proof}
We first show that~\eqref{eq::eigen-bound} holds for $\Delta \in
{\mathcal{T}_n}$.
Indeed, by Corollary \ref{cor:l2bd}, we have for all $\Delta \in \mathcal{T}_n$
\bens
\|\Delta\|_2  & \leq & 
\sqrt{\frac{K+1}{\min_k m_k}} \|\Delta\|_F  = \sqrt{\frac{K+1}{\min_k m_k}}
M r_{n,\mathbf{p}}  \\
& \le & 
\sqrt{\frac{K+1}{\min_k m_k}}
\frac{C}{2} \twonorm{\Sigma_0} \sqrt{(s + p)  (K+1)\frac{\log p}{n  \min_k m_k}}
\inv{\phi_{\min}^2(\Sigma_0)} \\
&  = & \frac{C}{2} \frac{\phi_{\max}(\Sigma_0)}
{\phi^2_{\min}(\Sigma_0)} \sqrt{ (s  + p) 
 \frac{\log p}{n}}  \frac{(K+1)}{\min_k m_k} < \half \phi_{\min}(\Omega_0) = \inv{ \phi_{\max}(\Sigma_0)}
\eens
so long as 
\bens
n (\min_{k} m_k)^2 > 2 C^2 \kappa(\Sigma_0)^4 (s+ p) (K+1)^2 \log p
\eens
where $\kappa(\Sigma_0) = \phi_{\max}(\Sigma_0)/\phi_{\min}(\Sigma_0)$
is the condition number of $\Sigma_0$.

Next, it is sufficient to show that $\Omega_0 + (1 + \ve) \Delta \succ 0$
and $\Omega_0 - \ve \Delta \succ 0$ for some $1 > \ve > 0$.
Indeed,  for $\ve < 1$, 
\bens
\phi_{\min} (\Omega_0 + (1 + \ve) \Delta) 
& \geq &
\phi_{\min} (\Omega_0) - (1 + \ve) \twonorm{\Delta} \\
& > & \phi_{\min} (\Omega_0) - 2 \sqrt{\frac{K+1}{\min_k m_k}}M
r_{n,\mathbf{p}}   >0
\eens
given that by definition of $\T_n$ and \eqref{eq::eigen-bound}.
\end{proof}

Thus we have that $\log|\Omega_0 + v \Delta|$ is infinitely differentiable on 
the open interval $I \supset [0, 1]$ of $v$. This allows us to 
use the Taylor's formula with integral remainder to prove Lemma
\ref{cor:LD}, drawn from  \cite{rothman2008sparse}.

Let us use $A$ as a shorthand for
$$\mvec{\Delta}^T \left( \int^1_0(1-v)
(\Omega_0 + v \Delta)^{-1} \otimes (\Omega_0 + v \Delta)^{-1}dv
\right) \mvec{\Delta},$$
where $\mathrm{vec}(\Delta) \in \mathbb{R}^{p^2}$ is $\Delta_{p \times p}$
vectorized. Now, the Taylor expansion gives
\begin{align}
\nonumber
\log|\Omega_0 + \Delta| - \log|\Omega_0| & = 
\left.\frac{d}{dv}\log|\Omega_0 + v\Delta|\right|_{v=0} \Delta \\\nonumber&\qquad + 
\int_0^1(1-v) \frac{d^2}{dv^2}  \log|\Omega_0 + v \Delta| dv \\
& =  \langle \Sigma_0,\Delta \rangle - a.\label{eq:LogD}
\end{align}
The last inequality holds because $\nabla_\Omega \log|\Omega| = \Omega^{-1}$ and $\Omega_0^{-1} = \Sigma_0$.

We now bound $a$, following arguments from \citep{zhou2011high,rothman2008sparse}. 
\begin{align*}
a &= \int_0^1(1-v) \frac{d^2}{dv^2}  \log|\Omega_0 + v \Delta| dv\\ 
&= \mathrm{vec}(\Delta)^T \left(\int_0^1 (1-v) (\Omega_0 + v \Delta)^{-1} \otimes (\Omega_0 + v \Delta)^{-1} dv)\right) \mathrm{vec}(\Delta)\\
&\geq \|\Delta\|_F^2 \phi_{\min}\left(\int_0^1 (1-v) (\Omega_0 + v \Delta)^{-1} \otimes (\Omega_0 + v \Delta)^{-1} dv\right).
\end{align*}
Now,  suppose that
\begin{align*}
\phi_{\min}&\left(\int_0^1 (1-v) (\Omega_0 + v \Delta)^{-1} \otimes (\Omega_0 + v \Delta)^{-1} dv\right) \\
&\geq \int_{0}^1(1-v)\phi^2_{\min}((\Omega_0 + v \Delta)^{-1}) dv\\ & \geq \min_{v \in [0,1]}\phi^2_{\min}((\Omega_0 + v \Delta)^{-1}) \int_0^1 (1-v)dv\\
&= \frac{1}{2} \min_{v \in [0,1]}\frac{1}{\phi^2_{\max}(\Omega_0 + v \Delta)} = \frac{1}{2 \max_{v \in [0,1]} \phi^2_{\max}(\Omega_0 + v \Delta)}\\
&\geq \frac{1}{2(\phi_{\max}(\Omega_0) + \|\Delta\|_2)^2}.
\end{align*}
where \eqref{eq::eigen-bound},
we have for all $\Delta \in \T_n$, 
\bens
\|\Delta\|_2 \leq \sqrt{\frac{K+1}{\min_k m_k}} \|\Delta\|_F = 
\sqrt{\frac{K+1}{\min_k m_k}} M r_{n,\mathbf{p}} < \frac{1}{2}\phi_{\min}(\Omega_0)
\eens
so long as the condition in (A3) holds, namely,
\bens
n (\min_{k} m_k)^2 > 2 C^2 \kappa(\Sigma_0)^4 (s+ p) (K+1)^2 \log p.
\eens
Hence,
\begin{align*}
\phi_{\min}&\left(\int_0^1 (1-v) (\Omega_0 + v \Delta)^{-1} \otimes (\Omega_0 + v \Delta)^{-1} dv\right) \geq \frac{2}{9\phi_{\max}^2(\Omega_0) }. 
\end{align*}
Thus, substituting into \eqref{eq:LogD}, the lemma is proved.
\end{proofof2}


\subsection{Proof of Theorem \ref{Thm:1Spec}: Factorwise and Spectral Norm Bounds}
\label{app:spec}
\begin{proof}
\textbf{Part I: Factor-wise bound}. 
From the proof of Theorem \ref{Thm:1}, we know that under event $\mathcal{A}$,
\begin{equation}\label{eq:fwb}
\|\Delta_\Omega \|_F^2 \leq c (K+1)(s+p)  \frac{\log p}{n \min_k m_k}.
\end{equation}
Furthermore, since the identifiable parameterizations of $\hat{\Omega}, \Omega_0$ 
are of the form \eqref{eq::tildeA} by construction in Lemma \ref{lem:DO})
\begin{align*}
\hat{\Omega} &= \hat{\tau} I_p + (\tilde{\Psi}_1 \oplus \dots \oplus \tilde{\Psi}_K)\\
{\Omega}_0 &= {\tau_0} I_p + (\tilde{\Psi}_{0,1} \oplus \dots \oplus \tilde{\Psi}_{0,K}),
\end{align*}
we have that the identifiable parameterization of $\Delta_\Omega$ is
\begin{align}\label{idid}
\Delta_\Omega = \tau_{\Delta} I_p + (\tilde{\Delta}_1 \oplus \dots \oplus \tilde{\Delta}_K),
\end{align}
where $\tau_\Delta = \hat{\tau} - \tau_0$, $\tilde{\Delta}_k = \tilde{\Psi}_{k} - \tilde{\Psi}_{0,k}$. Observe that $\tr(\tilde{\Delta}_k) = \tr(\tilde{\Psi}_{k}) - \tr(\tilde{\Psi}_{0,k}) = 0$.

By Lemma \ref{lem:DO} then, 
\begin{align*}
\|\Delta_\Omega\|_F^2 &=  p \tau_\Delta^2 + \sum_{k=1}^K m_k \|\tilde{\Delta}_k\|_F^2.
\end{align*}
Thus, the estimation error on the underlying parameters is bounded by \eqref{eq:fwb}
\[
p \tau_\Delta^2 + \sum_{k=1}^K m_k \|\tilde{\Delta}_k\|_F^2 \leq c (K+1)(s+p)  \frac{\log p}{n\min_k  m_k},
\]
or, dividing both sides by $p$
\begin{align}\label{eq:divp}
\tau_\Delta^2 + \sum_{k=1}^K \frac{\|\tilde{\Delta}_k\|_F^2}{d_k} &\leq c (K+1)\frac{s+p}{p} \frac{\log p}{n\min_k  m_k}\\\nonumber
&=c (K+1)\left(\frac{s}{p}+1\right)  \frac{\log p}{n\min_k  m_k}.
\end{align}
Recall that $s = \sum_{k = 1}^K m_k s_k$, so $\frac{s}{p} = \sum_{k=1}^K \frac{s_k}{d_k}$.
Substituting into \eqref{eq:divp}
\begin{align}\label{eq:id}
\tau_\Delta^2 + \sum_{k=1}^K \frac{\|\tilde{\Delta}_k\|_F^2}{d_k} &\leq c (K+1)\left(1 + \sum_{k=1}^K \frac{s_k}{d_k}\right)  \frac{\log p}{n\min_k  m_k}.
\end{align}
From this, it can be seen that the bound converges as the $m_k$ increase with constant $K$. To put the bound in the form stated in the theorem, note that since $\tau_{\Delta} I_p + (\tilde{\Delta}_1^+ \oplus \dots \oplus \tilde{\Delta}_K^+)$
\begin{align*}
\frac{\|\diag(\Delta_\Omega)\|_2^2}{\max_k d_k} &\leq \frac{\left(\tau_\Delta + \sum_{k=1}^K \|\tilde{\Delta}_k^+\|_2\right)^2}{\max_k d_k}\\
&\leq\frac{K+1}{\max_k d_k}\left( \tau_\Delta^2 + \sum_{k=1}^K \|\diag(\tilde{\Delta}_k)\|_2^2\right)\\
&\leq (K+1) \left(\tau_\Delta^2 + \sum_{k=1}^K \frac{\|\diag(\tilde{\Delta}_k)\|_F^2}{d_k}\right).
\end{align*}


\textbf{Part II: Spectral norm bound}. The factor-wise bound immediately implies the bound on the spectral norm $\|\Delta_\Omega\|_2$ of the error under event $A$. We recall the identifiable representation \eqref{idid}
\[
\Delta_\Omega = \tau_{\Delta} I_p + (\tilde{\Delta}_1 \oplus \dots \oplus \tilde{\Delta}_K).
\]
By Property \ref{prop:L2} in Appendix \ref{App:Ident} and the fact that the spectral norm is upper bounded by the Frobenius norm,
\begin{align*}
\|\Delta_\Omega\|_2 &\leq |\tau_\Delta| + \sum_{k = 1}^K \|\tilde{\Delta}_k\|_2 \leq |\tau_\Delta| + \sum_{k=1}^K \|\tilde{\Delta}_k\|_F \\
&\leq \sqrt{K+1} \sqrt{ \tau_\Delta^2 + \sum_{k=1}^K \|\tilde{\Delta}_k\|_F^2}\\
&\leq \sqrt{K+1} \sqrt{\max_k d_k} \sqrt{ \tau_\Delta^2 + \sum_{k=1}^K\frac{ \|\tilde{\Delta}_k\|_F^2}{d_k}}\\
&\leq c(K+1) \sqrt{(\max_k d_k)\left(1 + \sum_{k=1}^K \frac{s_k}{d_k}\right)} \sqrt{\frac{\log p}{n \min_k m_k}},
\end{align*}
%
%
%
where in the second line, we have used the fact that for $a_k$ elements of $\mathbf{a} \in \mathbb{R}^{K}$ the norm relation $\|\mathbf{a}\|_1 \leq \sqrt{K}\|\mathbf{a}\|_2$ implies  $(\sum_{k = 1}^K |a_k|) \leq \sqrt{K} \sqrt{\sum_{k=1}^K a_k^2}$.
\qed \end{proof}

\section{Proof of Lemma \ref{lemm:concNew}: Subgaussian Concentration } 

\label{App:Conc}
We first state the following concentration result, proved in Section \ref{App:ChaosPf}. Recall that $m_k = p/d_k$.
\begin{lemma}[Subgaussian Concentration]
\label{Cor:Chaos}


Suppose that $\log p \ll m_k n$ for all $k$. Then, with probability at least $1 - 2\exp (-c' \log p)$,
\begin{align*}
|\langle \Delta, S_k - \Sigma_0^{(k)} \rangle| 
&\leq C |\Delta|_1 \|\Sigma_0\|_2 \sqrt{\frac{ \log p}{m_k n}} 
\end{align*}
for all $\Delta \in \mathbb{R}^{d_k\times d_k}$, where $c'$ is a constant depending on $C$ given in the proof.


\end{lemma}

We can now prove Lemma \ref{lemm:concNew}.

\begin{proof} 
 By Lemma \ref{Cor:Chaos} we have that event $\mathcal{A}_k$ \eqref{eq:aa}, i.e. the event that 
\begin{equation*}
\max_{ij} \left|[S_k - \Sigma_0^{(k)}]_{ij}\right|= \max_{ij} \left|\langle \mathbf{e}_i \mathbf{e}_j^T ,S_k - \Sigma_0^{(k)}\rangle\right| \leq C \|\Sigma_0\|_2 \sqrt{\frac{\log p}{m_k n}},
\end{equation*}
holds with probability at least $1 - 2\exp (-c' \log p)$. 

Note that $\mathbb{E}[\mathrm{tr}(\hat{S})] = \mathrm{tr}(\Sigma_0)$. Viewing $\frac{1}{p}\mathrm{tr}(\Sigma_0)$ as a $1\times 1$ covariance factor since $\frac{1}{p}\tr(\hat{S}) = \frac{1}{pn} \sum_{i = 1}^n \mathrm{vec}(X_i) \mathrm{vec}(X_i)^T$, we can invoke the proof of Lemma \ref{Cor:Chaos} and show that with probability at least $1 - 2 \exp(-c' \log p)$ the event $\mathcal{A}_0$ \eqref{eq:a0} will hold. 
Recall that $\mathcal{A} = \mathcal{A}_0 \cap \mathcal{A}_1 \cap \dots \cap \mathcal{A}_K$. By the union bound, we have $\mathbb{P}(\mathcal{A}) \geq 1-2(K+1)\exp(-c\log p)$.\qed
\end{proof}

\subsection{Proof of Lemma \ref{Cor:Chaos}}
\label{App:ChaosPf}
Define a $K$-way generalization of the invertible Pitsianis-Van Loan type \citep{loan1992approximation} rearrangement operator $\mathcal{R}_k(\cdot)$, which maps $p\times p$ matrices to $d_k^2 \times m_k^2$ matrices. For a matrix $M \in \mathbb{R}^{p \times p}$ we set
\begin{align}\label{eq:Rk}
\mathcal{R}_k(M) &= [\begin{array}{ccc}\mathbf{m}_1 & \dots & \mathbf{m}_{m_k^2}\end{array}],\\\nonumber
\mathbf{m}_{(i-1)m_k + j} &= \mathrm{vec}(M(i,j|k)),
\end{align}
where we use the $M(i,j|k) \in \mathbb{R}^{d_k \times d_k}$ subblock notation (see Section \ref{sec:: sec::notation} in the main text). Using this notation, we have the following concentration result. 
\begin{lemma}
\label{Lemma:Conc}
Let $\mathbf{u} \in S^{d_k^2-1}$ and $\mathbf{f} = \mathrm{vec}(I_{m_k})$. Assume that $\mathbf{x}_t = {\Sigma}_0^{1/2}\mathbf{z}_t$ where $\mathbf{z}_t$ has independent entries $z_{t,f}$ such that $\mathbb{E} z_{t,f}= 0$, $\mathbb{E} z_{t,f}^2 = 1$, and $\| z_{t,f}\|_{\psi_2} \leq K$. Let ${\Delta}_n = {{\hat{S}}} - {\Sigma}_0$. Then for all $0 \leq \frac{\epsilon}{\sqrt{m_k}}  < \frac{1}{2}$:
\begin{align*}
\mathbb{P}&(|\mathbf{u}^T \mathcal{R}_k({\Delta}_n) \mathbf{f}| \geq \epsilon \sqrt{m_k}\|{\Sigma}_0\|_2 ) \leq 2\exp\left(- c  \frac{\epsilon^2 n }{K^4 } \right)
\end{align*}
where $c$ is an absolute constant and $\|\cdot\|_{\psi_2}$ is the subgaussian norm.

\end{lemma}


\begin{proof}
We prove the lemma for $k = 1$. The proof for the remaining $k$ follow similarly.

By the definition \eqref{eq:Rk} of the permutation operator $\mathcal{R}_1$ and letting $\mathbf{x}_t(i) = [x_{t,(i-1)m_1 +1},\dots,x_{t,im_1}]$,
\begin{align}
\mathcal{R}_1({\hat{S}}) = \frac{1}{n} \sum_{t=1}^n \left[\begin{array}{c}   \mathrm{vec}(\mathbf{x}_t(1)\mathbf{x}_t(1)^T)^T  \\ \mathrm{vec}(\mathbf{x}_t(1)\mathbf{x}_t(2)^T)^T \\ \vdots \\ \mathrm{vec}(\mathbf{x}_t(d_1)\mathbf{x}_t(d_1)^T)^T   \end{array} \right]
\end{align}
Hence,
\begin{equation}
\mathbf{u}^T \mathcal{R}_1(\hat{S}) \mathbf{f} = \frac{1}{n} \sum_{t=1}^n \mathbf{x}_t^T ({U} \otimes {I}_{m_k}) \mathbf{x}_t = \frac{1}{n} \sum_{t=1}^n \mathbf{z}_t^T  {M}  \mathbf{z}_t
\end{equation}
where ${M} = {\Sigma}_0^{1/2}({U} \otimes {I}_{m_k}){\Sigma}_0^{1/2}$, ${U} = \mathrm{vec}^{-1}_{d_1,d_1} (\mathbf{u})$. 

Thus, by the Hanson-Wright inequality \citep{rudelson2013hanson}, 
\begin{align}
\mathbb{P}&(|\mathbf{u}^T \mathcal{R}_1(\hat{S}) \mathbf{f} - \mathbb{E}[ \mathbf{u}^T  \mathcal{R}_1(\hat{S}) \mathbf{f}]| \geq \tau) \\\nonumber
&\leq 2\exp\left[- c \min \left( \frac{\tau^2N^2}{K^4 n  \| {M}  \|_F^2}, \frac{\tau n}{K^2 \|{M} \|_2}\right) \right]\\\nonumber
&\leq 2\exp\left[- c \min \left( \frac{\tau^2N}{K^4 m_1 \|{\Sigma}_0  \|_2^2}, \frac{\tau n}{K^2 \|{\Sigma}_0\|_2}\right) \right]
\end{align}
since $\|U \otimes {I}_{m_1}\|_2 = \|U\|_2 \leq 1$ and $\|U \otimes {I}_{m_1}\|_F^2 =\|U\|_F^2 \|{I}_{m_1}\|_F^2 =  m_1$. Substituting $\epsilon = \frac{\tau}{\sqrt{m_1}\|\Sigma_0\|_2}$  
\begin{align}
\mathbb{P}&(|\mathbf{u}^T \mathcal{R}_1({\Delta}_n) \mathbf{f}| \geq \epsilon \sqrt{m_1} \|{\Sigma}_0\|_2 ) \leq 2\exp\left(- c  \frac{\epsilon^2 n}{K^4 } \right)
\end{align}
for all $\frac{\epsilon^2 n}{K^4  } \leq \frac{\epsilon n \sqrt{m_1}}{K^2 }$, i.e.
$
\epsilon \leq K^2 \sqrt{m_1}\leq \frac{\sqrt{m_1}}{2},
$
since $K^2 > \frac{1}{2}$ by definition.


\qed \end{proof}


We can now prove Lemma \ref{Cor:Chaos}. 

\begin{proof} 



Consider the inner product $\langle {\Delta}, {S}_k - {\Sigma}_0^{(k)}\rangle$,  where ${\Delta}$ is an arbitrary $d_k \times d_k$ matrix. Let 
\begin{align*}
\mathbf{h} &= \mathrm{vec} ({\Delta}),\qquad \mathbf{f} = \mathrm{vec} ({I}_{m_k \times m_k}).
\end{align*}
By the definition of the factor covariances $S_k$ and the rearrangement operator $\mathcal{R}_k$, it can be seen that
\[
\mathrm{vec}(S_k) = \frac{1}{m_k} \mathcal{R}_k(\hat{S}) \mathbf{f},
\]
and that similarly by the definition of the factor covariances $\Sigma_0^{(k)}$
\[
\mathrm{vec}(\Sigma_0^{(k)}) = \frac{1}{m_k} \mathcal{R}_k (\Sigma_0) \mathbf{f}.
\]
Hence,
\begin{align}
\label{eq:gram}
\langle \Delta, S_k - \Sigma_0^{(k)} \rangle &= \frac{1}{m_k} \langle \mathrm{vec}(\Delta), \mathcal{R}_k(\hat{S} - \Sigma_0) \mathbf{f}\rangle\\\nonumber
&= \frac{1}{m_k} \mathbf{h}^T \mathcal{R}_k(\hat{S} - \Sigma_0) \mathbf{f}\\\nonumber
&= \frac{1}{m_k} \sum_{i = 1}^{d_k^2} h_i \mathbf{e}_i^T \mathcal{R}_k(\hat{S} - \Sigma_0) \mathbf{f}
\end{align}
by the linearity of the rearrangement operator, the definition of the inner product, and the definition of the unit vector $\mathbf e_i$ as the $i$-th column of the $d_k^2 \times d_k^2$ identity matrix. 

We can apply Lemma \ref{Lemma:Conc} and take a union bound over $i = 1, \dots, d_k^2$. By Lemma \ref{Lemma:Conc},
\[
\mathbb{P}\left(\left|\mathbf{e}_i^T \mathcal{R}_k(\hat{S} - \Sigma_0) \mathbf{f}\right| \geq \epsilon \sqrt{m_k} \|\Sigma_0\|_2  \right) \leq 2 \exp\left(-c\frac{\epsilon^2 n}{K^4}\right)
\]
for $0 \leq \frac{\epsilon}{\sqrt{m_k}} \leq \frac{1}{2}$. 
Taking the union bound over all $i$, we obtain
\begin{align*}
\mathbb{P}\left(\max_{i} |\mathbf{e}_i^T \mathcal{R}_k(\hat{S} - \Sigma_0) \mathbf{f}| \geq \epsilon \|\Sigma_0\|_2 \sqrt{m_k} \right) &\leq 2  d_k^2 \exp\left(-c\frac{\epsilon^2 n}{K^4}\right)\\ &\leq  2\exp\left(2\log d_k -c\frac{\epsilon^2 n}{K^4}\right).
\end{align*}
Setting $\epsilon = C\sqrt{ \frac{\log p}{n}}$ for large enough $C$ and recalling that $m_k=p/d_k$, with probability at least $1 - 2\exp(-c' \log p)$ we have
\[
\max_i |\mathbf{e}_i^T \mathcal{R}_k(\hat{S} - \Sigma_0) \mathbf{f}| \leq C\|\Sigma_0\|_2 \sqrt{m_k}\sqrt{\frac{ \log p}{n}}
\]
where we assume $\log p \leq \frac{n m_k}{4C^2}$ and let $c' = \frac{cC^2}{K^4}-2$.
Hence, by \eqref{eq:gram}
\begin{align*}
|\langle \Delta, S_k - \Sigma_0^{(k)} \rangle| &= \frac{1}{m_k} \left|\sum_{i = 1}^{d_k^2} h_i \mathbf{e}_i^T \mathcal{R}_k(\hat{S} - \Sigma_0) \mathbf{f}\right|\\
&\leq \frac{1}{m_k} \sum_{i = 1}^{d_k^2} |h_i \mathbf{e}_i^T \mathcal{R}_k(\hat{S} - \Sigma_0) \mathbf{f}|\\
&\leq C\|\Sigma_0\|_2 \frac{1}{\sqrt{m_k}}\sqrt{\frac{ \log p}{n}} \sum_{i = 1}^{d_k^2} h_i \\
&= C \|\Sigma_0\|_2 \sqrt{\frac{ \log p}{m_k n}}  |\Delta|_1
\end{align*}
with probability at least $1 - 2\exp(-c' \log p)$. The first inequality follows from the triangle inequality and the last inequality from the definition of $\mathbf{h}  = \mathrm{vec}(\Delta)$ and $|\cdot|_1$.
\qed

\end{proof}

\section{Nonconvex Regularizers: Proof of Theorem \ref{thm:NonCon}}\label{app:NonConv}
Recall that the support sets $\mathcal{S}, \mathcal{S}_k$ are the set of nonzero elements of $\Omega_0$ and $\Psi_{k,0}$, respectively. 
Define $\mathcal{B}$ to be the set of matrices in $\mathcal{K}_{\mathbf{p}}$ with support contained in $\mathcal{S}$, that is
\[
\mathcal{B} = \{\Omega = \Psi_1 \oplus \dots \oplus \Psi_K \in \mathcal{K}_{\mathbf{p}} | \mathrm{supp}(\Psi_k ) \subseteq \mathcal{S}_k, \: \forall k  \} .
\]
The set $\mathcal{B}$ is the set of Kronecker sum matrices following the true sparsity pattern of the Kronecker sum $\Omega_0 = \Psi_{1,0}\oplus \dots \oplus \Psi_{K,0}$. 

Note that $\mathcal{B}$ is a linear subspace of $\mathbb{R}^{p\times p}$ since $\mathcal{K}_{\mathbf{p}}$ is a linear subspace and the intersection of two linear subspaces is a linear subspace. Hence the (L2 norm) projection $\mathrm{Proj}_{\mathcal{B}}: \mathbb{R}^{p\times p} \rightarrow \mathcal{B}$ onto $\mathcal{B}$ is given by
\[
\mathrm{Proj}_{\mathcal{B}}(A) = \mathrm{Proj}_{\mathcal{S}} ( \mathrm{Proj}_{\mathcal{K}_{\mathbf{p}}}(A)),
\]
where $\mathrm{Proj}_{\mathcal{S}}$ is the linear projection operator projecting $\mathbb{R}^{p\times p}$ onto matrices in $\mathbb R^{p\times p}$ with sparsity pattern $\mathcal{S}$, and $\mathrm{Proj}_{\mathcal{K}_{\mathbf{p}}}$ is the previously defined projection onto $\mathcal{K}_{\mathbf{p}}$ defined in Section 2 of the main text.
Note that since the sparsity pattern $\mathcal{S}$ is the sparsity pattern of a Kronecker sum matrix in $\mathcal{K}_{\mathbf{p}}$, projection onto $\mathcal{S}$ does not change the Kronecker structure. 

By reshaping we obtain the representation
\begin{align}\label{eq:pb}
\mathrm{vec}(\mathrm{Proj}_{\mathcal{B}}(A)) = \mathcal{P}_{\mathcal{B}} \mathrm{vec}(A)
\end{align}
where $\mathcal{P}_{\mathcal{B}} \in \mathbb{R}^{p^2 \times p^2}$ is the \emph{projection matrix} associated with the linear subspace $\mathcal{B}$. Recall that $\mathrm{vec}(\cdot)$ is the vectorization operator, and the \emph{projection matrix} in linear algebra is $U U^T$ where $U$ is an orthonormal basis for the subspace.

We first summarize the proof of Theorem \ref{thm:NonCon}.

\textbf{Proof plan:} The proof concept is to apply the primal-dual witness technique of \cite{LOH} to our sparse Kronecker sum precision matrix estimator. Since the nonconvex graphical lasso proof in \cite{LOH} relied on the set of $\mathcal{S}$ sparse matrices being a linear subspace of $\mathbb{R}^{p\times p}$, we can simply replace the sparse subspace in their proof with our sparse Kronecker sum subspace $\mathcal{B}$ and proceed in a similar fashion. The primal-dual witness technique can be briefly summarized as
\begin{enumerate}
\item[(i)] Prove the regularized objective function \eqref{eq:nonconvobv} is strictly convex over the constraint set, so that any zero subgradient point is the unique global minimizer.
\item[(ii)] Construct a zero subgradient point of the \emph{oracle} estimator objective function using Brouwer's theorem.
\item[(iii)] Prove this zero subgradient point $\hat{\Omega}_{\mathrm{oracle}}$ converges to the true $\Omega_0$.
\item[(iv)] Prove that the zero subgradient point of the \emph{oracle} objective is also a zero subgradient point of the full objective function \eqref{eq:nonconvobv}, hence it is the unique global minimizer and converges to $\Omega_0$.
\end{enumerate}


Proceeding with the full proof, we first have the following lemma. 
\begin{lemma}\label{lem:ncconv}
Suppose $g_{\rho}$ is $\mu$-amenable. Then for $\kappa = \sqrt{\frac{2}{\mu}}$, the objective function \eqref{eq:nonconvobv} is strictly convex over the constraint set.
\end{lemma}
\begin{proof}
Recall that
\begin{equation}\label{eq:hess}
\nabla^2 \left(-\log | \Omega| + \langle\hat{S},\Omega\rangle \right) = (\Omega \otimes \Omega)^{-1}
\end{equation}
which is a deterministic quantity not depending on the data. Hence, for $\|\Omega\|_2 \leq \sqrt{1/\mu}$, the minimum eigenvalue satisfies
\[
\lambda_{\min}(\nabla^2 \left(-\log | \Omega| + \langle\hat{S},\Omega\rangle \right))=\lambda_{\min}((\Omega \otimes \Omega)^{-1}) \geq \frac{\mu}{2}.
\]
This implies that $-\log | \Omega| + \langle\hat{S},\Omega\rangle - \frac{\mu}{2}\|\Omega\|_F^2$ is convex for $\|\Omega\|_2 \leq \sqrt{1/\mu}$. Furthermore, by $\mu$-amenability, $\sum_{k=1}^K m_k \sum_{i\neq j}g_\lambda({[{\Psi}_k]_{ij}})+ \frac{\mu}{2}\|\Omega\|_F^2$ is convex for $\Omega \in \mathcal{K}_{\mathbf{p}}$. Therefore, since $\mathcal{K}_{\mathbf{p}}$ is a linear subspace, the complete objective \eqref{eq:nonconvobv} is convex for $\|\Omega\|_2 \leq \sqrt{1/\mu}$ and $\Omega \in \mathcal{K}_{\mathbf{p}}$. Since it is convex over $\mathcal{K}_{\mathbf{p}}$, it is convex over $\mathcal{K}_{\mathbf{p}}^\sharp$ as well, since $\mathcal{K}_{\mathbf{p}}^\sharp$ is the intersection of $\mathcal{K}_{\mathbf{p}}$ and the convex positive definite cone.
\qed
\end{proof}

Since the objective is convex, a point in the subspace $\mathcal{K}_{\mathbf{p}}$ with zero subgradient will be the unique global minimum. Our first step will be to construct such a zero subgradient point.

We will first construct the (unique) oracle estimate where the oracle gives the support set of $\Omega_0$.
We will then show that this oracle estimate is also a zero-subgradient point of the objective \eqref{eq:nonconvobv} and therefore its unique global minimizer.

{ Using the $\mathcal{B}$ notation, we can write the oracle estimate as
\begin{align}
\hat{\Omega}_{\mathrm{oracle}} = \arg\min_{\Omega \in \mathcal{B}} -\log |\Omega| + \langle \hat{S},\Omega\rangle. 
\label{eq:oracleEq}
\end{align} 

Our goal will be to construct a map $F: \mathcal{B} \rightarrow \mathcal{B}$ such that (a) $\Delta$ is a fixed point of $F$ if and only if $\Omega_0 + \Delta$ is a fixed point of the oracle estimate \eqref{eq:oracleEq}, (b) $F$ maps the intersection $\mathcal{B} \cap \mathbb{B}_{\infty}(r)$ of $\mathcal{B}$ and the radius-$r$ $\ell_\infty$-ball centered at the origin to itself for some $r$, and (c) this $r$ is such that $\Omega = \Omega_0 + \Delta\succ 0$, for all $\Delta \in \mathcal{B} \cap \mathbb{B}_{\infty}(r)$. Then by Brouwer's fixed point theorem we can show that $F$ must have a fixed point $\Delta_*$ in that ball. By construction (a) above, this fixed point $\Delta_*$ will correspond to a fixed point $\Omega_0 + \Delta^*$ in the oracle estimator objective, hence the oracle estimate will have $\ell_\infty$-ball error less than $r$. }

For $F$, we will choose a Newton method step (gradient step preconditioned by inverse Hessian). Denote the pseudoinverse of a matrix $A$ as $A^\dag$. {We now write the map $F: \mathcal{B} \rightarrow \mathcal{B}$ given by}
\[
F(\Delta_S) := -\Gamma^\dag  \mathrm{vec}\left( \mathrm{Proj}_{\mathcal{B}}(\hat{S} - (\Omega_0 + \Delta_S)^{-1})\right) + \mathrm{vec}(\Delta_S)
\]
where $\Delta_S \in \mathcal{B}$, and we let $\Gamma$ be the Hessian of the objective function within $\mathcal{B}$:\footnote{With $\mathcal{P}_{\mathcal{B}} = U U^T$ as above (where columns of $U$ form an orthonormal basis for the subspace $\mathcal{B}$), $\Gamma = U U^T ( \Sigma_0 \otimes \Sigma_0)U U^T$ and hence $\Gamma^\dag = U \left(U^T ( \Sigma_0 \otimes \Sigma_0)U\right)^{-1} U^T$ since $\Sigma_0$ is positive definite.}
\[
\Gamma = \mathcal{P}_{\mathcal{B}} ( \Sigma_0 \otimes \Sigma_0)\mathcal{P}_{\mathcal{B}}.
\]
The quantity $ \Sigma_0 \otimes \Sigma_0 $ is included as it is the Hessian of the objective function \eqref{eq:hess}. The pseudoinverse is needed since $\mathcal{P}_\mathcal{B}$ is low rank, making the Hessian within $\mathcal{B}$ low rank.

Clearly if $\mathrm{Proj}_{\mathcal{B}}(\hat{S} - (\Omega_0 + \Delta_S)^{-1}) = 0$, $F(\Delta_S) = \Delta_S$ and vice versa, hence $\Delta_S$ is a fixed point of $F$ if and only if $\Omega_0 + \Delta_S$ is a fixed point of the oracle objective \eqref{eq:oracleEq}. 
Now
\[
\|\Delta_S\|_2 \leq dr
\]
since $\Delta_S$ has at most $d$ nonzero entries per row. Hence the matrix $\Omega_0 + \Delta_S$ is invertible and positive definite whenever $dr < \lambda_{\min} (\Omega_0)$, making $F$ a continuous map on $\mathbb{B}_{\infty}(r) \cap \mathcal{B}$ and satisfying condition (c). 

Define the constants $\kappa_\Gamma = \|\Gamma^\dag\|_\infty$ and $\kappa_\Sigma = \|\Sigma_0\|_\infty$, in other words, we are assuming that the Hessian is well-conditioned in the $\infty$-norm sense, which is possible since $\Sigma_0$ has eigenvalues bounded from above and below. We now show the following lemma by verifying the remaining condition (b) on $F$ and applying Brouwer's fixed point theorem. Several relevant quantities are summarized in Table \ref{table:variables} for convenience.
\begin{figure}
\centering
\footnotesize{
\begin{tabular}{|c|l|}
\hline
\textbf{Variable} & \textbf{Definition} \\\hline
$\mathcal{L}_n(\Omega)$ & $\left.-\log |\Omega| + \langle \hat{S},\Omega\rangle \right|_{\Omega \in \mathcal{K}_{\mathbf{p}}}$: Objective function less regularization terms. \\\hline
$q_\rho(t)$ & $g_\rho(t) - \rho |t|$: Difference between regularizer and $\ell 1$ penalty. \\\hline
$\mathcal{K}_{\mathbf{p}}$ & Set of Kronecker sum matrices with fixed dimensions $\mathbf{p} = [d_1, \dots, d_K]$.  \\
\hline
$\mathcal{B}$ & Set of matrices in $\mathcal{K}_{\mathbf{p}}$ with support contained in $\mathcal{S}$.  \\
\hline
$\mathrm{Proj}_{\mathcal{B}}(\cdot)$ & Linear projection operator from $\mathbb{R}^{p\times p}$ onto $\mathcal{B}$. \\\hline
$\mathcal{P}_\mathcal{B}$ & Projection matrix corresponding to $\mathrm{Proj}_{\mathcal{B}}$. $\mathcal{P}_{\mathcal{B}} \mathrm{vec}(A)  = \mathrm{vec}(\mathrm{Proj}_{\mathcal{B}}(A))$. \\\hline
$\Gamma$ & $\mathcal{P}_{\mathcal{B}} ( \Sigma_0 \otimes \Sigma_0) \mathcal{P}_{\mathcal{B}}$: Hessian of $\mathcal{L}_n$ within subspace $\mathcal{B}$. \\\hline
$\kappa_\Gamma$ & $\|\Gamma^\dag\|_\infty$ \\\hline
$\kappa_\Sigma$ & $\|\Sigma_0\|_\infty$ \\\hline
$\tau_\Sigma$ & $\frac{\tr(\hat{S}) - \tr(\Sigma_0)}{p}$ \\\hline
$\mathbf{e}_i$ & $i$th unit vector in $\mathbb{R}^p$. \\\hline
\end{tabular}
}
\captionof{table}{Selected quantities used in the proof of Theorem \ref{thm:NonCon}}
\label{table:variables}
\end{figure}

\begin{lemma}\label{lemm:Lemm7}
Let $r = 2C_0\kappa_\Gamma \|\Sigma_0\|_2 (K+1) \sqrt{\frac{\log p}{n \min_k m_k}}$ where $C_0$ is a constant depending only on the subgaussian parameter of the data and
\[
dr \leq \min \left\{\frac{1}{2} \lambda_{\min} (\Omega_0), \frac{1}{2 \kappa_\Sigma}, \frac{1}{4\kappa_\Gamma \kappa_\Sigma^3} \right\}.
\]
Assume the sample size satisfies $n\min_k m_k \geq \kappa_\Gamma^2 \log p$. Then under event $\mathcal{A}$ as in Theorem \ref{Thm:1} there exists $\hat{\Omega}_{\mathrm{oracle}} \in \mathcal{B}$ such that
\[
\|\hat{\Omega}_{\mathrm{oracle}} - \Omega_0 \|_{\max} \leq r, \qquad \|\hat{\Omega}_{\mathrm{oracle}} - \Omega_0\|_2 \leq dr, \: \: and \quad \mathrm{Proj}_{\mathcal{B}}(\hat{S} - \hat{\Omega}_{\mathrm{oracle}}^{-1}) = 0.
\]
\end{lemma}

\begin{proof}

First, note that $\Gamma^\dag \Gamma \mathrm{vec}(\Delta) = \mathrm{vec}(\Delta)$ for any $\Delta \in \mathcal{B}$, since $\Gamma$ is the projection of the positive definite matrix $\Sigma_0 \otimes \Sigma_0$ onto the low rank subspace $\mathcal{B}$. 

Suppose $\Delta_S \in \mathbb{B}_\infty(r)$. Then
\begin{align*}
F(\mathrm{vec}(\Delta_S)) := 
-\Gamma^\dag &\left\{\mathrm{vec}\left(\mathrm{Proj}_{\mathcal{B}}(\hat{S} - \Sigma_0)\right)\right.\\& \left.+ \mathrm{vec}\left(\mathrm{Proj}_{\mathcal{B}}(\Sigma_0 -  (\Omega_0 + \Delta_S)^{-1})\right) +\Gamma \mathrm{vec}(\Delta_S)\right\},
\end{align*}
hence
\begin{align}
\|F(\mathrm{vec}(\Delta_S))\|_\infty \leq& \kappa_\Gamma \|\mathrm{vec}\left(\mathrm{Proj}_{\mathcal{B}}(\hat{S} - \Sigma_0)\right)\|_\infty \label{eq:90}\\&+ \kappa_\Gamma \| \mathrm{vec}\left(\mathrm{Proj}_{\mathcal{B}}(\Sigma_0 -  (\Omega_0 + \Delta_S)^{-1})\right) +\Gamma \mathrm{vec}(\Delta_S)\|_\infty \nonumber,
\end{align}
by the definition of $\kappa_\Gamma$ and the triangle inequality.

The first term of \eqref{eq:90} can be bounded via the concentration inequalities used for the $\ell1$ case. Specifically, note that
\begin{align}
\kappa_\Gamma& \|\mathrm{vec}\left(\mathrm{Proj}_{\mathcal{B}}(\hat{S} - \Sigma_0)\right)\|_\infty \\\nonumber&\leq \kappa_\Gamma \|\mathrm{vec}\left(\mathrm{Proj}_{\mathcal{K}_{\mathbf{p}}}(\hat{S} - \Sigma_0)\right)\|_\infty \\\nonumber
&= \kappa_\Gamma \|(S_1 - \Sigma_0^{(1)} \oplus \dots \oplus (S_K - \Sigma_0^{(K)}) - \tau_\Sigma I_p \|_{\max}\\\nonumber
&\leq \kappa_\Gamma \sum_{k = 1}^K \| S_k - \Sigma_0^{(k)} \|_{\max} +\frac{|\tr(\hat{S} - \tr(\Sigma_0)|}{p}.\nonumber
\end{align}
where we have used $\tau_\Sigma = \frac{\tr(\hat{S}) - \tr(\Sigma_0)}{p}$. Now recall that under event $\mathcal{A}_k$, defined in \eqref{eq:aa} above,
\begin{equation*}
\max_{ij} \left|[S_k - \Sigma_0^{(k)}]_{ij}\right|= \max_{ij} \left|\langle \mathbf{e}_i \mathbf{e}_j^T ,S_k - \Sigma_0^{(k)}\rangle\right| \leq C_0 \|\Sigma_0\|_2 \sqrt{\frac{\log p}{m_k n}},
\end{equation*}
and under event $\mathcal{A}_0$, defined above in \eqref{eq:a0},
\[
|\tau_\Sigma| \leq C_0 \|\Sigma_0\|_2 \sqrt{\frac{\log p}{p n}}.
\]
Hence under event $\mathcal{A} = \bigcup_{k = 0}^K \mathcal{A}_k$,
\begin{align}
\kappa_\Gamma \|\mathrm{vec}\left(\mathrm{Proj}_{\mathcal{B}}(\hat{S} - \Sigma_0)\right)\|_\infty &\leq C_0\kappa_\Gamma \|\Sigma_0\|_2\left[\sqrt{\frac{\log p}{p n}} + \sum_{k = 1}^K \sqrt{\frac{\log p}{m_k n}}\right]\nonumber\\ &\leq C_0\kappa_\Gamma \|\Sigma_0\|_2 (K+1) \sqrt{\frac{\log p}{n \min_k m_k}}\nonumber\\
&= \frac{r}{2}.\label{eq:SigInf}
\end{align}
Finally, recall that by Lemma \ref{lemm:concNew} event $\mathcal{A}$ holds with probability $ \geq 1 - 2(K + 1) \exp( - c \log p)$. 

Moving on to the second term of \eqref{eq:90}, we apply the matrix expansion
\begin{equation}\label{eq:matExp}
(A + \Delta)^{-1} - A^{-1} = \sum_{\ell = 1}^\infty (-A^{-1} \Delta)^\ell A^{-1},
\end{equation}
{ and note that (since $\Delta_S \in \mathcal{B}$ implies $\mathcal{P}_\mathcal{B}\mathrm{vec}(\Delta_S) = \mathrm{vec}(\Delta_S)$)
\begin{align*}
\Gamma \mathrm{vec}(\Delta_S) 
&= \mathcal{P}_{\mathcal{B}} \left((\Sigma_0 \otimes \Sigma_0) \mathcal{P}_\mathcal{B}\mathrm{vec}(\Delta_S)\right)\\
&= \mathrm{vec}\left(\mathrm{Proj}_{\mathcal{B}}\left(\mathrm{vec}^{-1}\left((\Sigma_0 \otimes \Sigma_0) \mathrm{vec}(\Delta_S)\right) \right)\right)\\& = \mathrm{vec}(\mathrm{Proj}_{\mathcal{B}}(\Sigma_0 \Delta_S \Sigma_0)),
\end{align*}}
where we have used the fact that for symmetric matrices $A, B$, $\mathrm{vec}(A B A) = (A \otimes A)\mathrm{vec}(B)$.

We then obtain
\begin{align*}
&\mathrm{vec}\left(\mathrm{Proj}_{\mathcal{B}}(\Sigma_0 -  (\Omega_0 + \Delta_S)^{-1})\right) +\Gamma \mathrm{vec}(\Delta_S)\\& = \mathrm{vec}\left(\mathrm{Proj}_{\mathcal{B}}(\Sigma_0 -  (\Omega_0 + \Delta_S)^{-1} - \Sigma_0 \Delta_S \Sigma_0)\right)\\
&= \mathrm{vec}\left(\mathrm{Proj}_{\mathcal{B}}\left(\sum_{\ell = 2}^\infty (-\Sigma_0 \Delta_S)^\ell \Sigma_0\right)\right).
\end{align*}
We have used $\mathrm{vec}^{-1}(\cdot)$ to denote the inverse of the vectorization operator.

Via the triangle inequality and the linearity of the vectorization and projection operators,
\begin{align}
\nonumber\| \mathrm{vec}&\left(\mathrm{Proj}_{\mathcal{B}}(\Sigma_0 -  (\Omega_0 + \Delta_S)^{-1})\right) +\Gamma \mathrm{vec}(\Delta_S)\|_\infty \\&\leq \max_{(j,k) \in S} \sum_{\ell = 2}^\infty |\mathbf{e}_j^T (\Sigma_0\Delta)^\ell \Sigma_0 \mathbf{e}_k|.\label{eq:93}
\end{align}
Now we can apply Holder's inequality to obtain
\begin{align*}
|\mathbf{e}_j^T (\Sigma_0\Delta)^\ell \Sigma_0 \mathbf{e}_k| &\leq \|\mathbf{e}_j^T(\Sigma_0 \Delta)^{\ell-1}\Sigma_0\|_1 \|\Delta \Sigma_0\|_{\infty}\\
&\leq \|\Sigma_0(\Delta \Sigma_0)^{\ell-1}\|_1 \|\Delta\|_{\max}\|\Sigma_0 e_k\|_1\\
&\leq \|\Sigma_0\|_1^{\ell-1} \|\Delta\|_{\max}\|\Sigma_0\|_1\\
&=\|\Sigma_0\|_{\infty}^{\ell+1} \|\Delta\|_{\infty}^{\ell-1}\|\Delta\|_{\max}.
\end{align*}
Then, using the fact that $\|\Delta\|_2 \leq \|\Delta\|_\infty \leq dr$ and substituting back into \eqref{eq:93}, we have
\begin{align*}
\| \mathrm{vec}&\left(\mathrm{Proj}_{\mathcal{B}}(\Sigma_0 -  (\Omega_0 + \Delta_S)^{-1})\right) +\Gamma \mathrm{vec}(\Delta_S)\|_\infty\\
&\leq \sum_{\ell = 2}^\infty \kappa_{\Sigma}^{\ell + 1} d^{\ell - 2} r^\ell \\
&= \frac{\kappa_{\Sigma}^3 d r^2}{1 - \kappa_{\Sigma} dr} \\
&\leq 2\kappa_{\Sigma}^3 d r^2.
\end{align*}
Since our assumption implies that $2 \kappa_\Sigma^3 d r^2 \leq r$, we therefore have that
\[
\|F(\mathrm{vec}(\Delta_S))\|_\infty \leq r
\]
under event $\mathcal{A}$. Since $F(\mathbb{B}_\infty(r)\cap \mathcal{B}) \in \mathbb{B}_\infty(r)\cap \mathcal{B}$, by Brouwer's fixed point theorem \citep{ortega1970iterative}, $F$ must have a fixed point $\Delta_S^*$. Recalling that $\Delta^*_S, \Omega_0 \in \mathcal{B}$, we choose $\hat{\Omega} _{\mathrm{oracle}}= \Omega_0 + \Delta^*_S$. Hence by construction $\|\hat{\Omega}_{\mathrm{oracle}} - \Omega_0\|_{\max} \leq r$ and $\|\hat{\Omega}_{\mathrm{oracle}} - \Omega_0\|_2 \leq dr$ since both matrices have degree bounded by $d$. 
{
The last equality follows since $\Delta_S^*$ is the fixed point of $F$, i.e. $F(\Delta_S^*) = \Delta_S^*$, which can only occur if 
\[
\mathrm{vec}\left( \mathrm{Proj}_{\mathcal{B}}(\hat{S} - (\Omega_0 + \Delta_S^*)^{-1})\right) = 0.
\]
}
\qed
\end{proof}

{ 
Using this lemma it remains to show that $\hat{\Omega}_{\mathrm{oracle}}$ satisfies the constraints and is a zero-subgradient point of the complete objective \eqref{eq:nonconvobv}, and hence is the unique global optimum. 
}
Define 
$\mathcal{L}_n(\Omega)$ to be the objective function \eqref{eq:nonconvobv} less the regularization terms, i.e.
\[
\mathcal{L}_n(\Omega) = \left.-\log |\Omega| + \langle \hat{S},\Omega\rangle \right|_{\Omega \in \mathcal{K}_{\mathbf{p}}}.
\]
{
\begin{lemma}\label{lemm:zerosub}
The oracle estimate $\hat{\Omega}_{\mathrm{oracle}}$ will be a zero-subgradient point of the global objective \eqref{eq:nonconvobv} if the inequalities
\begin{equation}\label{eq:62a}
\|\nabla_{\mathcal{K}_{\mathbf{p}}} \mathcal{L}_n (\Omega_0)\|_\infty \leq \frac{1}{2} \rho
\end{equation}
and 
\begin{equation}\label{eq:62b}
\|\hat{Q}_{S^c S} (\hat{Q}_{SS})^{\dag} \nabla_{\mathcal{K}_{\mathbf{p}}} \mathcal{L}_n (\Omega_0)_S \|_\infty \leq \frac{1}{2} \rho
\end{equation}
hold, where
\[
\hat{Q} = \int_0^1 \nabla^2_{\mathcal{K}_{\mathbf{p}}} \mathcal{L}_n\left(\Omega_0 + t(\hat{\Omega}_{\mathrm{oracle}}-\Omega_0)\right)dt.
\]
We have denoted $\nabla_{\mathcal{K}_{\mathbf{p}}} f = \mathcal{P}_{\mathcal{K}_{\mathbf{p}}} \nabla f$ and $\nabla^2_{\mathcal{K}_{\mathbf{p}}} f = \mathcal{P}_{\mathcal{K}_{\mathbf{p}}}(\nabla^2 f)\mathcal{P}_{\mathcal{K}_{\mathbf{p}}}$ to be the gradient and Hessian respectively of $f$ projected onto the subspace $\mathcal{K}_{\mathbf{p}}$ ($\mathcal{P}_{\mathcal{K}_{\mathbf{p}}}$ is the projection matrix onto $\mathcal{K}_{\mathbf{p}}$).
\label{lem:Conds}
\end{lemma}
\begin{proof}
In this proof, for simplicity we write $q_\rho(\hat{\Omega})$ to indicate $q_\rho(t) = g_\rho(t) - \rho |t|$ applied elementwise to the offdiagonal elements of $\hat{\Omega}$:
\[
[q_\rho(\hat{\Omega})]_{ij} = \left\{\begin{array}{ll} q_\rho(\hat{\Omega}_{ij}) & i \neq j\\
0 & \mathrm{otherwise}.\end{array}\right. 
\]
Observe that by construction $\nabla_{\mathcal{K}_{\mathbf{p}}} q_\rho({\Omega}) = \nabla_{\mathbb{R}^{p\times p}} q_\rho({\Omega}) = \nabla q_\rho (\Omega)$ for any ${\Omega} \in \mathcal{K}_{\mathbf{p}}$.

For the objective \eqref{eq:nonconvobv}, the zero subgradient condition is given by
\[
\nabla_{\mathcal{K}_{\mathbf{p}}} \mathcal{L}_n (\hat{\Omega}) - \nabla q_{\rho}(\hat{\Omega}) + \rho \hat{z} = 0,
\]
where $\hat{z} = \partial | \hat \Omega |_{1,\mathrm{off}}$ is an element of the subgradient of the off-diagonal $\ell$1 norm at $\hat{\Omega}$. Adding and subtracting $\nabla_{\mathcal{K}_{\mathbf{p}}} \mathcal{L}_n (\Omega_0)$ gives
\[
(\nabla_{\mathcal{K}_{\mathbf{p}}} \mathcal{L}_n (\hat{\Omega}) - \nabla_{\mathcal{K}_{\mathbf{p}}} \mathcal{L}_n (\Omega_0)) + (\nabla_{\mathcal{K}_{\mathbf{p}}} \mathcal{L}_n (\Omega_0) - \nabla q_{\rho}(\hat{\Omega})) + \rho \hat{z} = 0.
\]
By the fundamental theorem of calculus we have (for $\hat{\Omega} = \hat{\Omega}_{\mathrm{oracle}}$) that $\nabla_{\mathcal{K}_{\mathbf{p}}} \mathcal{L}_n (\hat{\Omega}_{\mathrm{oracle}}) - \nabla_{\mathcal{K}_{\mathbf{p}}} \mathcal{L}_n (\Omega_0) = \hat{Q}\mathrm{vec}(\hat{\Omega}_{\mathrm{oracle}} - \Omega_0)$, hence
\[
\hat{Q}\mathrm{vec}(\hat{\Omega}_{\mathrm{oracle}} - \Omega_0) + (\nabla_{\mathcal{K}_{\mathbf{p}}} \mathcal{L}_n (\Omega_0) - \nabla q_{\rho}(\hat{\Omega}_{\mathrm{oracle}})) + \rho \hat{z} = 0.
\]
Rewriting in block form gives 
\begin{align*}
&\left[ \begin{array}{cc} \hat{Q}_{SS} & \hat{Q}_{S S^c}\\ \hat{Q}_{S^c S} & \hat{Q}_{S^c S^c}\end{array}\right]
\left(\hat{\Omega}_{\mathrm{oracle}} - \Omega_0\right) \\
&+ \left(\left[\begin{array}{c} \nabla_{\mathcal{K}_{\mathbf{p}}} \mathcal{L}_n(\Omega_0)_S - \nabla q_\rho(\hat{\Omega}_{\mathrm{oracle}})_S\\
\nabla_{\mathcal{K}_{\mathbf{p}}} \mathcal{L}_n(\Omega_0)_{S^c} - \nabla q_\rho(0)\end{array}\right]\right)+ \rho\left[\begin{array}{c} \hat{z}_S \\\hat{z}_{S^c} \end{array} \right] = 0,
\end{align*}
where $\hat{Q}_{SS}$ is the block of $\hat{Q}$ corresponding to the elements in $\mathcal{S}$ along both axes, $\hat{Q}_{S^cS^c}$ is the block of $\hat{Q}$ corresponding to the elements in the complement of $\mathcal{S}$, etc.
After some algebra we obtain a solution
\begin{align*}
\hat{z}_{S^c} =& \frac{1}{\rho}\left\{ \left(\nabla q_\rho(0) - \nabla_{\mathcal{K}_{\mathbf{p}}} \mathcal{L}_n(\Omega_0)_{S^c}\right)
\right.\\&+\left. \hat{Q}_{S^c S} \hat{Q}_{SS}^{\dag}\left(\nabla_{\mathcal{K}_{\mathbf{p}}}\mathcal{L}_n(\Omega_0)_S - \nabla q_\rho(\hat{\Omega}_{\mathrm{oracle}})_S + \rho \hat{z}_S\right)\right\},
\end{align*}
since $\nabla q_\rho(0) = 0$ by definition. Now from Lemma \ref{lemm:Lemm7}, under event $\mathcal{A}$
\[
\|\hat{\Omega}_{\mathrm{oracle}} - \Omega_0 \|_{\max} \leq r,
\]
and observe that $\rho \gamma > r$ since we have assumed that $n \min_k m_k \geq c_0 d^2 \log p$ for some $c_0$ large enough.
By our assumption that $|[\Omega_{0}]_{ij}| \geq \rho \gamma + r $ for all $i,j$, we then have (again under event $\mathcal{A}$)
\[
\min_{ij} |[\hat{\Omega}_{\mathrm{oracle}}]_{ij}| \geq \rho\gamma + r - r = \rho \gamma.
\]
Therefore, using condition (f) of the definition of a $(\mu,\gamma)$ regularizer, $ - \nabla q_{\rho}(\hat{\Omega}_{\mathrm{oracle}})_S + \rho \hat{z}_S = 0$ and
\begin{align*}
\|\hat{z}_{S^c}\|_\infty &= \frac{1}{\rho} \left\|- \nabla_{\mathcal{K}_{\mathbf{p}}} \mathcal{L}_n(\Omega_0)_{S^c} + \hat{Q}_{S^c S} \hat{Q}_{SS}^{\dag}\nabla_{\mathcal{K}_{\mathbf{p}}}\mathcal{L}_n(\Omega_0)_S  \right\|_\infty\\
&\leq \frac{1}{\rho} \|\nabla_{\mathcal{K}_{\mathbf{p}}} \mathcal{L}_n(\Omega_0)_{S^c}\|_{\infty} + \frac{1}{\rho} \| \hat{Q}_{S^c S} \hat{Q}_{SS}^{\dag}\nabla_{\mathcal{K}_{\mathbf{p}}}\mathcal{L}_n(\Omega_0)_S\|_{\infty}\\
&\leq \frac{1}{\rho} \|\nabla_{\mathcal{K}_{\mathbf{p}}} \mathcal{L}_n(\Omega_0)\|_{\infty} + \frac{1}{\rho} \| \hat{Q}_{S^c S} \hat{Q}_{SS}^{\dag}\nabla_{\mathcal{K}_{\mathbf{p}}}\mathcal{L}_n(\Omega_0)_S\|_{\infty}\\
&\leq \frac{1}{2} + \frac{1}{2} = 1,
\end{align*}
where we have applied the assumed inequalities.
Since $\|\hat{z}_{S^c}\|_\infty \leq 1$, it is a feasible subgradient and therefore $\hat{\Omega}_{\mathrm{oracle}}$ is a zero subgradient point of the global objective function \eqref{eq:nonconvobv}.

\qed
\end{proof}
}
We now show the inequalities \eqref{eq:62a}, \eqref{eq:62b} assumed by Lemma \ref{lemm:zerosub} hold under event $\mathcal{A}$. Note that
\[
\|\nabla_{\mathcal{K}_{\mathbf{p}}} \mathcal{L}_n (\Omega_0)\|_\infty = \|\mathrm{Proj}_{\mathcal{K}_{\mathbf{p}}}(\hat{S} - \Sigma_0) \|_{\max}
\]
and thus by \eqref{eq:SigInf}, under event $\mathcal{A}$ equation \eqref{eq:62a} holds with 
$
\rho = \frac{r}{\kappa_\Gamma}.
$

{

It remains to show \eqref{eq:62b} holds with $\rho = r$. We will first bound
\[
\|(\nabla_{\mathcal{K}_{\mathbf{p}}}^2 \mathcal{L}_n(\Omega_0))_{S^cS}(\nabla_{\mathcal{K}_{\mathbf{p}}}^2\mathcal{L}_n(\Omega_0))_{SS}^{\dag}(\nabla_{\mathcal{K}_{\mathbf{p}}}\mathcal{L}_n(\Omega_0)_S\|_\infty,
\]
and then show that the expression on the left hand side of \eqref{eq:62b} is close to this quantity.

First, by the definition of the infinity norm and $\mathcal{P}_{\mathcal{K}_{\mathbf{p}}}$ it can be shown that
\begin{align}\label{eq:Pinfnorm}
\|\mathcal{P}_{\mathcal{K}_{\mathbf{p}}}\|_\infty = \sup_{A \in \mathbb{R}^{p\times p}} \frac{\|\mathcal{P}_{\mathcal{K}_{\mathbf{p}}} \mathrm{vec}(A)\|_{\infty}}{\|A\|_{\max}} = \sup_{A \in \mathbb{R}^{p\times p}} \frac{\|\mathrm{Proj}_{\mathcal{K}_{\mathbf{p}}}(A)\|_{\max}}{\|A\|_{\max}}  \leq 2K,
\end{align}
where we have used the expression \eqref{eq:elementProj} for the elements of the projected matrix and the fact that an average of a set of elements of $A$ cannot have magnitude larger than $\|A\|_{\max}$. 
Noting that $(\nabla_{\mathcal{K}_{\mathbf{p}}}^2\mathcal{L}_n(\Omega_0))_{SS}^{\dag} = (\Gamma^{\dag})_{SS}$,
\begin{align}
\label{eq:HessBds}
\|(\nabla_{\mathcal{K}_{\mathbf{p}}}^2 &\mathcal{L}_n(\Omega_0))_{S^cS}(\nabla_{\mathcal{K}_{\mathbf{p}}}^2\mathcal{L}_n(\Omega_0))_{SS}^{\dag}(\nabla_{\mathcal{K}_{\mathbf{p}}}\mathcal{L}_n(\Omega_0)_S\|_\infty\\
&\leq \|(\nabla_{\mathcal{K}_{\mathbf{p}}}^2 \mathcal{L}_n(\Omega_0))_{S^cS}\|_\infty \cdot \|(\Gamma^{\dag})_{SS} \|_{\infty} \cdot \|(\nabla_{\mathcal{K}_{\mathbf{p}}}\mathcal{L}_n(\Omega_0)_S \|_{\infty}\nonumber \\
&\leq O\left(\frac{r}{\kappa_\Gamma}\right),\nonumber
\end{align}
since $\|(\nabla_{\mathcal{K}_{\mathbf{p}}}^2 \mathcal{L}_n(\Omega_0))_{S^cS}\|_\infty \leq \|\nabla_{\mathcal{K}_{\mathbf{p}}}^2 \mathcal{L}_n(\Omega_0)\|_\infty = \left\|\mathcal{P}_{\mathcal{K}_{\mathbf{p}}} \left(\nabla^2\mathcal{L}_n(\Omega_0)\right)\mathcal{P}_{\mathcal{K}_{\mathbf{p}}}\right\|_\infty \leq \|\mathcal{P}_{\mathcal{K}_{\mathbf{p}}}\|_\infty^2 \|\nabla^2\mathcal{L}_n(\Omega_0)\|_\infty = \|\mathcal{P}_{\mathcal{K}_{\mathbf{p}}}\|_\infty^2 \|\Sigma_0\otimes \Sigma_0\|_\infty =\|\mathcal{P}_{\mathcal{K}_{\mathbf{p}}}\|_\infty^2 \|\Sigma_0\|_\infty^2 \leq 4K^2 \kappa_\Sigma^2 $, and \eqref{eq:62a} holds under event $\mathcal{A}$ with $\rho = \frac{r}{\kappa_\Gamma}$. 

We now relate the bound in \eqref{eq:HessBds} to that required to show \eqref{eq:62b}. Note that
\begin{align}\label{eq:HessQuant}
&\|\hat{Q}_{S^c S} (\hat{Q}_{SS})^{\dag} \nabla_{\mathcal{K}_{\mathbf{p}}} \mathcal{L}_n (\Omega_0)_S\|_\infty \\\nonumber
&\leq \|(\nabla_{\mathcal{K}_{\mathbf{p}}}^2 \mathcal{L}_n(\Omega_0))_{S^cS}(\nabla_{\mathcal{K}_{\mathbf{p}}}^2\mathcal{L}_n(\Omega_0))_{SS}^{\dag}(\nabla_{\mathcal{K}_{\mathbf{p}}}\mathcal{L}_n(\Omega_0)_S)\|_\infty + \|\Xi(\nabla_{\mathcal{K}_{\mathbf{p}}}\mathcal{L}_n(\Omega_0)_S )\|_\infty\\\nonumber
&\leq O\left(\frac{r}{\kappa_\Gamma}\right) + \|\Xi(\nabla_{\mathcal{K}_{\mathbf{p}}}\mathcal{L}_n(\Omega_0)_S) \|_\infty,
\end{align}
where we have defined
\[
\Xi = \hat{Q}_{S^c S}(\hat{Q}_{SS})^{\dag} - (\nabla_{\mathcal{K}_{\mathbf{p}}}^2 \mathcal{L}_n(\Omega_0))_{S^cS}(\nabla_{\mathcal{K}_{\mathbf{p}}}^2\mathcal{L}_n(\Omega_0))_{SS}^{\dag}.
\]
Now, again invoking \eqref{eq:62a},
\begin{align}\label{eq:RR}
\|\Xi(\nabla_{\mathcal{K}_{\mathbf{p}}}\mathcal{L}_n(\Omega_0)_S) \|_\infty \leq \|\Xi\|_\infty \cdot \|\nabla_{\mathcal{K}_{\mathbf{p}}}\mathcal{L}_n(\Omega_0)_S \|_{\infty} \leq \|\Xi\|_\infty O\left(r\right).
\end{align}
The infinity norm of $\Xi$ can be bounded as
\begin{align}\label{eq:Xi}
\|\Xi\|_\infty \leq& \left\| \left(\hat{Q}_{S^cS} - (\nabla_{\mathcal{K}_{\mathbf{p}}}^2\mathcal{L}_n(\Omega_0))_{S^cS}\right)\left((\hat{Q}_{SS})^{\dag} - (\nabla_{\mathcal{K}_{\mathbf{p}}}^2\mathcal{L}_n(\Omega_0))_{SS}^{\dag}\right)\right\|_\infty\\\nonumber
&+\left\|(\hat{Q}_{S^cS} - (\nabla_{\mathcal{K}_{\mathbf{p}}}^2\mathcal{L}_n(\Omega_0))_{S^cS})(\nabla_{\mathcal{K}_{\mathbf{p}}}^2\mathcal{L}_n(\Omega_0))_{SS}^{\dag}\right\|_\infty\\\nonumber
&+\left\|(\nabla_{\mathcal{K}_{\mathbf{p}}}^2\mathcal{L}_n(\Omega_0))_{S^cS}\left((\hat{Q}_{SS})^{\dag} - (\nabla_{\mathcal{K}_{\mathbf{p}}}^2\mathcal{L}_n(\Omega_0))_{SS}^{\dag}\right)\right\|\\\nonumber
&\leq \delta_1 \delta_2 + \delta_1 \|(\hat{Q}_{SS})^{\dag}\|_\infty + \delta_2 \|(\nabla_{\mathcal{K}_{\mathbf{p}}}^2\mathcal{L}_n(\Omega_0))_{S^cS}\|_\infty
\end{align}
where we have set 
\begin{align*}
\delta_1 &:= \|\hat{Q}_{S^c S} - (\nabla_{\mathcal{K}_{\mathbf{p}}}^2\mathcal{L}_n(\Omega_0))_{S^c S}\|_\infty\\
\delta_2 &:= \|(\hat{Q}_{SS})^{\dag} - (\nabla_{\mathcal{K}_{\mathbf{p}}}^2\mathcal{L}_n(\Omega_0))_{SS}^{\dag}\|_\infty.
\end{align*}
First note that by \eqref{eq:Pinfnorm}
\begin{align*}
\|(\nabla_{\mathcal{K}_{\mathbf{p}}}^2\mathcal{L}_n(\Omega_0))_{S^cS}\|_\infty &\leq \|(\nabla_{\mathcal{K}_{\mathbf{p}}}^2\mathcal{L}_n(\Omega_0))\|_\infty \\
&\leq \|\mathcal{P}_{\mathcal{K}_{\mathbf{p}}}(\nabla^2\mathcal{L}_n(\Omega_0))\mathcal{P}_{\mathcal{K}_{\mathbf{p}}}\|_\infty\\
&\leq \|\mathcal{P}_{\mathcal{K}_{\mathbf{p}}}\|_\infty^2 \|\nabla^2\mathcal{L}_n(\Omega_0)\|_\infty\\
& \leq 4K^2 \|\nabla^2\mathcal{L}_n(\Omega_0)\|_\infty \leq 4K^2 \kappa_\Sigma^2,
\end{align*}
and $\|(\hat{Q}_{SS})^\dag\|_\infty = O(1 + \delta_2)$ by the definition of $\delta_2$.

Substituting into \eqref{eq:Xi} gives
\begin{equation}\label{eq:Xi2}
\|\Xi\|_\infty \leq O(\delta_1 \delta_2) + O(\delta_1) + O(\delta_2).
\end{equation}

We bound $\delta_1$ and $\delta_2$ in the following lemma, proved in Section \ref{supp:Delts}.
\begin{lemma}\label{lem:Delts}
Under the conditions of Lemma \ref{lemm:Lemm7}, 
\begin{align*}
\delta_1 &= O(dr)\\
\delta_2 &= O(dr).
\end{align*}
\end{lemma}
Applying Lemma \ref{lem:Delts} to \eqref{eq:Xi2}, we obtain
\[
\|\Xi\|_\infty = O(d^2 r^2) + O(dr)+ O(dr) = O(dr)
\]
and substituting into \eqref{eq:RR}
\[
\|\Xi(\nabla_{\mathcal{K}_{\mathbf{p}}}\mathcal{L}_n(\Omega_0)_S) \|_\infty = O(dr)\cdot O(r) = O(r),
\]
since $dr = o(1)$ by our assumption that $n \min_k m_k \geq c_0 d^2 \log p$ for some $c_0$ large enough.

Therefore, substituting into \eqref{eq:HessQuant} we obtain 
\begin{align*}
\|\hat{Q}_{S^c S} (\hat{Q}_{SS})^{\dag} \nabla_{\mathcal{K}_{\mathbf{p}}} \mathcal{L}_n (\Omega_0)_S\|_\infty = O(r) + O(r) = O(r),
\end{align*}
proving the desired condition \eqref{eq:62b} holds.} Hence the conditions of Lemma \ref{lem:Conds} hold under event $\mathcal{A}$, and $\hat{\Omega}_{\mathrm{oracle}}$ is the unique global minimizer of the complete objective \eqref{eq:nonconvobv}. 

The Frobenius and spectral norm bounds follow from the identities 
\[
\|\hat{\Omega} - \Omega_0\|_F \leq \sqrt{s + p} \|\hat{\Omega} - \Omega_0\|_{\max}
\]
and
\[
\|\hat{\Omega} - \Omega_0\|_2 \leq d \|\hat{\Omega} - \Omega_0\|_{\max},
\]
where the latter identity follows by symmetry of $\Omega$.

{
\subsection{Proof of Lemma \ref{lem:Delts}: Bound on $\delta_1, \delta_2$}\label{supp:Delts}
\begin{proof}
Consider that
\begin{align*}
\hat{Q} - \nabla_{\mathcal{K}_{\mathbf{p}}}^2\mathcal{L}_n(\Omega_0) &= \int_0^1 \nabla^2_{\mathcal{K}_{\mathbf{p}}} \mathcal{L}_n\left(\Omega_0 + t(\hat{\Omega}_{\mathrm{oracle}}-\Omega_0)\right)dt - \nabla_{\mathcal{K}_{\mathbf{p}}}^2\mathcal{L}_n(\Omega_0)\\
&= \mathcal{P}_{\mathcal{K}_{\mathbf{p}}} \left(\int_0^1 \nabla^2 \mathcal{L}_n\left(\Omega_0 + t(\hat{\Omega}_{\mathrm{oracle}}-\Omega_0)\right)dt - \nabla^2 \mathcal{L}_n(\Omega_0) \right) \mathcal{P}_{\mathcal{K}_{\mathbf{p}}}.
\end{align*}
Hence, since $\| \mathcal{P}_{\mathcal{K}_{\mathbf{p}}}\|_\infty \leq 2K$ by \eqref{eq:Pinfnorm},
\begin{align*}
&\|\hat{Q} - \nabla_{\mathcal{K}_{\mathbf{p}}}^2\mathcal{L}_n(\Omega_0)\|_\infty\\
&\leq \| \mathcal{P}_{\mathcal{K}_{\mathbf{p}}}\|_\infty^2 \left\| \int_0^1 \nabla^2 \mathcal{L}_n\left(\Omega_0 + t(\hat{\Omega}_{\mathrm{oracle}}-\Omega_0)\right)dt - \nabla^2 \mathcal{L}_n(\Omega_0) \right\|_\infty\\
&\leq 4K^2\left\|\int_0^1 (\Omega_0 + t(\hat{\Omega}_{\mathrm{oracle}} - \Omega_0))^{-1} \otimes (\Omega_0 + t(\hat{\Omega}_{\mathrm{oracle}} - \Omega_0))^{-1} - \Omega_0^{-1} \otimes \Omega_0^{-1}dt\right\|_\infty\\
&\leq 4K^2\int_0^1 \left\|(\Omega_0 + t(\hat{\Omega}_{\mathrm{oracle}} - \Omega_0))^{-1} \otimes (\Omega_0 + t(\hat{\Omega}_{\mathrm{oracle}} - \Omega_0))^{-1} - \Omega_0^{-1} \otimes \Omega_0^{-1}\right\|_\infty dt.
\end{align*}
By Lemma \ref{lemm:Lemm7}, for $t \in [0,1]$, 
\[
\|\Omega_0 + t(\hat{\Omega}_{\mathrm{oracle}} - \Omega_0) - \Omega_0 \|_\infty = t\|\hat{\Omega}_{\mathrm{oracle}} - \Omega_0\|_\infty \leq d\|\hat{\Omega}_{\mathrm{oracle}} - \Omega_0\|_{\max} \leq dr.
\]
We make use of the following matrix inequalities \citep{LOH}. For any invertible $A, B \in \mathbb{R}^{p\times p}$ and matrix norm $\|\cdot\|$, 
\begin{align}
\label{eq:Lem11}
\|A^{-1} - B^{-1}\| \leq \frac{\|A^{-1}\|^2 \|A-B\|}{1 - \|A^{-1}\| \|A-B\|} = O(\|A^{-1}\|^2 \|A - B\|).
\end{align}
if $\|A^{-1}\| \|A - B\| \leq 1/2$.
For any $A$ and $B$ matrices of equal dimension we have
\begin{align}\label{eq:Lem13}
\|A\otimes A - B\otimes B\|_\infty \leq \|A - B\|_\infty^2 + 2\min(\|A\|_\infty, \|B\|_\infty)\cdot \|A-B\|_\infty.
\end{align}
Applying \eqref{eq:Lem11} we get
\begin{align*}
&\left\|\left(\Omega_0 +  t(\hat{\Omega}_{\mathrm{oracle}} - \Omega_0)\right)^{-1} - \Omega_0^{-1} \right\|_\infty \\
&\leq O\left(\left\|\Omega_0^{-1}\right\|_\infty^2\|\Omega_0 + t(\hat{\Omega}_{\mathrm{oracle}} - \Omega_0) - \Omega_0 \|_\infty\right)\\
& \leq  O(dr),
\end{align*}
since $\|\Omega_0^{-1}\|_{\infty}= \|\Sigma_0\|_{\infty}$ is bounded by $\kappa_\Sigma$.
Applying \eqref{eq:Lem13} to this yields
\[
\|\hat{Q} - \nabla_{\mathcal{K}_{\mathbf{p}}}^2\mathcal{L}_n(\Omega_0)\|_\infty = O(dr),
\]
which gives
\[
\delta_1 = \|\hat{Q}_{S^c S} - (\nabla_{\mathcal{K}_{\mathbf{p}}}^2\mathcal{L}_n(\Omega_0))_{S^c S}\|_\infty = O(dr)
\]
and 
\begin{equation}\label{eq:lastEq}
\|\hat{Q}_{S S} - (\nabla_{\mathcal{K}_{\mathbf{p}}}^2\mathcal{L}_n(\Omega_0))_{S S}\|_\infty = O(dr).
\end{equation}

Finally, recall that the projection matrix onto $\mathcal{K}_{\mathbf{p}}$ can be written as $U U^T$ with $U^T U = I$ so
\begin{align*}
\delta_2 &= \|(\hat{Q}_{SS})^{\dag} - (\nabla_{\mathcal{K}_{\mathbf{p}}}^2\mathcal{L}_n(\Omega_0))_{SS}^{\dag}\|_\infty\\& =\left\|U\left(\left(U^T \hat{Q}_{SS} U\right)^{-1} - \left(U^T \nabla_{\mathcal{K}_{\mathbf{p}}}^2\mathcal{L}_n(\Omega_0)_{SS}U\right)^{-1}\right) U^T \right\|_\infty.
\end{align*}
By the matrix expansion \eqref{eq:matExp} we then have
\begin{align*}
&(\hat{Q}_{SS})^{\dag} - (\nabla_{\mathcal{K}_{\mathbf{p}}}^2\mathcal{L}_n(\Omega_0))_{SS}^{\dag}\\&=U\left[\left(U^T \hat{Q}_{SS} U\right)^{-1} - \left(U^T \nabla_{\mathcal{K}_{\mathbf{p}}}^2\mathcal{L}_n(\Omega_0)_{SS}U\right)^{-1}\right]U^T\\
& = U\left[\sum_{\ell = 1}^\infty \left[\left( -\left(U^T \nabla_{\mathcal{K}_{\mathbf{p}}}^2\mathcal{L}_n(\Omega_0)_{SS}U\right)^{-1} \left(U^T \hat{Q}_{SS} U -U^T \nabla_{\mathcal{K}_{\mathbf{p}}}^2\mathcal{L}_n(\Omega_0)_{SS}U \right)\right)^\ell\right.\right.\\&\qquad\qquad\qquad\qquad\qquad \cdot\left.\left. \left(U^T \nabla_{\mathcal{K}_{\mathbf{p}}}^2\mathcal{L}_n(\Omega_0)_{SS}U\right)^{-1}\right]\right]U^T\\
&= \sum_{\ell=1}^\infty \left( - \left((\nabla_{\mathcal{K}_{\mathbf{p}}}^2\mathcal{L}_n(\Omega_0))_{SS}^{\dag}\right)\left(\hat{Q}_{S S} - (\nabla_{\mathcal{K}_{\mathbf{p}}}^2\mathcal{L}_n(\Omega_0))_{S S}\right) \right)^\ell\left((\nabla_{\mathcal{K}_{\mathbf{p}}}^2\mathcal{L}_n(\Omega_0))_{SS}^{\dag}\right).
\end{align*}
We can then use the bound \eqref{eq:lastEq} to obtain
\begin{align*}
\delta_2 &= \|(\hat{Q}_{SS})^{\dag} - (\nabla_{\mathcal{K}_{\mathbf{p}}}^2\mathcal{L}_n(\Omega_0))_{SS}^{\dag}\|_\infty\\
&\leq\sum_{\ell = 1}^\infty O\left(\|\hat{Q}_{S S} - (\nabla_{\mathcal{K}_{\mathbf{p}}}^2\mathcal{L}_n(\Omega_0))_{S S}\|_\infty^\ell\right) \\
&= O\left(\frac{\|\hat{Q}_{S S} - (\nabla_{\mathcal{K}_{\mathbf{p}}}^2\mathcal{L}_n(\Omega_0))_{S S}\|_\infty}{1 - \|\hat{Q}_{S S} - (\nabla_{\mathcal{K}_{\mathbf{p}}}^2\mathcal{L}_n(\Omega_0))_{S S}\|_\infty}\right)\\
& = O(dr),
\end{align*}
since $dr = o(1)$.

\qed\end{proof}
}

\section{Numerical Convergence of TG-ISTA}
\label{supp:conv}
The following theorem shows that the iterates of the TG-ISTA implementation of TeraLasso converge geometrically to the global minimum:
\begin{theorem}\label{thm:conv}
Let $\rho_k \geq 0$ for all $k$ and let $\Omega_{\mathrm{init}}$ be the initialization of the TG-ISTA implementation of TeraLasso (Algorithm \ref{alg:lowlevel}). Let
\begin{align*}
a = \frac{1}{\sum_{k = 1}^K \|S_k\|_2 +  d_k \rho_k}, \qquad b = \|\Omega^*\|_2 + \|\Omega_{\mathrm{init}} - \Omega^*\|_F,
\end{align*}
and assume $\zeta_t \leq a^2$ for all $t$. Suppose further that $\Omega^*$ is the global optimum. Then
\[
\|\Omega_{t+1} - \Omega^* \|_F \leq \max\left\{\left|1-\frac{\zeta_t}{b^2}\right|, \left|1 - \frac{\zeta_t}{a^2}\right|\right\} \|\Omega_t - \Omega^*\|_F.
\]
Furthermore, the step size $\zeta_t$ which yields an optimal worst-case contraction bound $s(\zeta_t)$ is $\zeta = \frac{2}{a^{-2} + b^{-2}}$. The corresponding optimal worst-case contraction bound is 
\begin{equation}\label{eq:soptt}
s(\zeta) = 1 - \frac{2}{1 + \frac{b^2}{a^2}}.
\end{equation}
\end{theorem}
Our proof uses results on the structure of the Kronecker sum subspace
to extend to our subspace restricted setting the methodology that
\cite{GISTA} used to derive the unstructured GLasso convergence rates.

We decompose the claims of Theorem \ref{thm:conv} into the following two theorems which we prove separately.
\begin{theorem}\label{thm:conv2}
Assume that the iterates $\Omega_t$ of Algorithm \ref{alg:lowlevel} satisfy $a I \preceq \Omega_t \preceq bI$, for all $t$, for some fixed constants $0 < a < b < \infty$. Suppose further that $\Omega^*$ is the global optimum. If $\zeta_t \leq a^2$ for all $t$, then
\[
\|\Omega_{t+1} - \Omega^* \|_F \leq \max\left\{\left|1-\frac{\zeta_t}{b^2}\right|, \left|1 - \frac{\zeta_t}{a^2}\right|\right\} \|\Omega_t - \Omega^*\|_F.
\]
Furthermore, the step size $\zeta_t$ which yields an optimal worst-case contraction bound $s(\zeta_t)$ is $\zeta = \frac{2}{a^{-2} + b^{-2}}$. The corresponding optimal worst-case contraction bound is 
\begin{equation}\label{eq:soptt}
s(\zeta) = 1 - \frac{2}{1 + \frac{b^2}{a^2}}.
\end{equation}
\end{theorem}
\begin{theorem}\label{thm:eigiter}
Let $\rho_k \geq 0$ for all $k$ and let $\Omega_{\mathrm{init}}$ be the initialization of the TG-ISTA implementation of TeraLasso (Algorithm \ref{alg:lowlevel}). Let
\begin{align*}
a = \frac{1}{\sum_{k = 1}^K \|S_k\|_2 +  d_k \rho_k}, \qquad b = \|\Omega^*\|_2 + \|\Omega_{\mathrm{init}} - \Omega^*\|_F,
\end{align*}
and assume $\zeta_t \leq a^2$ for all $t$. Then the iterates $\Omega_t$ of Algorithm \ref{alg:lowlevel} satisfy $a I \preceq \Omega_t \preceq b I$ for all $t$. 
\end{theorem}
Observe that by Theorem \ref{thm:eigiter}, the worst case contraction factor \eqref{eq:soptt} 
\[
s(\zeta) = 1 - \frac{2}{1 + (\|\Omega^*\|_2 + \|\Omega_{\mathrm{init}} - \Omega^*\|_F)^2 (\sum_{k=1}^K \|S_k \|_2 + d_k \rho_k)^2}
\]
scales at most as $s(\zeta) = O(1 - \frac{2}{1 + K^2})$ for $\|\Omega^*\|_2, \|\Sigma_0\|_2$ of fixed order, since $\|S_k\|_2  \sim \|\Sigma_0\|_2$ with high probability. 

Let $T$ be the number of iterations required for $\|\Omega_{T} - \Omega^*\|_F \leq \|\Omega^* - \hat{\Omega}\|_F$ to hold, i.e. for the optimization error to be smaller than the statistical error. By Theorem \ref{Thm:1}, we require
\begin{equation}\label{eq:neww}
\|\Omega_{T} - \Omega^*\|_F^2 \leq C_1 K^2 (s+p) \frac{\log p}{n \min_k m_k}. 
\end{equation}
Using worst case contraction factor $s(\zeta)$, \eqref{eq:neww} will hold for $T$ such that (with high probability)
\[
\|\Omega_{\mathrm{init}} - \Omega^*\|_F^2 \left(1 - \frac{2}{1 + \frac{b^2}{a^2}}\right)^{2T} \leq C_1 K^2 (s+p) \frac{\log p}{n \min_k m_k}.
\]
Taking the logarithm of both sides and using $s(\zeta) = O(1 - \frac{2}{1 + K^2})$, we have that the optimization error is guaranteed to equal the statistical error after $T$ iterations, where
\begin{align*}
T = O_p\left(\frac{2 \log K + \log(s+p) + \log \log p - \log (n \min_k m_k)}{\log \left(1 - \frac{2}{1 + K^2}\right)}\right).
\end{align*}



\subsection{Proof of Theorem \ref{thm:conv2}}
\label{app:Conv}



\begin{proofof2}
For convenience, define the Kronecker sum shrinkage operator as
\begin{equation}\label{eq:shrinkOp}
\mathrm{shrink}^{-}_{\rho}(A) = \mathrm{shrink}^{-}_{\rho_1}(A^{(1)}) \oplus \dots \oplus \mathrm{shrink}^{-}_{\rho_K}(A^{(K)})
\end{equation}
for $A = A^{(1)} \oplus \dots \oplus A^{(K)} \in \mathcal{K}_{\mathbf{p}}$ and $\mathbf{\rho} = [\rho_1,\dots, \rho_K]$ with all $\rho_k \geq 0$. 
Note that $\mathrm{shrink}^{-}_{\rho}(A) = \arg \min_{\Omega \in \mathcal{K}_{\mathbf{p}}} \left\{ \frac{1}{2}\left\| \Omega - A \right\|_F^2 + \sum_{k=1}^K m_k  \rho_k  |{\Psi}_k|_{1,\off} \right\}$.
Since $\sum_{k=1}^K m_k  \rho_k  |{\Psi}_k|_{1,\off}$ is a convex function on $\mathcal{K}_{\mathbf{p}}$, and since $\mathcal{K}_{\mathbf{p}}$ is a linear subspace, $\mathrm{shrink}^{-}_{\epsilon}(\cdot)$ is a proximal operator by definition.

Recall that we can write the TG-ISTA update \eqref{eq:Ogstp} using this Kronecker sum shrinkage operator as
\begin{align*}
&\Omega_{t+1} = \arg \min_{\Omega \in \mathcal{K}_{\mathbf{p}}} \left\{ \frac{1}{2}\left\| \Omega - \left(\Omega_t - \zeta_t\left(\tilde{S} - G^t\right)\right) \right\|_F^2 + \zeta_t \sum_{k=1}^K m_k  \rho_k  |{\Psi}_k|_{1,\off} \right\}\\
&= \arg \min_{\Omega \in \mathcal{K}_{\mathbf{p}}} \left\{ \frac{1}{2}\left\| \Omega - \left(\Omega_t - \zeta_t\left(\mathrm{Proj}_{\mathcal{K}_{\mathbf{p}}}({\hat{S}} - \Omega_t^{-1}\right)\right) \right\|_F^2 + \zeta_t \sum_{k=1}^K m_k  \rho_k  |{\Psi}_k|_{1,\off} \right\}\\
&=\mathrm{shrink}^{-}_{\zeta_t \rho}(\Omega_t - \zeta_t\mathrm{Proj}_{\mathcal{K}_{\mathbf{p}}}(\hat{S}-\Omega_t^{-1})),
\end{align*}
where $\hat{S}$ is the sample covariance \eqref{eq::gram} and $\tilde{S} = \mathrm{Proj}_{\mathcal{K}_{\mathbf{p}}}(\hat{S})$ is its projection onto $\mathcal{K}_{\mathbf{p}}$ \eqref{eq:ingrad}.

By convexity in $\mathcal{K}_{\mathbf{p}}$ and Theorem \ref{Thm:Conv}, the optimal point $\Omega^*_\rho$ is a fixed point of the ISTA iteration (\cite{combettes2005signal}, Prop 3.1). Thus,
\[
\Omega_\rho^* = \mathrm{shrink}^{-}_{\zeta_t \rho}(\Omega^*_\rho - \zeta_t \mathrm{Proj}_{\mathcal{K}_{\mathbf{p}}}(\hat{S} - (\Omega^*_\rho)^{-1}).
\]
Since proximal operators are not expansive \citep{combettes2005signal}, we have
\begin{align*}
\|\Omega_{t+1} &- \Omega_\rho^*\|_F \\&= \left\|\mathrm{shrink}^{-}_{\zeta_t \rho}(\Omega_t - \zeta_t \mathrm{Proj}_{\mathcal{K}_{\mathbf{p}}}(\hat{S} - \Omega_t^{-1}))\right.\\
&\qquad\qquad\qquad \left.- \mathrm{shrink}^{-}_{\zeta_t \rho}(\Omega_\rho^* - \zeta_t \mathrm{Proj}_{\mathcal{K}_{\mathbf{p}}}(\hat{S} - (\Omega_\rho^*)^{-1}))\right\|_F \\
&\leq \|\Omega_t - \zeta_t \mathrm{Proj}_{\mathcal{K}_{\mathbf{p}}}(\hat{S} - \Omega_t^{-1}) - (\Omega_\rho^* - \zeta_t \mathrm{Proj}_{\mathcal{K}_{\mathbf{p}}}(\hat{S} - (\Omega_\rho^*)^{-1}))\|_F\\
&= \|\Omega_t + \zeta_t \mathrm{Proj}_{\mathcal{K}_{\mathbf{p}}}(\Omega_t^{-1}) - (\Omega_\rho^* + \zeta_t \mathrm{Proj}_{\mathcal{K}_{\mathbf{p}}}((\Omega_\rho^*)^{-1}))\|_F.
\end{align*}
For $\gamma > 0$ define $h_\gamma: \mathcal{K}_{\mathbf{p}}^{\sharp} \rightarrow \mathcal{K}_{\mathbf{p}}^{\sharp}$ by
\[
h_\gamma (\Omega) = \mathrm{vec}(\Omega) + \mathrm{vec}(\gamma \mathrm{Proj}_{\mathcal{K}_{\mathbf{p}}}(\Omega^{-1})).
\]
Since $\partial \Omega^{-1}/\partial \Omega= - \Omega^{-1} \otimes \Omega^{-1}$, 
\[
\frac{\partial \mathrm{Proj}_{\mathcal{K}_{\mathbf{p}}}(\Omega^{-1})}{\partial \Omega} = - P(\Omega^{-1} \otimes \Omega^{-1})P^T
\]
where $P$ is the projection matrix that projects $\mathrm{vec}(\Omega)$ onto the vectorized subspace $\mathcal{K}_{\mathbf{p}}$. Thus, we have the Jacobian (valid for all $\Omega \in \mathcal{K}_{\mathbf{p}}^{\sharp}$)
\[
J_{h_\gamma}(\Omega) = PP^T - \gamma P(\Omega^{-1} \otimes \Omega^{-1})P^T.
\]
Recall that if $h : U \subset \mathbb{R}^n \rightarrow \mathbb{R}^m$ is a differentiable mapping, then if $x,y \in U$ and $U$ is convex, then if $J_h(\cdot)$ is the Jacobian of $h$,
\[
\|h(x) - h(y)\| \leq \sup_{c \in [0,1]} \|J_h(cx + (1-c)y)\|\|x-y\|.
\]
Thus, letting $Z_{t,c} = \mathrm{vec}(c \Omega_t + (1-c) \Omega^*_\rho)$, for $c \in [0,1]$ we have
\begin{align*}
\|h_{\zeta_t} (x) - h_{\zeta_t}(y)\| \leq \sup_{c \in [0,1]} \|PP^T - \zeta_t P(Z_{t,c}^{-1} \otimes Z_{t,c}^{-1})P^T\|\|\Omega_t - \Omega^*_\rho\|_F.
\end{align*}
By Weyl's inequality, $\lambda_{\max}(Z_{t,c}) \leq \max\{ \|\Omega_t\|, \|\Omega_\rho^*\|\}$ and 
\[
\lambda_{\min} (Z_{t,c}) \geq \min\{ \lambda_{\min}(\Omega_t), \lambda_{\min}(\Omega_\rho^*)\}.
\] Furthermore, note that for any $Y$ and projection matrix $P$
\[
 \lambda_{\max}(PYP^T) \leq \lambda_{\max}(Y).
\]
We then have
\begin{align*}
\\|PP^T - \zeta_t P(Z_{t,c}^{-1} \otimes Z_{t,c}^{-1})P^T\| 
&\leq \|I_{p^2} - \zeta_t Z_{t,c}^{-1} \otimes Z_{t,c}^{-1}\|
\\&\leq \max \left\{ \left| 1 - \frac{\zeta_t}{b^2}\right|,\left|1 - \frac{\zeta_t}{a^2}\right|\right\},
\end{align*}
where the latter inequality comes from \citep{GISTA}. Thus, 
\begin{align*}
\|\Omega_{t+1}& - \Omega_\rho^*\|_F \leq s(\zeta_t) \|\Omega_{t} - \Omega_\rho^*\|_F\\
\text{ and  } \; \; 
s(\zeta) &=  \max \left\{ \left| 1 - \frac{\zeta}{b^2}\right|,\left|1 - \frac{\zeta}{a^2}\right|\right\}
\end{align*}
as desired. Algorithm \ref{alg:lowlevel} will then converge if $s(\zeta_t) \in (0,1)$ for all $t$. The minimum of $s(\zeta)$ occurs at $\zeta = \frac{2}{a^{-2} + b^{-2}}$, completing the proof of Theorem \ref{thm:conv2}. 
\end{proofof2}

\subsection{Proof of Theorem \ref{thm:eigiter}}

\begin{proofof2}
We first prove the following properties of the Kronecker sum projection operator.
\begin{lemma}\label{lem:Pcont}
For 
any $A \in \mathbb{R}^{p\times p}$ and orthogonal matrices $U_k \in \mathbb{R}^{d_k \times d_k}$, let $U = U_1 \otimes \dots \otimes U_K \in \mathcal{K}_{\mathbf{p}}$. Then 
\[
\mathrm{Proj}_{\mathcal{K}_{\mathbf{p}}}(A) = U \mathrm{Proj}_{\mathcal{K}_{\mathbf{p}}}(U^T A U) U^T.
\]
Furthermore, if the eigendecomposition of $A$ is of the form $A = (U_1 \otimes \dots \otimes U_K) \Lambda (U_1 \otimes \dots \otimes U_K)^T$ with $\Lambda = \mathrm{diag}(\lambda_1, \dots, \lambda_p)$, we have
\[
\mathrm{Proj}_{\mathcal{K}_{\mathbf{p}}}(A) = U \mathrm{Proj}_{\mathcal{K}_{\mathbf{p}}}(\Lambda) U^T
\]
and
\[
\lambda_{\min}(A) \leq \lambda_{\min}(\mathrm{Proj}_{\mathcal{K}_{\mathbf{p}}}(A)) \leq \lambda_{\max}(\mathrm{Proj}_{\mathcal{K}_{\mathbf{p}}}(A)) \leq \lambda_{\max}(A).
\]
\end{lemma}
\begin{proof}
Recall
\[
\mathrm{Proj}_{\mathcal{K}_{\mathbf{p}}}(A) = \arg \min_{M \in \mathcal{K}_{\mathbf{p}}} \|A - B\|_F^2 = \arg \min_{B \in \mathcal{K}_{\mathbf{p}}} \|U^T A U - U^T B U\|_F^2
\]
since $U^T A U = \Lambda$ and the Frobenius norm is unitarily invariant. Now, note that for any matrix $B = B_1 \oplus \dots \oplus B_K \in \mathcal{K}_{\mathbf{p}}$,
\begin{align*}
(U_1 \otimes \dots &\otimes U_K)^T B (U_1 \otimes \dots \otimes U_K)\\ &= \sum_{k = 1}^K (U_1 \otimes \dots \otimes U_K)^T (I_{[d_{1:k-1}]}\otimes B_k \otimes I_{[d_{k+1}:K]}) (U_1 \otimes \dots \otimes U_K)\\
&= \sum_{k = 1}^K  I_{[d_{1:k-1}]}\otimes U_k^T B_k U_k \otimes I_{[d_{k+1}:K]}\\
&= (U_1^T B_1 U_1) \oplus \dots \oplus  (U_K^T B_K U_K)\\
&\in \mathcal{K}_{\mathbf{p}},
\end{align*}
since $U_k^T I_{d_k} U_k  = I_{d_k}$. Since $U^TB U \in \mathcal{K}_{\mathbf{p}}$, the constraint $B \in \mathcal{K}_{\mathbf{p}}$ can be moved to $C = U^T B U$, giving
\begin{align*}
\mathrm{Proj}_{\mathcal{K}_{\mathbf{p}}}(A) &= U(\arg \min_{C \in \mathcal{K}_{\mathbf{p}}} \|U^T A U - C\|_F^2)U^T\\
&= U(\mathrm{Proj}_{\mathcal{K}_{\mathbf{p}}}(U^T A U))U^T.
\end{align*}
If $A = (U_1 \otimes \dots \otimes U_K) \Lambda (U_1 \otimes \dots \otimes U_K)^T$, then $U^T A U = \Lambda$, completing the first part of the proof.
As shown in Lemma \ref{lem:Projj}, $\mathrm{Proj}_{\mathcal{K}_{\mathbf{p}}}(\Lambda)$ is a diagonal matrix whose entries are weighted averages of the diagonal elements $\lambda_i$. Hence 
\[
\min_{i} \lambda_i \leq \min_i [\mathrm{Proj}_{\mathcal{K}_{\mathbf{p}}}(\Lambda)]_{ii} \leq \max_i [\mathrm{Proj}_{\mathcal{K}_{\mathbf{p}}}(\Lambda)]_{ii} \leq \max_i \lambda_i.
\]
Since $\mathrm{Proj}_{\mathcal{K}_{\mathbf{p}}}(\Lambda)$ gives the eigenvalues of $\mathrm{Proj}_{\mathcal{K}_{\mathbf{p}}}(A)$ by the orthogonality of $U$, this completes the proof. 
\qed
\end{proof}


\begin{lemma}\label{lem:4}
Let $0 < a < b$ be given positive constants and let $\zeta_t > 0$. Assume $a I \preceq \Omega_t \preceq bI$. Then for 
\[
\Omega_{t+1/2} := \Omega_t - \zeta_t( \mathrm{Proj}_{\mathcal{K}_{\mathbf{p}}}(\hat{S} -\Omega_t^{-1}))
\]
we have
\begin{align*}
\lambda_{\min}(\Omega_{t+1/2}) \geq \left\{\begin{array}{lc} 2\sqrt{\zeta_t} -\zeta_t \lambda_{\max}(\hat{S}) & \mathrm{if} \: a \leq \sqrt{\zeta_t} \leq b\\
\min\left(a + \frac{\zeta_t}{a}, b + \frac{\zeta_t}{b}\right) - \zeta_t \lambda_{\max}(\hat{S}) & \mathrm{o.w.} \end{array}\right. 
\end{align*}
\end{lemma}
\begin{proof}
Let $U\Gamma U^T = \Omega_t$ be the eigendecomposition of $\Omega_t$,
where $\Gamma = \mathrm{diag}(\gamma_1, \dots, \gamma_p)$. Then all $b
\geq \gamma_i \geq a > 0$. Since $\Omega_t \in
\mathcal{K}_{\mathbf{p}}$, by the eigendecomposition property in
Appendix \ref{App:Ident} 
we have $U = U_1 \otimes \dots \otimes U_K$ and $\Gamma \in \mathcal{K}_{\mathbf{p}}$, letting us apply Lemma \ref{lem:Pcont}: 
\begin{align*}
\Omega_{t+1/2} &= \Omega_t - \zeta_t(\mathrm{Proj}_{\mathcal{K}_{\mathbf{p}}}(\hat{S}) - \mathrm{Proj}_{\mathcal{K}_{\mathbf{p}}}(\Omega_t^{-1}))\\
&= U \Gamma U^T - \zeta_t (\mathrm{Proj}_{\mathcal{K}_{\mathbf{p}}}(\hat{S}) - U \mathrm{Proj}_{\mathcal{K}_{\mathbf{p}}}(\Gamma^{-1})U^T)\\
&= U\left(\Gamma - \zeta_t(U^T \mathrm{Proj}_{\mathcal{K}_{\mathbf{p}}}(\hat{S}) U -\mathrm{Proj}_{\mathcal{K}_{\mathbf{p}}}(\Gamma^{-1}))\right)U^T\\
&= U\left(\mathrm{Proj}_{\mathcal{K}_{\mathbf{p}}}(\Gamma) - \zeta_t\left(\mathrm{Proj}_{\mathcal{K}_{\mathbf{p}}}(U^T \hat{S} U) -\mathrm{Proj}_{\mathcal{K}_{\mathbf{p}}}(\Gamma^{-1})\right)\right)U^T\\
&= \mathrm{Proj}_{\mathcal{K}_{\mathbf{p}}}\left(U(\Gamma + \zeta \Gamma^{-1} - \zeta_t(U^T \hat{S} U))U^T\right)\\
&=\mathrm{Proj}_{\mathcal{K}_{\mathbf{p}}}(\tilde{\Omega}_{t+1/2}),
\end{align*}
where we set $\tilde{\Omega}_{t+1/2} = U(\Gamma + \zeta \Gamma^{-1} - \zeta_t(U^T \hat{S} U))U^T$ and recall the linearity of the projection operator $\mathrm{Proj}_{\mathcal{K}_{\mathbf{p}}}(\cdot)$ (Lemma \ref{lem:Projj}). 
By Weyl's inequality, 
\begin{align*}
\gamma_1 + \frac{\zeta_t}{\gamma_1} - \zeta_t \lambda_{\max}(\hat{S}) \leq \lambda_{\min}(\tilde{\Omega}_{t + 1/2}). 
\end{align*}
By Lemma \ref{lem:Pcont},
\begin{align*}
\gamma_1 + \frac{\zeta_t}{\gamma_1} - \zeta_t \lambda_{\max}(\hat{S}) \leq \lambda_{\min}({\Omega}_{t + 1/2}). 
\end{align*}
Note that the only extremum of the function $f(x) = x + \frac{\zeta_t}{x}$ over $a \leq x \leq b$ is a global minimum at $x = \sqrt{\zeta_t}$. Hence 
\begin{align*}
\inf_{a \leq x \leq b} x + \frac{\zeta_t}{x} &= \left\{\begin{array}{lc} 2 \sqrt{\zeta_t} & \mathrm{if} \: a \leq \sqrt{\zeta_t} \leq b\\
\min \left(a + \frac{\zeta_t}{a}, b + \frac{\zeta_t}{b}\right)& \mathrm{o.w.} \end{array}\right. 
\end{align*}
By our assumption, $a \leq \gamma_1 \leq b$. Thus
\begin{align*}
\lambda_{\min}({\Omega}_{t + 1/2}) &\geq \left\{\begin{array}{lc} 2 \sqrt{\zeta_t}- \zeta_t \lambda_{\max}(\hat{S}) & \mathrm{if} \: a \leq \sqrt{\zeta_t} \leq b\\
\min \left(a + \frac{\zeta_t}{a}, b + \frac{\zeta_t}{b}\right)- \zeta_t \lambda_{\max}(\hat{S})& \mathrm{o.w.} \end{array}\right. 
\end{align*}
as desired, completing the proof.
\qed
\end{proof}

We then have the following lemma.
\begin{lemma}\label{lem:6}
For $A  \in \mathcal{K}_{\mathbf{p}}^{\sharp}$ and $\mathbf{\epsilon} = [\epsilon_1, \dots, \epsilon_K]$ with $\epsilon_k \geq 0$:
\[
\lambda_{\min} (A) - \sum_{k = 1}^K d_k \epsilon_k \leq \lambda_{\min}(\mathrm{shrink}^{-}_{\epsilon}(A))
\]
\end{lemma}
\begin{proof}
Since by definition \eqref{eq:shrinkOp}
\[
\mathrm{shrink}^{-}_{\epsilon}(A) = \mathrm{shrink}^{-}_{\epsilon_1}(A^{(1)}) \oplus \dots \oplus \mathrm{shrink}^{-}_{\epsilon_K}(A^{(K)}),
\]
we can use the fact that the eigenvalues of a Kronecker sum are the sums of the eigenvalues to show
\[
\lambda_{\min}(\mathrm{shrink}^{-}_{\epsilon}(A)) = \sum_{k = 1}^K \lambda_{\min}(\mathrm{shrink}^{-}_{\epsilon_k}(A^{(k)})).
\]
We have used the fact that $A$ is positive definite since it is in $\mathcal{K}_{\mathbf{p}}^{\sharp}$.

Via Weyl's inequality and the proof of Lemma 6 in \citep{GISTA}, 
$$\lambda_{\min}(\mathrm{shrink}^{-}_{\epsilon_k}(A^{(k)})) \geq
\lambda_{\min}(A^{(k)}) - d_k \epsilon_k.$$ 
Hence,
\[
\lambda_{\min}(\mathrm{shrink}^{-}_{\epsilon}(A)) \geq \sum_{k = 1}^K \lambda_{\min} (A^{(k)}) - \sum_{k = 1}^K d_k \epsilon_k = \lambda_{\min} (A) - \sum_{k = 1}^K d_k \epsilon_k 
\]

\qed
\end{proof}

\subsubsection{Proof of Theorem \ref{thm:eigiter}}

To prove the lower inequality in Theorem \ref{thm:eigiter}, we show the following. 
\begin{lemma}\label{lem:7}
Let $\rho = [\rho_1, \dots, \rho_K]$ with all $\rho_i > 0$. Define
\[
\chi = \sum_{k = 1}^K d_k \rho_k
\]
 and let $\alpha = \frac{1}{\|\hat{S}\|_2 + \chi} < b'$. Assume $\alpha I \preceq \Omega_{t+1}$. Then $\alpha I \preceq \Omega_{t+1}$ for every $0 < \zeta_t < \alpha^2$.
\end{lemma}
\begin{proof}
Since $\zeta_t < \alpha^2$, $\sqrt{\zeta_t} \notin [\alpha,b']$, and $\min\left(\alpha + \frac{\zeta_t}{\alpha}, b' + \frac{\zeta_t}{b'}\right) = \alpha + \frac{\zeta_t}{\alpha}$. 
Lemma \ref{lem:4} then implies that
\begin{align*}
\lambda_{\min}(\Omega_{t + 1/2}) &\geq \min\left(\alpha + \frac{\zeta_t}{\alpha}, b' + \frac{\zeta_t}{b'}\right) - \zeta_t \lambda_{\max}(\hat{S})\\
&= \alpha + \frac{\zeta_t}{\alpha} - \zeta_t \lambda_{\max}(\hat{S}).
\end{align*}
By Lemma \ref{lem:6}, 
\begin{align*}
\lambda_{\min}(\Omega_{t+1}) &= \lambda_{\min}\left(\mathrm{shrink}^{-}_{\zeta_t \rho}(\Omega_{t + 1/2})\right)\\
& \geq \lambda_{\min}(\Omega_{t+1/2}) - \zeta_t \chi\\
&\geq \alpha + \frac{\zeta_t}{\alpha} - \zeta_t \lambda_{\max}(\hat{S}) - \zeta_t \chi.
\end{align*}
Hence, since $\zeta_t > 0$, $\lambda_{\min}(\Omega_{t+1}) \geq \alpha$ whenever
\begin{align*}
\zeta_t\left(\frac{1}{\alpha} - \lambda_{\max}(\hat{S}) - \chi\right) &\geq 0\\
\frac{1}{\alpha} - \lambda_{\max}(\hat{S}) - \chi &\geq 0 \\
\alpha \leq \frac{1}{\|\hat{S}\|_2 + \chi}.
\end{align*}
\qed
\end{proof}

The upper bound in Theorem \ref{thm:eigiter} results from the following lemma. 
\begin{lemma}
Let $\chi$ be as in Lemma \ref{lem:7} and let $\alpha = \frac{1}{\|\hat{S}\|_2 + \chi}$. Let $\zeta_t \leq \alpha^2$ for all $t$. We then have $\Omega_t \preceq b' I$ for all $t$ when $b' = \|\Omega_\rho^*\|_2 + \| \Omega_0 - \Omega_\rho^* \|_F$.
\end{lemma}
\begin{proof}
By Lemma \ref{lem:7}, $\alpha I \preceq \Omega_t$ for every $t$. Since $\Omega_t \rightarrow \Omega_\rho^*$, by strong convexity $\alpha I \preceq \Omega_\rho^*$. Hence $a = \min \{\lambda_{\min}(\Omega_t), \lambda_{\min}(\Omega^*_\rho)\} \geq \alpha$. For $b > a$ and $\zeta_t \leq \alpha^2$, 
\[
\max \left\{\left|1 - \frac{\zeta_t}{b^2}\right|, \left| 1- \frac{\zeta_t}{a^2}\right| \right\} \leq 1.
\]
Hence, by Theorem \ref{thm:conv} $\|\Omega_t - \Omega^*_\rho \|_F \leq \|\Omega_{t-1} - \Omega^*_\rho  \|_F \leq \|\Omega_0 - \Omega^*_\rho \|_F$. Thus 
\[
\|\Omega_t\|_2 - \|\Omega_\rho^*\|_2 \leq \|\Omega_t - \Omega_\rho^*\|_2 \leq \|\Omega_t - \Omega_\rho^*\|_F \leq \|\Omega_0 - \Omega_\rho^*\|_F
\]
so
\[
\|\Omega_t\|_2 \leq \|\Omega_\rho^*\|_2 + \|\Omega_0 - \Omega_\rho^*\|_F.
\]
\qed
\end{proof}
This completes the proof of Theorem \ref{thm:eigiter}.
\end{proofof2}

\section{Useful Properties of the Kronecker Sum and $\mathcal{K}_{\mathbf{p}}$}
\label{App:Ident}
\subsection{Basic Properties}
\label{App:BasicProp}
As the properties of Kronecker sums are not always widely known, we have compiled a list of some fundamental algebraic relations we use.
\begin{enumerate}
\item Sum or difference of Kronecker sums \citep{laub2005matrix}:
\begin{align*}
c_A(A_1 \oplus &\dots \oplus A_K) + c_B(B_1 \oplus \dots \oplus B_K) \\&= (c_A A_1 + c_B B_1) \oplus \dots \oplus (c_A A_K + c_B B_K).
\end{align*}
\item \label{prop:disj} Factor-wise disjoint off diagonal support \citep{laub2005matrix}. By construction, if for any $k$ and $i \neq j$
\[
[I_{[d_{1:k-1}]} \otimes A_k \otimes I_{[d_{k+1:K}]}]_{ij} \neq 0,
\]
then for all $\ell \neq k$
\[
[I_{[d_{1:\ell-1}]} \otimes A_{\ell} \otimes I_{[d_{\ell+1:K}]}]_{ij} = 0.
\]
Thus,
\[
|A_1 \oplus \dots \oplus A_K|_{1,\mathrm{off}} = \sum_{k=1}^K |I_{[d_{1:k-1}]} \otimes \offd(A_k) \otimes I_{[d_{k+1:K}]}|_1 = \sum_{k=1}^K m_k |A_k|_{1,\mathrm{off}}.
\]

\item \label{prop:eig} Eigendecomposition: If $A_k = U_k \Lambda_k U_k^T$ are the eigendecompositions of the factors, then \citep{laub2005matrix}
\[
A_1 \oplus \dots \oplus A_K = (U_1 \otimes \dots \otimes U_K) (\Lambda_1 \oplus \dots \oplus \Lambda_K)(U_1 \otimes \dots \otimes U_K)^T
\]
is the eigendecomposition of $A_1 \oplus \dots \oplus A_K$. Some resulting identities useful for doing numerical calculations are as follows:
\begin{enumerate}
\item 
\label{prop:L2} 
L2 norm:
\begin{align*}
\|A_1 \oplus \dots \oplus A_K \|_2 &= \max \left(\sum_{k = 1}^K \max_i [\Lambda_k]_{ii}, -\sum_{k = 1}^K \min_i [\Lambda_k]_{ii} \right) \\&\leq \sum_{k = 1}^K \|A_k\|_2.
\end{align*}
\item Determinant:
\begin{align*}
 \log |A_1 \oplus \dots \oplus A_K| &= \log |\Lambda_1 \oplus \dots \oplus \Lambda_K| \\=& \underbrace{\sum_{i_1 = 1}^{d_1} \dots \sum_{i_K =1}^{d_K}}_{\mathrm{K \: sums}} \log\left(\underbrace{[\Lambda_1]_{i_1i_1} + \dots + [\Lambda_K]_{i_Ki_K}}_{\mathrm{K \: terms}}\right).
\end{align*}
\item Matrix powers (e.g. inverse, inverse square root): 
\[
(A_1 \oplus \dots \oplus A_K)^v = (U_1 \otimes \dots \otimes U_K) (\Lambda_1 \oplus \dots \oplus \Lambda_K)^v(U_1 \otimes \dots \otimes U_K)^T.
\]
Since the $\Lambda_k$ are diagonal, this calculation is memory and computation efficient.

\end{enumerate}
\end{enumerate}

\subsection{Eigenstructure of $\Omega \in \mathcal{K}_{\mathbf{p}}$}\label{supp:eigg}
Kronecker sum matrices $\Omega \in \mathcal{K}_{\mathbf{p}}$ have Kronecker product eigenvectors with linearly related eigenvalues, as contrasted to the multiplicatively related eigenvalues in the Kronecker product. For simplicity, we illustrate in the $K=2$ case, but the result generalizes to the full tensor case. Suppose that $\Psi_1 = U_1 \Lambda_1 U_1^T$ and $\Psi_2 = U_2 \Lambda_2 U_2^T$ are the eigendecompositions of $\Psi_1$ and $\Psi_2$.
Then by \cite{laub2005matrix}, if $\Omega = \Psi_1 \oplus\Psi_2$, the eigendecomposition of $\Omega$ is
\[
\Omega = \Psi_1 \oplus \Psi_2 = (U_1 \otimes U_2) (\Lambda_1 \oplus \Lambda_2) (U_1 \otimes U_2)^T.
\]
Thus, the eigenvectors of the Kronecker sum are the Kronecker products of the eigenvectors of each factor. This ``block" structure is evident in the inverse Kronecker sum example in Section 1 of the main text. The structure of $\Omega^{-1}$ is discussed further in \cite{canuto2014decay}. 

This eigenstructure representation parallels the eigenvector structure of the Kronecker product - 
specifically when $\Omega = \Psi_1 \otimes \Psi_2$
\[
\Omega = \Psi_1 \otimes \Psi_2 = (U_1 \otimes U_2) (\Lambda_1 \otimes \Lambda_2) (U_1 \otimes U_2)^T.
\]
Hence, use of the Kronecker sum model can be viewed as replacing the nonconvex, multiplicative eigenvalue structure of the Kronecker product with the convex linear eigenvalue structure of the Kronecker sum. This additive structure results in relatively more stable estimation of the precision matrix. As the tensor dimension $K$ increases, this structural stability of the Kronecker sum as compared to the Kronecker product becomes more dominant ($K$ term sums instead of $K$-order products).

%
%



\subsection{Projection onto $\mathcal{K}_{\mathbf{p}}$}\label{sec:prjjj}

We first introduce a submatrix notation. Fix a $k$, and choose $i,j \in \{ 1,\dots m_k\}$. Let $E_1 \in \mathbb{R}^{\prod_{\ell=1}^{k-1} d_k \times \prod_{\ell=1}^{k-1} d_k}$ and $E_2 \in \mathbb{R}^{\prod_{\ell = k+1}^K d_k \times \prod_{\ell = k+1}^K d_k}$ be such that $[E_1\otimes E_2]_{ij} = 1$ with all other elements zero. Observe that $E_1 \otimes E_2 \in \mathbb{R}^{m_k \times m_k}$.
For any matrix $A \in \mathbb{R}^{p \times p}$, let $A(i,j|k) \in \mathbb{R}^{d_k \times d_k}$ be the submatrix of $A$ defined via
\begin{align}\label{eq:submats}
[A(i,j|k)]_{rs} = \tr( (E_1 \otimes \mathbf{e}_{r} \mathbf{e}_{s} \otimes E_2)A), \qquad r,s = 1,\dots, d_k.
\end{align}
The submatrix $A(i,j|k)$ is defined for all $i,j \in \{ 1,\dots m_k\}$ and $k = 1,\dots, K$.
When $A$ is a covariance matrix associated with a tensor $X$, this subblock corresponds to the covariance matrix between the $i$th and $j$th slices of $X$ along the $k$th dimension.

We can now express the projection operator $\mathrm{Proj}_{\mathcal{K}_{\mathbf{p}}}(A)$ in closed form: 
\begin{figure}[h]
\centering
\includegraphics[width=4.5in]{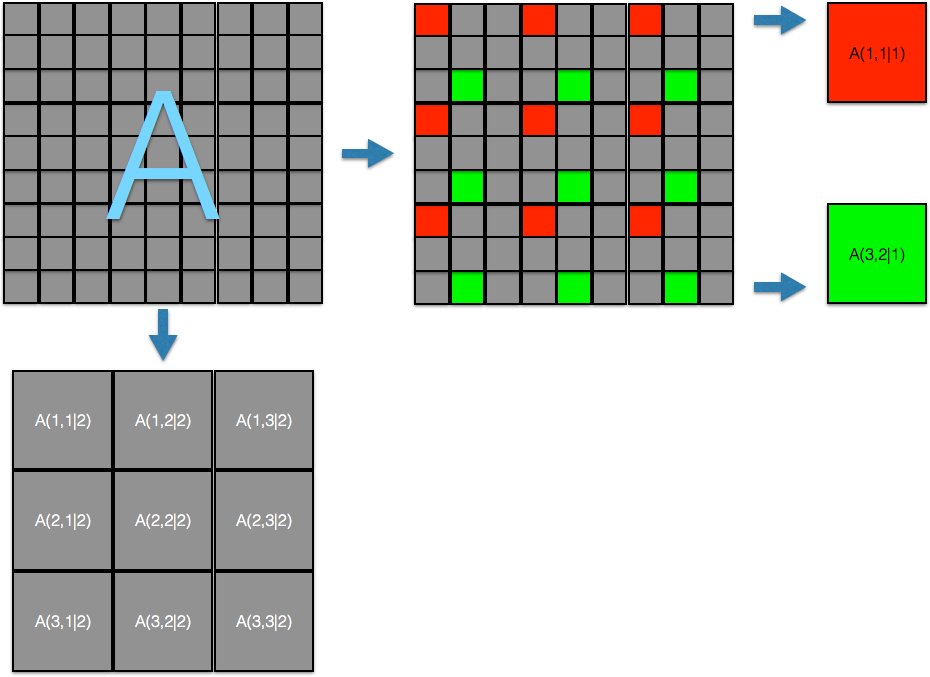}
\caption{Submatrix notation (equation \eqref{eq:submats}). Shown is a 9x9 matrix $A$, with $K=2$ and $d_1=d_2 = 3$. Displayed are the subblocks corresponding to the $A(i,j|2)$ and two example $A(i,j|1)$. $A(1,1|1)\in \mathbb{R}^{3\times 3}$ is formed from the 9 red entries, and $A(3,2|1)$ from the nine green entries. The remaining $A(i,j|1)$ follow similarly according to \eqref{eq:submats}. }
\label{fig:sub}
\end{figure}
\begin{lemma}[Projection onto $\mathcal{K}_{\mathbf{p}}$] \label{lem:Projj}
For any $A \in \mathbb{R}^{p \times p}$, 
\begin{align*}
\mathrm{Proj}_{\mathcal{K}_{\mathbf{p}}}(A) &= A_1 \oplus \dots \oplus A_K - (K-1)\frac{\tr(A)}{p} I_p\\
=& \left(A_1-\frac{K-1}{K}\frac{\tr(A_1)}{d_1} I_{d_1}\right) \oplus \dots \oplus \left(A_K-\frac{K-1}{K}\frac{\tr(A_K)}{d_K} I_{d_K}\right),
\end{align*}
where
\[
A_k = \frac{1}{m_k}\sum_{i= 1}^{m_k} A(i,i|k).
\]
Since the submatrix operator $A(i,i|k)$ is clearly linear, $\mathrm{Proj}_{\mathcal{K}_{\mathbf{p}}}(\cdot)$ is a linear operator. 
\end{lemma}

\begin{proof}
Since $\mathcal{K}_{\mathbf{p}}$ is a linear subspace, projection can be found via inner products. Specifically, recall that if a subspace $\mathcal{A}$ is spanned by an orthonormal basis $U$, then 
\[
\mathrm{Proj}_{\mathcal{A}}(\mathbf{x}) = UU^T \mathbf{x}.
\]
Since $\mathcal{K}_{\mathbf{p}}$ is the space of Kronecker sums, the off diagonal elements are independent and do not overlap across factors. 
The diagonal portion is more difficult as each factor overlaps on the same entries, creating an overdetermined system. We can create an alternate parameterization of $\mathcal{K}_{\mathbf{p}}$:
\begin{align}\label{eq:abar}
\mathrm{Proj}_{\mathcal{K}_{\mathbf{p}}}(A) = \bar{A}_1 \oplus \dots \oplus \bar{A}_K + \tau_A I_{p} = \tau_A I_{p} + \sum_{k = 1}^K I_{[d_{1:k-1}]} \otimes \bar{A}_k \otimes I_{[d_{k+1:K}]}
\end{align}
where we constrain $\tr(\bar{A}_k) = 0$. Each of the $K+1$ terms in this sum is now orthogonal to all other terms since by construction
\begin{align*}
\langle I_{[d_{1:k-1}]} \otimes \bar{A}_k \otimes & I_{[d_{k+1:K}]},I_{[d_{1:\ell-1}]} \otimes \bar{A}_\ell \otimes I_{[d_{\ell+1:K}]}\rangle\\ 
&= \frac{p}{d_k d_\ell} \tr((\bar{A}_k \otimes I_{d_\ell} )(I_{d_k} \otimes \bar{A}_\ell)) =  \frac{p}{d_k d_\ell}  \tr(\bar{A}_k) \tr(\bar{A}_\ell) = 0\\
\langle \tau_A I_{p}, I_{[d_{1:k-1}]} &\otimes \bar{A}_k \otimes I_{[d_{k+1:K}]}\rangle \\&= \langle \tau_A I_{[d_{1:k-1}]} \otimes I_{d_k} \otimes I_{[d_{k+1:K}]}, I_{[d_{1:k-1}]} \otimes \bar{A}_k \otimes I_{[d_{k+1:K}]}\rangle\\ &= m_k \langle I_{d_k}, \bar{A}_k \rangle = m_k \tr(\bar{A}_k)= 0
\end{align*}
for $\ell \neq k$ and all possible $\bar{A}_k$, $\tau_A$. 
Thus, we can form bases for the $\bar{A}_k$ and $\tau_A$ independently. To find the $\bar{A}_k$ it suffices to project $A$ onto a basis for $\bar{A}_k$. We can divide this projection into two steps. In the first step, we ignore the constraint on $\tr(\bar{A}_k)$ and create the orthonormal basis
\[
\mathbf{u}^{(ij)}_k := \frac{1}{\sqrt{m_k}}I_{[d_{1:k-1}]} \otimes \mathbf{e}_i \mathbf{e}_j^T \otimes I_{[d_{k+1:K}]}
\]
for all $i,j = 1,\dots d_k$. Recall that in a projection of $\mathbf{x}$, the coefficient of a basis component $\mathbf{u}$ is given by $\mathbf{u}^T\mathbf{x} = \langle \mathbf{u},\mathbf{x}\rangle$. We can thus apply this elementwise to the projection of $A$. Hence projecting $A$ onto these basis components yields a matrix $B\sqrt{m_k} \in \mathbb{R}^{d_k \times d_k}$ where
\[
B_{ij} = \frac{1}{m_k}\langle A, I_{[d_{1:k-1}]} \otimes \mathbf{e}_i \mathbf{e}_j^T \otimes I_{[d_{k+1:K}]}\rangle.
\]
To enforce the $\tr(\bar{A}_k) = 0$ constraint, we project away from $B$ the one-dimensional subspace spanned by $I_{d_k}$. This projection is given by
\begin{equation}\label{eq:bb}
B - \frac{\tr(B)}{d_k} I_{d_k}, 
\end{equation}
where by construction
\begin{align*}
\frac{\tr(B)}{d_k} &= \frac{1}{d_k m_k} \sum_{i = 1}^{d_k} \langle A, I_{[d_{1:k-1}]} \otimes \mathbf{e}_i \mathbf{e}_i^T \otimes I_{[d_{k+1:K}]}\rangle\\
&=\frac{1}{p} \langle A, I_p\rangle =\frac{\tr(A)}{p}.
\end{align*}
Equation \eqref{eq:bb} completes the projection onto a basis for $\bar{A}_k$, so we can expand the projection $\sqrt{m_k} B$ back into the original space. This yields a $\bar{A}_k$ of the form
\[
[\bar{A}_k]_{ij} = \left\{\begin{array}{ll} \frac{1}{m_k}\langle A, I_{[d_{1:k-1}]} \otimes \mathbf{e}_i \mathbf{e}_j^T \otimes I_{[d_{k+1:K}]}\rangle& i \neq j\\  \frac{1}{m_k}\langle A, I_{[d_{1:k-1}]} \otimes \mathbf{e}_i \mathbf{e}_i^T \otimes I_{[d_{k+1:K}]}\rangle - \frac{\mathrm{tr}(A)}{p}& i = j\end{array}\right.
\]
Finally, for $\tau_A$ we can compute
\[
\tau_A = \frac{1}{p} \langle A, I_p \rangle = \frac{\tr (A)}{p}.
\]


Combining all these together and substituting into \eqref{eq:abar} allows us to define the projection in terms of matrices $\tilde{A}_k$, where we split the $\tau_A I_p$ term evenly across the other $K$ factors. Specifically
\[
\mathrm{Proj}_{\mathcal{K}_{\mathbf{p}}}(A) = \tilde{A}_1 \oplus \dots \oplus \tilde{A}_K.
\]
where
\begin{equation}\label{eq:elementProj}
[\tilde{A}_k]_{ij} = \left\{\begin{array}{ll} \frac{1}{m_k}\langle A, I_{[d_{1:k-1}]} \otimes \mathbf{e}_i \mathbf{e}_j^T \otimes I_{[d_{k+1:K}]}\rangle& i \neq j\\  \frac{1}{m_k}\langle A, I_{[d_{1:k-1}]} \otimes \mathbf{e}_i \mathbf{e}_i^T \otimes I_{[d_{k+1:K}]}\rangle - \frac{K-1}{K}\frac{\mathrm{tr}(A)}{p}& i = j\end{array}\right. .
\end{equation}
An equivalent representation using factorwise averages is
\[
\tilde{A} = A_k - \frac{K-1}{K} \frac{\tr(A)}{p},
\]
where
\[
A_k = \frac{1}{m_k}\sum_{i = 1}^{m_k} A(i,i|k).
\]
Moving the trace corrections to a last term and putting the result in terms of the $A_k$ yields the lemma.

In Algorithm \ref{alg:lowlevel} we use an efficient method of computing this projected inverse in our setting by exploiting the eigendecomposition identities in Section \ref{supp:eigg}. 
\qed \end{proof}

\section{Known diagonal elements (correlation matrix form)}
\label{app:B}
In the case where the diagonal $\mathrm{diag}(\Omega_0)$ of the precision matrix is known a priori, the estimation problem becomes easier. For simplicity, we consider the case that $\Omega_0$ is in the form of a correlation matrix, i.e. $\mathrm{diag}(\Omega_0) = {I}_{p}$, noting this was the setting originally the focus of \cite{kalaitzis2013bigraphical}.

Note that since the diagonal elements are known, we do not need to estimate them and indeed can set all the $\mathrm{diag}(\Psi_k) = 1/K I_{d_k}$. Revisiting the proof of Theorem \ref{Thm:1}, it is easy to show the following corollary, which shows strong $O(\sqrt{(K+1)s \frac{\log p}{n \min_k m_k}})$ convergence in the case of $\ell 1$ regularization. This replacement of the $\sqrt{p + s}$ term in rate of Theorem \ref{Thm:1} with a $\sqrt{s}$ guarantees single sample convergence in the sparse setting when $\min_k m_k \gg s$.
\begin{corollary}
\label{Cor:KnownDiag}
Suppose the conditions of Theorem \ref{Thm:1}, and that $\mathrm{diag}(\Omega_0) = {I}_{p}$ is known. Then under event $\mathcal{A}$,
\begin{align*}
\|\hat{{\Omega}} - {\Omega}_0\|_F &\leq  \frac{2 C_1  \twonorm{\Sigma_0}}{\phi_{\min}^2(\Sigma_0)}
\sqrt{(K+1)s \frac{\log p}{n \min_k m_k}}.
\end{align*}
Furthermore, event $\mathcal{A}$ holds with probability at least $1-2(K+1) \exp(-c\log p)$.

\end{corollary}
\begin{proof}
Dropping the diagonal term from the proof of Lemma \ref{lemm:18}, we have that the $\sqrt{p}$ dependence vanishes, and on event $\mathcal{A}$, we have $G(\Delta) > 0$ for all $\Delta \in \mathcal{T}_n$ where
\[
\mathcal{T}_n = \left\{ \Delta_\Omega \in \mathcal{K}_{\mathbf{p}}: \Delta_\Omega = \Omega - \Omega_0, \|\Delta_\Omega\|_F = M r_{n,\mathbf{p}}\right\}
\]
and
\[
r_{n,\mathbf{p}} = C \|\Sigma_0\|_2 \sqrt{s(K+1) \frac{\log p}{n \min_k m_k}}.
\]

The rest of the proof follows by substituting this new value of $r_{n,\mathbf{p}}$ into the proof of Theorem \ref{Thm:1}.
\end{proof}

\section{SCAD and MCP regularizers}
\label{supp:regs}

The SCAD penalty \citep{fan2001variable} with parameter $a>2$ (giving $\mu = 1/(a-1)$) is given by
\begin{align}\label{eq:SCAD}
g_\rho(t) = \left\{\begin{array}{ll} \rho |t| & \mathrm{if}\:|t| \leq \rho\\
- \frac{t^2 - 2a \rho |t| + \rho^2}{2(a-1)} & \mathrm{if} \: \rho < |t| \leq a \rho\\
\frac{(a+1)\rho^2}{2} & \mathrm{if}\: a \rho < |t| \end{array}\right. 
\end{align}
which is linear (as the $\ell$1 norm) for small $|t|$, constant for large $|t|$, and has a transition between the two regimes for moderate $|t|$.

The MCP penalty \citep{zhang2010nearly} with parameter $a > 0$ (giving $\mu = 1/a$) is given by
\begin{align}\label{eq:MCP}
g_\rho(t) = \mathrm{sign}(t) \rho \int_{0}^{|t|} \left(1 - \frac{z}{\rho a}\right)_+ dz,
\end{align}
giving a more smooth transition between the approximately linear region and the constant region ($t > \rho a$).

\end{appendices}

\end{document}